\documentclass[11pt]{article}

\usepackage{amsfonts}
\usepackage{amsmath}
\usepackage{amsthm} 
\usepackage{amssymb}
\usepackage{epsfig, graphics}
\usepackage{hyperref}
\usepackage{url}
\usepackage{color}
\usepackage{setspace}
\usepackage{float}
\usepackage{subfig}
\usepackage{colortbl}

\usepackage[square, numbers, comma, sort&compress]{natbib}
\usepackage{graphicx}
\usepackage{tikz}

\hypersetup{colorlinks,linkcolor={blue},citecolor={red},urlcolor={blue}}

\def\EE{\mathbb{E}}
\def\PP{\mathbb{P}}
\def\NN{\mathbb{N}}
\def\RR{\mathbb{R}}

\def\ind{{\rm 1\hspace{-0.90ex}1}}
\def\Var{\mathrm{Var}}
\def\Cov{\mathrm{Cov}}

\def\cGn{\cG^{(n)}(\bd_n^+,\bd_n^-)}

\def\bd{{\mathbf d}}
\def\bX{{\mathbf X}}
\def\bY{{\mathbf Y}}
\def\bZ{{\mathbf Z}}
\def\bM{{\mathbf M}}

\def\bS{\mathbf{S}}
\def\bD{\mathbf{D}}
\def\bH{\mathbf{H}}
\def\bI{\mathbf{I}}
\def\bHp{\mathbf{H^+}}
\def\bHn{\mathbf{H^-}}
\def\bIp{\mathbf{I^+}}
\def\bIn{\mathbf{I^-}}

\def\hf{\widehat{f}}

\def\beps{\boldsymbol\epsilon}
\def\balpha{\boldsymbol\alpha}

\newcommand{\Bin}{\mathsf{Bin}}

\def\tod{\stackrel{d}{\longrightarrow}}
\def\top{\stackrel{p}{\longrightarrow}}

\newcommand{\cV}{\mathcal{V}}
\newcommand{\cD}{\mathcal{D}}

\newcommand{\cA}{\mathcal{A}}
\newcommand{\cR}{\mathcal{R}}
\newcommand{\cX}{{\mathcal{X}}}
\newcommand{\m}{\mathcal{M}}
\newcommand{\cN}{{\mathcal{N}}}
\newcommand{\cZ}{{\mathcal{Z}}}
\newcommand{\cH}{{\mathcal{H}}}
\newcommand{\cG}{{\mathcal{G}}}
\newcommand{\interval}{\left[0,\infty\right)}
\newcommand{\pip}{\pi_x(\theta)}
\newcommand{\zst}{z^{\star}}
\newcommand{\tst}{\tau_n^{\star}}

\newcommand{\hp}{\widehat{z}}
\newcommand{\hz}{z^{\star}}

\newcommand{\cE}{\mathcal{E}}

\newcommand{\bigsigma}{\makebox{\large\ensuremath{\sigma}}}

\def\tod{\stackrel{d}{\longrightarrow}}
\def\top{\stackrel{p}{\longrightarrow}}

\newtheorem*{pr}{Proof}
\newtheorem{theorem}{Theorem}[section]
\newtheorem{remark}[theorem]{Remark}

\newtheorem{assumption}{Assumption}

\newtheorem{proposition}[theorem]{Proposition}
\newtheorem{lemma}[theorem]{Lemma}

\newtheorem{intassumption}{Assumption}
\numberwithin{intassumption}{assumption}

\author{Hamed Amini\thanks{Georgia State University, Department of Risk Management and Insurance, Atlanta, GA 30303, USA, email: {\tt hamini@gsu.edu}.} \and Zhongyuan Cao \thanks{INRIA Paris,  2 rue Simone Iff, CS 42112, 75589 Paris Cedex 12, France, and Universit\'e Paris-Dauphine, email: {\tt zhongyuan.cao@inria.fr}.}  \and Agn\`es  Sulem \thanks{INRIA Paris,  2 rue Simone Iff, CS 42112, 75589 Paris Cedex 12, France, and Universit\'e Paris-Est, email: {\tt agnes.sulem@inria.fr}.}}

\date{March 2021}

\begin{document}


\title{Limit Theorems for Default Contagion and Systemic Risk}

\maketitle

\abstract{We consider a general tractable model for default contagion and systemic risk in a heterogeneous financial network, subject to an exogenous macroeconomic shock. We show that, under some regularity assumptions, the default cascade model could be transferred to a death process problem represented by balls-and-bins model. We also reduce the dimension of the problem by classifying banks according to different types, in an appropriate type space. These types may be calibrated to real-world data by using machine learning techniques. We then state various limit theorems regarding the final size of default cascade over different types. In particular, under suitable assumptions on the degree and threshold distributions, we show that the final size of default cascade has asymptotically Gaussian fluctuations. We next state limit theorems for different system-wide wealth aggregation functions and show how the systemic risk measure, in a given stress test scenario, could be related to the structure and heterogeneity of financial networks. We finally show how these results could be used by a social planner  to optimally target interventions during a financial crisis, with a budget constraint and under partial information of the financial network.}

\bigskip


\newpage
\tableofcontents

\newpage
\section{Introduction}

The financial crisis 2007-2009 has illustrated the significance of network structure on the amplification of initial shocks in the banking system to the level of the global financial system, leading to an economic recession. 

This paper studies structural and dynamic models for loss propagation in the network of liabilities. This is in contrast to the well-known systemic risk indicators such as CoVAR~\cite{tobias2016covar} or SES~\cite{acharya2017measuring}, which are based on measuring losses in terms of market equity.

Empirical studies on network topology of banking systems show that we may have very different structures; from centralized networks as in~\cite{muller2006interbank} to core-periphery structures~\cite{craig2014interbank, fricke2015core, li2019dealer} and scale-free structures as in~\cite{boss04, cont2010network}. The main focus of this paper is to use limit theorems to establish a link between the (final) size of default cascade to the structure and heterogeneity of financial networks. The paper also studies limit theorems to quantify the system-wide wealth and systemic risk in the financial system. We also show how these results could be used by a social planner  to optimally target interventions during a financial crisis, under partial information of the financial network and under some budget constraint.

Dealing with missing data is another important aspect of this  paper. While recent efforts have made available a tremendous amount of financial data, we are far from a global financial data collection firm-by-firm and there are systemic institutions which are not required to report their positions. As pointed out in~\cite{anand2015filling, upper2004estimating, iori2015networked, gandy2017bayesian, roukny2018interconnectedness}, in general, there is only partial information available on the financial network, e.g., the total size of the assets and liabilities for each institution. Our probabilistic approaches allow us to deal with the incomplete observation of system connections. We also reduce the dimension of the problem by considering a classification of financial institutions according to different types (characteristics), in an appropriate type space $\cX$. These characteristics may be calibrated to real-world data by using machine learning techniques for classification.


There is an emerging literature on systemic risk and financial networks, see e.g.~\cite{jackson2020systemic, chinazzi2015systemic} for two recent surveys and references there. An extensive research in this area focuses on equilibrium approach, to derive recovery rates from an elegant fixed point equation~\cite{eisenberg2001systemic, rogers2013failure, elliott2014financial, glasserman2015likely}. This relies on the assumption that all debts are instantaneously cleared, unlikely to hold in reality. Even in a given shock scenario, recovery rates are uncertain. For example, recovery rates after the failure of Lehman were around $8\%$ (\cite{morris2009illiquidity}). In this paper we model recovery rates as given. This could be easily extended to random recovery rates satisfying some cash-flow consistency conditions, see e.g.~\cite{amini2016inhomogeneous}.

Our work is related to the literature on network structure and contagion, see e.g.,~\cite{morris2000contagion, akbarpour2018diffusion, watts2002simple, lelarge12b, kleinberg2007cascading, young2009innovation} in the context of social networks and \cite{acemoglu2015systemic, detering2019managing, nier2007network, allen2000financial, erol2018network, cabrales2017risk, leitner2005financial, gai2010contagion, amini2019dynamic} for default contagion in financial networks. In particular, more closely related to our work, \cite{amcomi-res} study default contagion on configuration model and derive a criterion for the resilience of a financial network to insolvency contagion, based on connectivity and the structure of contagious links (i.e., those exposures of a bank larger than its capital).

The current paper significantly pushes the technical knowledge in the literature on financial networks by providing central limit theorems for default contagion and systemic risk in a heterogeneous financial network. We propose multiple extensions of~\cite{amcomi-res}:
\begin{itemize}
\item We generalize the default contagion model and allow for more network heterogeneity by considering the type-dependent threshold model. These types may be calibrated to real-world data by using machine learning techniques for classification.
 
\item We transfer the default cascade model to a death process problem represented by balls-and-bins model. This allows us to provide a simpler proof and state the limit theorems for the dynamic contagion model~\footnote{The balls-and-bins model has been previously used in the economic literature; see e.g.~\cite{armenter2014balls} for a balls-and-bins  model of international trade.}.
\item We provide central limit theorems and show that the final size of default cascade has asymptotically Gaussian fluctuations. We state various theorems regarding the joint asymptotic normality between different contagion parameters, including the number of solvent banks, defaulted banks, healthy links (those initiated by solvent banks) and infected links (those initiated by defaulted banks) at any time $t$. 
\item Our next set of results concern the limit theorems for  system-wide wealth aggregation functions, which could be used for measuring and quantifying systemic risk. This also provide an indicator for the health of financial system in different stress scenarios. 
\item Finally, we consider a social planner  who seeks to optimally target interventions during a financial crisis, under partial information of the financial network and with a budget constraint. We show how limit theorems allows us to simplify the optimization problem. The complete information setup has been recently studied in \cite{galeotti2020targeting, jackson2020credit}.
\end{itemize}

Aside from the application to default contagion and systemic risk in financial networks, our results contribute to the literature on diffusion processes on random graphs. Related problems are the $k$-core and bootstrap percolation. The $k$-core of any finite graph can be found by removing nodes with degree less than $k$, in any order, until no such nodes exist. The asymptotic normality of $k$-core has been studied in~\cite{janson2008asymptotic}. We also study the default contagion in a similar way as in~\cite{janson2008asymptotic}, by transferring the process to a death process problem represented by balls-and-bins model. We generalize their results by allowing different types and threshold levels to each of nodes. The bootstrap percolation is a diffusion process that has been studied on a variety of graphs, see e.g., \cite{ar:JLTV10, amini10, ar:AmFount2012}. In bootstrap percolation process, for a fixed threshold $\theta\geq 2$, there is an initially subset of active nodes and in each round, each inactive node that has at least $\theta$ active neighbors becomes active and remains so forever. The asymptotic normality of bootstrap percolation has been recently studied in~\cite{AmBaCh21}. Our results generalize those of previous studies on bootstrap percolation in random graphs to the case of heterogeneous random directed networks with type-dependent random thresholds.

We end this introduction by the following remark. In the real world application, using limit theorems requires  some caution. For example, in order for the asymptotic analysis to be relevant, the financial network should be sufficiently large. This could be true for example at the level of a large economic zone. Moreover, financial networks may have small cycles. Most existing literature on random networks future locally tree-like property. However, recent literature shows that the basic configuration model can be extended to incorporate clustering; see e.g., \cite{van2015hierarchical, coupechoux2014clustering}. Moreover, following the recent literature on portfolio compression in financial networks (see e.g.,~\cite{amini2020optimal, veraart2020does, d2019compressing}), the study of default contagion and systemic risk in sparse financial networks regime becomes significantly important, as portfolio compression removes small cycles. In light of its tractability and interpretability, as well as its potential to be enriched with clustering, in this paper we use the configuration model as our base model. Note that the closed form interpretable limit theorems that we provide could also serve as a mandate for regulators to collect data on those specific network characteristics  and assess systemic risk via more intensive computational methods.  


\paragraph{Outline.} The paper is organized as follows. Section~\ref{sec:model} introduces a model for the network of financial counterparties and describe a mechanism for default cascade in such a network, after an exogenous macroeconomic shock. We also describe how the default contagion model could be transferred to a death process problem represented by balls-and-bins model. Section~\ref{sec:main} gives our main results on limit theorems for the final size of default cascade. In particular, under some regularity assumptions, we show that different default contagion parameters have asymptotically Gaussian fluctuations. Section~\ref{sec:sys} states limit theorems for different financial system aggregation functions, which are used for measuring and quantifying systemic risk. Section~\ref{sec:intervene} shows how these limit theorems could be used by a social planner to optimally target interventions during a financial crisis, with a budget constraint and under partial information of the financial network. All proofs are given in Appendix~\ref{sec:proofs}. Further extensions  are discussed in Appendix~\ref{app:Ext}.

\paragraph{Notation.} 
Let $\{ X_n \}_{n \in \mathbb{N}}$ be a sequence of real-valued random variables on a probability space
$ (\Omega, \mathcal{F}, \mathbb{P})$. If $c \in \mathbb{R}$ is a constant, we write $X_n \stackrel{p}{\longrightarrow} c$ to denote that $X_n$ converges in probability to $c$. That is, for any $\epsilon >0$, we have $\mathbb{P} (|X_n - c|>\epsilon) \rightarrow 0$ as $n \rightarrow \infty$. We write $X_n \stackrel{d}{\longrightarrow} X$ to denote that $X_n$ converges in distribution to $X$.
Let $\{ a_n \}_{n \in \mathbb{N}}$ be a sequence of real numbers that tends to infinity as $n \rightarrow \infty$.
We write $X_n = o_p (a_n)$, if $|X_n|/a_n \top 0$.
If $\mathcal{E}_n$ is a measurable subset of $\Omega$, for any $n \in \mathbb{N}$, we say that the sequence
$\{ \mathcal{E}_n \}_{n \in \mathbb{N}}$ occurs with high probability (w.h.p.) or almost surely (a.s.) if $\mathbb{P} (\mathcal{E}_n) = 1-o(1)$, as
$n\rightarrow \infty$.
Also, we denote by 
$\Bin (k,p)$ a binomial  distribution  corresponding to the number of
successes of a sequence of $k$ independent Bernoulli trials each having probability of success  $p$. The notation $\ind{\{\cE\}}$ is used for the indicator of an event $\cE$; this is 1 if $\cE$ holds and 0 otherwise. We deonte by $\cD[0,\infty)$ the standard space of right-continuous functions with left limits on $[0,\infty)$ equipped with the Skorohod topology (see e.g.~\cite{jacod2013limit, kallenberg1997foundations}). 
We will suppress the dependence of parameters on the size of the network $n$, if it is clear from the context.

\section{Model}\label{sec:model}
\subsection{Financial Network}
Consider an economy $\cE_n$ consisting of $n$ interlinked financial institutions (banks) denoted by $[n]:= \{1,2,\dots,n\}$ that intermediate credit among end-users. Banks hold claims on each other. Interbank liabilities are represented by a matrix of nominal liabilities $(\ell_{ij})$. 
For two financial institutions $i,j \in [n]$,  $\ell_{ij}\geq 0$ denotes the cash-amount that bank $i$ owes bank $j$. This also represents the maximum loss related to direct claims, incurred by bank $j$ upon the default of bank $i$. The total nominal liabilities of bank $i$ sum up to $\ell_i=\sum_{j\in[n]} \ell_{ij}$, while the total value of interbank assets sum up to $a_i = \sum_{j \in [n]} \ell_{ji}$. The total value of claims held by end-users on bank $i$ (deposits) is given by $d_i$. The total value of claims held by bank $i$ on end-users (external assets) is denoted by $e_i$.

In a stress testing framework, we apply a (fractional) shock $\epsilon_i$ to the external assets of bank $i$. Table~\ref{tab:aftershock} summarizes a stylized balance sheet of bank $i$ after the shock $\epsilon_i$. The capital of bank $i$ after the shock denoted by $c_i=c_i(\epsilon_i)$ satisfies 
\begin{equation*}
c_i=(1-\epsilon_i) e_i+a_i-\ell_i-d_i,
\end{equation*}
which represents the capacity of bank $i$ to absorb losses while remaining solvent.
\begin{table}[h]
  \centering
   \begin{tabular}{|c|c|}
    \hline
       External assets   & Deposits\\
  $e_{i}$   &  $d_i$ \\
   \cline{2-2}
  \cellcolor[gray]{.8} $\epsilon_ie_{i}$ -  loss on assets &  Interbank liabilities \\
    \cline{1-1} & $ \ell_i = \sum_{j \in [n]} \ell_{ij}$\\
      Interbank assets   &  \\
   \cline{2-2}
     $a_i = \sum_{j \in [n]} \ell_{ji}$ &   Capital\\
    
      & $c_{i}$ \\
      &\cellcolor[gray]{.8} $\epsilon_ie_{i}$ - loss on capital\\
         \hline
    Assets & Liabilities\\
    \hline
  \end{tabular}
\caption{Stylized balance sheet of bank $i$ after shock.}
\label{tab:aftershock}
\end{table}
  
A financial institution $i\in[n]$ is said to be {\it fundamentally insolvent} if its capital after the shock is negative, i.e. $c_i<0$. For a given shock scenario $\beps=(\epsilon_1, \dots, \epsilon_n)\in [0,1]^{n}$, we let the set of fundamental defaults
\begin{equation}
\cD_0(\beps) = \{i\in [n]: c_i(\epsilon_i)<0\}.
\end{equation}
  
In the next section, we define the default cascade triggered by fundamentally insolvent institutions $\cD_0(\beps)$.

\subsection{Default Cascade}
In the given shock scenario $\beps$, following the  fundamentally insolvent institutions $\cD_0(\beps)$, there will be a default contagion process. Let us denote by $R_{ij}=R_{ij}(\beps)$ the recovery rate of the liability of $i$ to $j$. The matrix of recovery rates is denoted by $\cR=(R_{ij})$.  Since any bank $i$ cannot pay more than its external assets $(1-\epsilon_i)e_i$ plus what it recovered from its debtors, the recovery rates of $i$ should satisfy the following cash-flow consistency constraints 
\begin{equation*}
(1-\epsilon_i)e_i + \sum_{j=1}^{n} R_{ji}\ell_{ji} \geq \sum_{j=1}^{n} R_{ij} \ell_{ij} +d_i.
\end{equation*}

Given the shock scenario $\beps$ and the matrix of recovery rates $\cR$, following the set of fundamental default $\cD_0$, there is a default cascade that reaches the set $\cD^\star$ in equilibrium. This represents the set of financial institutions whose capital is insufficient to absorb losses and should satisfy the following fixed point equation:

\begin{equation*}
\cD^\star=\cD^\star(\beps,\cR) = \Bigl\{i\in [n]: c_i < \sum_{j\in \cD^\star}(1-R_{ji}) \ell_{ji} \Bigr\}.
\end{equation*}

As shown in~\cite{amini2016inhomogeneous}, the above fixed point default cascade set has in general multiple solutions. The smallest fixed point set which corresponds to smallest number of defaults can be obtained by starting from $\cD_0$ and setting at step $k$:
\begin{equation}
\cD_k = \cD_{k}(\beps,\cR)=\bigl\{i\in [n]: c_i < \sum_{j\in \cD_{k-1}} (1-R_{ji}) \ell_{ji} \bigr\}.
\end{equation} 

The cascade ends at the first time $k$ such that $\cD_k=\cD_{k-1}$. Hence in a financial network of size $n$, the cascade will end after at most $n-1$ steps and $\cD_{n-1}=\cD_{n-1}(\beps,\cR)$ represents the final set of insolvent institutions starting from the initial set of defaults $\cD_0$.

\subsection{Node Classification and Configuration Model}
Under some regularity assumptions detailed below, as also shown in~\cite{amcomi-res} for a similar setup, one can show that the information regarding assets, liabilities, capital after exogenous shocks and recovery rates (distributions) could all be encoded in a single probability threshold function. Namely, for a given shock scenario $\beps$ and the matrix of recovery rates $\cR$\footnote{Our results can easily be extended to a framework with independent random recovery rates; see e.g.~\cite{amini2016inhomogeneous}.}, we introduce the (random) threshold $\Theta_i=\Theta_i^{(n)}$ for any institution $i\in[n]$ which measures how many defaults $i$ can tolerate before becoming insolvent, if its counterparties default in a uniformly at random order, i.e., when $i$'s debtors default order environment is chosen uniformly at random among all possible permutations. 

In the following, in order to reduce the dimensionality of the problem, we consider a classification of financial institutions into a countable (finite or infinite) possible set of characteristics $\cX$. All (observable) characteristics for financial institution $i$ is encoded in $x_i=(d_{i}^{+},d_{i}^-, t_i, ...) \in \cX,$ where  $d_{i}^+$ denotes the in-degree (number of institutions $i$ is exposed to), $d_{i}^-$ denotes the out-degree (number of institutions exposed to $i$)  and $t_i$ denotes any other institution's type specific (e.g., credit rating, seniority class, etc.).

As we are interested in limit theorems, we consider a sequence of economies $\{\cE_n\}_{n\in \NN}$, indexed by the number of institutions. In particular, in the economy $\cE_n$, the characteristic of any institution $i\in[n]$ will be denoted by $$x_{i}^{(n)} = (d_{i}^{+(n)},d_{i}^{-(n)}, t_{i}^{(n)}, ...) \in \cX.$$

Note that, without loss of generality, the institutions in the same class $x\in\cX$ assumed to have the same number of creditors (denoted by $d_x^-$) and the same number of debtors (denoted by $d_x^+$). Further, for tractability, we make the following assumption on the probability threshold functions.


\begin{assumption}\label{cond-threshold}
We assume that there exists a classification of the financial institutions into a countable set of possible characteristics $\cX$ such that, for each $n\in \NN$, the institutions in the same characteristic class have the same threshold distribution function (denoted by $q_x^{(n)}$ for institutions in class $x\in \cX$). Namely, for economy $\cE_n, i\in [n]$ and for all $\theta \in \NN$:  
 \begin{align*}
 \PP(\Theta_i^{(n)}=\theta)=q_{x_i^{(n)}}^{(n)}(\theta).
 \end{align*}
\end{assumption}

In particular, in the network of size $n$, $q_x^{(n)}(0)$ represents the proportion of initially insolvent institutions with type $x\in \cX$. 

\medskip

Let $\mu_x^{(n)}$ denotes the fraction of institutions with characteristic $x\in\cX$ in the economy $\cE_n$. In order to study the asymptotics, it is natural to assume the following.

\begin{assumption}\label{cond-limit}
We assume that for some probability distribution functions $\mu$ and $q$ over the set of characteristics $\cX$ and independent of $n$, we have $\mu_x^{(n)} \to \mu_x$ and $q_x^{(n)} (\theta) \to q_x(\theta)$ as $n\to \infty$, for all $x\in \cX$ and $\theta=0,1, \dots, d_x^+$.
\end{assumption}

In the following, we will suppress the dependence of parameters on the size of the network $n$, if it is clear from the context. 

\medskip

Given the degree sequences $\bd^+_n=(d_1^+, \dots, d_n^+)$ and $\bd^-_n=(d_1^-, \dots, d_n^-)$ such that $\sum_{i\in [n]} d_i^+ = \sum_{i\in[n]} d_i^-$,  we associate to each institution $i$ two sets: $\cH^+_i$ the set of incoming half-edges and $\cH^-_i$ the set of outgoing half-edges, with $|\cH_i^+|=d^+_i$ and $|\cH^-_i|=d^-_i$. Let $\mathbb{H}^+=\bigcup_{i=1}^{(n)} \cH^+_i$ and $\mathbb{H}^-=\bigcup_{i=1}^{(n)} \cH^-_i$. A {\it configuration} is a matching of $\mathbb{H}^+$ with $\mathbb{H}^-$. When an out-going half-edge of instituion $i$ is matched with an in-coming half-edge of institution $j$, a directed edge from $i$ to $j$ appears in the graph. The configuration model is the random directed multigraph which is uniformly distributed across all configurations. The random graph constructed by configuration model will be denoted by $\cGn$. 
It is then easy to show that conditional on the multigraph being a simple graph, we obtain a uniformly distributed random graph with these given degree sequences denoted by $\cG_*^{(n)}(\bd_n^+,\bd_n^-)$. In particular, any property that holds with high
probability on the configuration model also holds with high probability conditional on this random graph being simple (for the random graph $\cG_*^{(n)}(\bd_n^+,\bd_n^-)$) provided
$\liminf_{n\rightarrow\infty}\mathbb{P}(\cGn \  \mbox{simple})>0$, see e.g.~\cite{hofstad16}. 

\subsection{Death Process and Final Solvent Institutions}\label{sec:deathFin}
We consider the default contagion process in the random financial network $\cGn$, initiated by the set of fundamentally insolvent institutions $\cD_0$.

Recall that $ \Theta_i$ denotes the random threshold of institution $i\in[n]$  which measures how many defaults $i$ can tolerate before becoming insolvent in the uniformly chosen $i$'s counterparties default order environment. 
By Assumption~\ref{cond-threshold} and standard coupling arguments, as also proved in~\cite{amcomi-res}, one can assume that these thresholds are assigned initially to any institution $i\in [n]$ according to the distribution  $q_{x_i}^{(n)}(.)$

\paragraph{Finding the final solvent institutions.} We consider the above default contagion progress in the following way. At time 0 in the (random) graph $\cGn$, all institutions with threshold 0 become defaulted. We remove all the initially defaulted institutions $\cD_0$ from the network. Next, in order to find $\cD_1$, we identify the partners of $\cD_0$. Note that the out-degree and in-degree of each institution in the network induced by $[n]\setminus \cD_0$ is less than or equal to those in the previous network. At step $k\in \NN$, the default set $\cD_{k}$ can be identified by
\begin{equation}
\cD_{k}=\Bigl\{i\in [n]: \sum_{j: j\to i} \ind{\{j\in \cD_{k-1}\}} \geq \Theta_i\Bigr\},
\end{equation}
where $\ind{\{\cE\}}$  denotes the indicator of an event $\cE$, i.e., this is 1 if $\cE$ holds and 0 otherwise. We denote the in-degree and out-degree of each institution $i$ after $k$ steps evolution by $d_i^+(k)$ and $d_i^-(k)$ respectively. Note that initially $d_i^+(0)=d_i^+$ and $d_i^-(0)=d_i^-$. At step $k$, we remove all institutions $i\in[n]$ with $d_i^+(k)<d_i^+-\Theta_i$. At the end of the above procedure, all the removed institutions are defaulted and the remaining institutions are solvent.

\paragraph{Transferring to a death process problem represented by balls-and-bins.}
It is not hard to see that the calendar time does not take any important role in the above contagion process. We can define the time interval as we want. So instead of removing institutions, we can also remove the links and define a proper time interval between two successive removals. Namely, at each step, we only look at one removal (interaction) between two institutions, yielding at least one default. 

In the following, we simultaneously run the default contagion process and construct the configuration model. We call all out half-edges and in half-edges that belong to a defaulted institution the {\it infected} half-edges, all out half-edges and in half-edges that belong to a solvent institution the {\it healthy} half-edges. We consider all the institutions as bins and all the (in and out) half-edges as (in and out) balls. Consequently, the bins are called defaulted ($\bD$ type) or solvent ($\bS$ type) according to their states as institutions. Similarly the balls are called infected ($\bI$ type) or healthy ($\bH$ type) when they are infected or healthy as half-edges. Hence, all balls are of four different types. For convenience, we denote them as $\bHp$ (healthy in), $\bHn$ (healthy out), $\bIp$ (infected in) and $\bIn$ (infected out) balls, respectively.

We start from the set of fundamental defaults $\cD_0$, which gives the set of initially defaulted bins and infected balls. Consequently, at each step, we first remove a uniformly chosen ball of type $\bIn$ and then a uniformly chosen ball from $\bHp \cup \bIp$. In this process $\bS$ bins may change to $\bD$ bins and, consequently, $\bH$ balls may change to $\bI$ balls. We continue the above process until there is no more $\bIn$ balls.

We now change the description a little by introducing colors for the $\bIn$ balls and life for all in balls $\bHp \cup \bIp$. We assume that all $\bIn$ balls are initially white and all in balls $\bHp\cup \bIp$ are initially alive. We begin by recoloring one random  $\bIn$ ball red. Subsequently, in each removal step, we first kill a random in ball from $\bHp\cup \bIp$ and then recolor a random white ball from $\bIn$ red. This is repeated until no more white $\bIn$ balls remain.

 We next run the above death process in continuous time. We assume that each in ball $\bHp\cup \bIp$ has an exponentially distributed random lifetime with mean one, independent of all other balls. Namely, if there are $\ell$ alive in balls remaining, then we wait an exponential time with mean $1/\ell$ until the next pair of interactions. We stop when we should recolor a white $\bIn$ ball but there is no such ball. 
 
 Let us denote by $W_n(t)$ the number of white $\bIn$ balls at time $t$. Hence, the above death process ends at the stopping time  $\tau^\star_n$ which is the first time when we need to recolor a white ball but there are no white balls left. However, we pretend that we recolor a (nonexistent) white ball at time $\tau^\star_n$ and write $W_n(\tau^\star_n)=-1$. 
 
 We denote by $I_n^+(t)$ and $H_n^+(t)$ the number of alive (in) balls in $\bIp$ and $\bHp$ at time $t$, respectively. For $x\in \cX, \theta\in \NN, \ell=0,\dots, \theta-1$, we let $S^{(n)}_{x,\theta,\ell}(t)$ denotes the number of solvent institutions (bins) with type $x$, initial threshold $\theta$ and $\ell$ defaulted neighbors at time $t$. Further, let $S
 _n(t)$ and $D_n(t)$ be the numbers of $\bS$ bins and $\bD$ bins at time $t$. Hence, $S_n(\tau^\star_n)$ denotes the final number of solvent institutions. Further, $D_n(\tau^\star_n)=n-S_n(\tau^\star_n)=|\cD_{n-1}|$ will be the final number of defaulted institutions.


\section{Limit Theorems}\label{sec:main}

In this section we consider the above dynamic default contagion model (which is now transferred to a death process problem represented by balls-and-bins) and state our main results regarding the limit theorems in the random financial network $\cGn$.

\medskip

We first define some functions that will be used later. Let for $z\in[0,1]$:
\begin{align}
b(d,z,\ell):=& \PP(\Bin(d,z)= \ell) = \binom{d}{\ell} z^\ell (1-z)^{d-\ell}, \\
\beta(d,z,\ell):=& \PP(\Bin(d,z)\geq \ell) = \sum_{r=\ell}^d \binom{d}{r} z^r (1-z)^{d-r},
\end{align}
and $\Bin(d,z)$ denotes the binomial distribution with parameters $d$ and $z$.

\subsection{Asymptotic Magnitude of Default Contagion}\label{sec:LLN}

In this section we consider the random graph $\cGn$ and assume that the average degrees converges to a finite limit.
\stepcounter{assumption}
\begin{intassumption}
\label{cond-average}
We assume that, as $n\to \infty$, the average degrees converges and is finite:
$$\lambda^{(n)}:=\sum_{x\in \cX} d_x^+ \mu_x^{(n)} = \sum_{x\in \cX} d_x^- \mu_x^{(n)} \longrightarrow \lambda:= \sum_{x\in \cX} d_x^+ \mu_x \in (0,\infty).$$
\end{intassumption}

\medskip

For $z \in[0,1]$, we define the functions:
\begin{align*}
f_S(z):=& \sum_{x\in \cX} \mu_x \sum_{\theta=1}^{d_x^+}q_x(\theta) \beta\bigl(d_x^+, z, d_x^+-\theta+1\bigr), \ \ f_D(z)=1-f_S(z),\\
f_{H^+}(z):=& \sum_{x\in \cX} \mu_x \sum_{\theta=1}^{d_x^+}q_x(\theta) \sum_{\ell=d_x^+-\theta+1}^{d_x^+} \ell b\bigl(d_x^+,z,\ell\bigr) , \ \ f_{I^+}(z) = \lambda z-f_{H^+}(z),\\
f_{W}(z):=& \lambda z - \sum_{x\in \cX} \mu_x d_x^-\sum_{\theta=1}^{d_x^+}q_x(\theta) \beta\bigl(d_x^+,z, d_x^+-\theta+1\bigr).
\end{align*}

The following theorem  states the law of large numbers for the number of solvent banks, defaulted banks, healthy links, infected links and the total number of existing white balls (remaining interactions yielding at least one default) at any time $t$ in the economy $\cE_n$ satisfying above regularity assumptions.

\begin{theorem}\label{centrality}
Let $\tau_n\leq \tau^\star_n$ be a stopping time such that $\tau_n\top t_0$ for some $t_0>0$. Then for all $x\in\cX, \theta=1, \dots, d_x^+$ and $\ell=0, \dots, \theta-1$, we have  (as $n\rightarrow\infty$)
\begin{align*}
\sup\limits_{t\leq \tau_n}\bigl|\frac{S_{x,\theta,\ell}^{(n)}(t)}{n}-\mu_x q_x(\theta)b\left(d_x^+,1-e^{-t},\ell\right)\bigr|\top 0.
\end{align*}
Further, as $n\to \infty$,
\begin{align*}
\sup\limits_{t\leq \tau_n}\bigl|\frac{S_n(t)}{n}-f_S(e^{-t})\bigr|\top 0, & \ \   \sup\limits_{t\leq \tau_n}\bigl|\frac{D_n(t)}{n}-f_D(e^{-t})\bigr|\top 0,\\
\sup\limits_{t\leq \tau_n}\bigl|\frac{H^+_n(t)}{n}-f_{H^+}(e^{-t})\bigr|\top 0, & \ \   \sup\limits_{t\leq \tau_n}\bigl|\frac{I^+_n(t)}{n}-f_{I^+}(e^{-t})\bigr|\top 0,
\end{align*}
and the number of white balls satisfies
\begin{align*}
\sup\limits_{t\leq \tau_n}\bigl|\frac{W_n(t)}{n}-f_W(e^{-t})\bigr|\top 0.
\end{align*}

\end{theorem}
The proof of above theorem is provided in Appendix~\ref{sec:thmCentrality}.

\medskip

We consider now the stopping time  $\tau^\star_n$ which is the first time such that $W_n(\tau^\star_n)=-1$ (becomes negative). Let us define
 \begin{align}
 \hz:=\sup\bigl\{z\in[0,1]: f_W(z)=0\bigr\}.
 \end{align}
Then we have the following lemma.

\begin{lemma}\label{tau}
We have (as $n\rightarrow\infty$):
\begin{itemize}
  \item[(i)]  If $\hz=0$ then $\tau^\star_n\top\infty$.
  \item[(ii)] If $\hz\in\left(0,1\right]$ and $\hz$ is a stable solution, i.e. $f_W'(\hz)>0$, then $\tau^\star_n\top-\ln\hz$.
  \end{itemize}
\end{lemma}

The proof of lemma is provided in Appendix~\ref{sec:lemStop}.

\medskip

As a corollary of Theorem~\ref{centrality} and Lemma~\ref{tau}, we next provide the law of large numbers for the final state of default contagion. Namely, under Assumption~\ref{cond-average}, the following holds.
\begin{theorem}\label{thm:LLN}
The final fraction of defaults satisfies:
\begin{itemize}
\item[(i)] If $\hz=0$ then asymptotically almost all institutions default during the cascade and $$\bigl|\cD_{n-1}\bigr|=n-o_p(n).$$
\item[(ii)] If $\hz\in\left(0,1\right]$ and $\hz$ is a stable solution, i.e. $f_W'(\hz)>0$, then
$$\frac{\bigl|\cD_{n-1}\bigr|}{n}\top f_D(\hz).$$
Further, in this case, for all $x\in\cX, \theta=1, \dots, d_x^+$ and $\ell=0, \dots, \theta-1$, the final fraction of solvent institutions with type $x$, threshold $\theta$ and $\ell$ defaulted neighbors satisfy
\begin{align*}
\frac{S_{x,\theta,\ell}^{(n)}}{n} \top \mu_x q_x(\theta)b\left(d_x^+,1-\hz,\ell\right).
\end{align*}

\end{itemize}
\end{theorem}

The proof of above theorem is provided in Appendix~\ref{sec:thmLLN}. In particular, the above theorem (in different forms) has been used in~\cite{amcomi-res} to give a resilience condition for contagion in random financial networks. Namely, in the notations above, starting from a small fraction $\epsilon$ of institutions representing the fundamental defaults, i.e., $\sum_{x\in \cX} \mu_x q_x(0)=\epsilon$, the financial network is said to be resilient if $\lim_{\epsilon\to 0} \hz = 0$. We refer to~\cite{amcomi-res, amini2020contagion} for the resilience conditions.

\medskip

\subsection{Asymptotic Normality of Default Contagion}\label{sec:CLT}

In order to study the central limit theorems, we need to restrict our attention to the sparse networks regime. Namely, we consider the random graph $\cGn$ and assume that degrees sequences satisfy the following moment condition.

\begin{intassumption}\label{cond:CLT}
We assume that for every constant $A>1$, we have $$\sum_{i=1}^{n}A^{d^{+}_i} = n\sum_{x\in\cX}\mu_x^{(n)}A^{d^+_x}=O(n) \ \ \text{and} \ \  \sum_{i=1}^{n}A^{d^{-}_i} = n\sum_{x\in\cX}\mu_x^{(n)}A^{d^-_x}=O(n) .$$
\end{intassumption}

\begin{remark}\label{rem-finmo}
Let $(D_n^+,D_n^-)$ be random variables with joint distribution
\begin{align*}
\PP(D_n^+=d^+, D_n^-=d^-)=\sum_{x\in\cX} \mu^{(n)}_x \ind\{d_x^+=d^+,d_x^-=d^-\},
\end{align*}
which is the joint distribution of in- and out- degrees for a random node in $\cGn$. Let also $(D^+,D^-)$ be random variables (over nonnegative integers) with joint distribution
$$\PP(D^+=d^+, D^-=d^-)=\sum_{x\in\cX} \mu_x \ind\{d_x^+=d^+,d_x^-=d^-\}.$$
Then Assumption~\ref{cond:CLT} can be rewritten as $\EE[A^{D^+_n}]=O(1)$ and $\EE[A^{D^-_n}]=O(1)$ for each $A>1$, which in particular implies the uniform integrability of $D^+_n$ and $D_n^-$, so
\begin{equation*}
\lambda^{(n)}:=\sum_{x\in \cX} d_x^+ \mu_x^{(n)}  = \EE[D^+_n] \longrightarrow \EE[D^+]=\lambda \in (0,\infty).
\end{equation*}
Similarly, all higher moments converge. 
\end{remark}

By the construction of the balls-and-bins model, the independency exists between any two different types. Hence we study the asymptotic normality type by type. We first show the following joint convergence theorem for all $x\in\cX, \theta=1, \dots, d_x^+, \ell=0, \dots, \theta-1$.

\begin{theorem}\label{thm-CLT-S}
Let $\tau_n\leq \tau^\star_n$ be a stopping time such that $\tau_n\top t_0$ for some $t_0>0$.  For all $x\in\cX, \theta=1, \dots, d_x^+, \ell=0, \dots, \theta-1$ and jointly in $\cD\left[0,\infty\right)$,
\begin{align*}
n^{-1/2}\left(S^{(n)}_{x,\theta,\ell}(t\wedge\tau_n)-n\mu^{(n)}_x q^{(n)}_x(\theta)b\bigl(d_x^+,1-e^{-(t\wedge \tau_n)},\ell\bigr)\right)\tod Z_{x,\theta,\ell}(t\wedge t_0),
\end{align*}
where $Z_{x,\theta,\ell}(t)$ is a Gaussian process with mean 0 and variance $\bigsigma^2_{x,\theta,\ell}(t)$ given by~\eqref{Var_sll}.
\end{theorem}

The proof of above theorem is provided in Appendix~\ref{sec-CLT-S}. Indeed, we prove a stronger version of above theorem, providing also the covariance between $Z_{x_1,\theta_1,\ell_1}$ and $Z_{x_2,\theta_2,\ell_2}$, for any two triplets $(x_1,\theta_1,\ell_1)$ and $(x_2,\theta_2,\ell_2)$; the covariance is given by \eqref{Cov_sll}.

\medskip

For $z \in[0,1]$, we define the functions:
\begin{align*}
f^{(n)}_S(z):=& \sum_{x\in \cX} \mu^{(n)}_x \sum_{\theta=1}^{d_x^+}q^{(n)}_x(\theta) \beta\bigl(d_x^+, z, d_x^+-\theta+1\bigr), \ \ f^{(n)}_D(z)=1-f^{(n)}_S(z),\\
f_{H^+}^{(n)}(z):=& \sum_{x\in \cX} \mu_x^{(n)} \sum_{\theta=1}^{d_x^+}q^{(n)}_x(\theta) \sum_{\ell=d_x^+-\theta+1}^{d_x^+} \ell b\bigl(d_x^+,z,\ell\bigr) , \ \ f^{(n)}_{I^+}(z) = \lambda z-f^{(n)}_{H^+}(z),\\
f^{(n)}_W(z):=&\lambda z - \sum_{x\in \cX} \mu^{(n)}_x d_x^- \sum_{\theta=1}^{d_x^+}q^{(n)}_x(\theta) \beta\bigl(d_x^+,z, d_x^+-\theta+1\bigr).
\end{align*}

In the following theorem, we show the joint asymptotic normality between the total number of solvent institutions, number of defaulted institutions, number of infected and healthy links, and the total number of white balls (controlling the default contagion stopping time) at any time $t$ before the end of default cascade.
\medskip

For convenience, we set $$\hf^{(n)}_{\clubsuit}(t)=f^{(n)}_{\clubsuit}(e^{-t}),$$ for $\clubsuit\in \{S, D, H^+, I^+, W\}$.

\begin{theorem}\label{normality}
Let $\tau_n\leq \tau^\star_n$ be a stopping time such that $\tau_n\top t_0$ for some $t_0>0$. Then, under the above assumptions and jointly in $\cD\left[0,\infty\right)$,
\begin{align}\label{joint_normality}
n^{-1/2}\left(\clubsuit_n(t\wedge\tau_n)-n\hf^{(n)}_{\clubsuit}(t\wedge\tau_n)\right)\tod Z_{\clubsuit}(t\wedge t_0)
\end{align}
for $\clubsuit\in \{S, D, H^+, I^+, W\}$, where $\{Z_\clubsuit\}$ are continuous Gaussian processes on $[0,t_0]$ with mean 0 and covariances that satisfy, for $0\leq t\leq t_0$ and $\clubsuit, \spadesuit \in\{S, D, H^+, I^+, W\}$,
$$\Cov\bigl(Z_{\clubsuit}(t),Z_{\spadesuit}(t)\bigr)=\bigsigma_{\clubsuit,\spadesuit}(e^{-t}),$$
where the form of $\bigsigma_{\clubsuit,\spadesuit}(x)$ are given by \eqref{sigma S}-\eqref{sigma_HDIC}.
\end{theorem}

The proof of the theorem is provided in Appendix~\ref{sec:normality}. Note that since $D_n(t)=n-S_n(t)$, we only provide the covariances for $\clubsuit, \spadesuit \in\{S, H^+, I^+, W\}$. It would be easy to transfer the result to $D$ by setting $\bigsigma_{D,D}=\bigsigma_{S,S}$ and $\bigsigma_{D,\clubsuit}=-\bigsigma_{S,\clubsuit}$ .

\medskip

We further define
$$s_{x,\theta,\ell}(z):=\mu_x q_x(\theta)b\bigl(d_x^+,1-z,\ell\bigr), \ \ s^{(n)}_{x,\theta,\ell}(z):=\mu^{(n)}_x q^{(n)}_x(\theta)b\bigl(d_x^+,1-z,\ell\bigr).$$
As a corollary of Theorem~\ref{normality} and Lemma~\ref{tau}, we have the following theorem regarding the final state of default contagion.

\begin{theorem}\label{normalityFinal}
Let $t^\star=-\ln \hz$ and $\hp_n$ be the largest $z\in[0,1]$ such that $f_W^{(n)}(z)=0$. Under the above assumptions, we have:

\begin{itemize}
\item[(i)] If $\hz=0$ then asymptotically almost all institutions default during the cascade and $$\bigl|\cD_{n-1}\bigr|=n-o_p(n).$$
\item[(ii)] If $\hz\in\left(0,1\right]$ and $\hz$ is a stable solution, i.e. $\alpha:=f_W'(\hz)>0$, then we have 
       \begin{equation}\label{norm_tau}
            n^{-1/2}(\clubsuit_n(\tst)-nf^{(n)}_{\clubsuit}(\hp_n))
           \tod    Z_{\clubsuit}(t^{\star})-\alpha^{-1}f'_{\clubsuit}(\hz)Z_{W}(t^{\star}),
       \end{equation}
      for $\clubsuit\in \{S, D, H^+, I^+, W\}$, where the limit distributions  compose a Gaussian vector. Furthermore, $\hp_n\rightarrow\hz$ and, for all $x\in \cX$, $0\leq\ell < \theta\leq d^+_x$,  
       \begin{equation}
       n^{-1/2}(S^{(n)}_{x,\theta,\ell}(\tst)-ns^{(n)}_{x,\theta,\ell}(\hp_n))\tod Z_{x,\theta,\ell}(t^{\star})-\alpha^{-1}s'_{x,\theta,\ell}(\hz)Z_{W}(t^{\star}).
       \end{equation}
\end{itemize}
\end{theorem}

The proof of the theorem is provided in Appendix~\ref{sec:normalityFinal}.

\section{Quantifying Systemic Risk}\label{sec:sys}
In order to determine the health of the financial network, we consider in this section a \emph{systemic risk measure} applied to the (random) financial network, introduced in previous sections. These measures are decomposed as $\rho \circ \Gamma$ for a stand-alone risk measure (usually assumed convex) $\rho$ and an aggregation function $\Gamma=\Gamma(\beps)$ for losses under the stress scenario $\beps$. This was first introduced in~\cite{chen2013axiomatic,kromer2013systemic}; see also~\cite{feinstein2017measures, amini2020optimal}.

\medskip

The following three aggregation functions has been considered in the literature. At time $t$, for the economy $\cE_n$ and given shock scenario $\beps$, we let:
\begin{itemize}
\item {\bf Number of solvent banks}: $\Gamma_n^\#(t) := S_n(t)=n-D_n(t)$.
\item {\bf External wealth}: Let $\bar{\Gamma}_n^{\odot} $ denotes the total external wealth to society if there is no default in the financial system (small shock regime). 
We define the external wealth (societal) aggregation function as $$\Gamma_n^{\odot}(t) := \bar{\Gamma}^{\odot}_n -\sum_{x\in \cX} \bar{L}^{\odot}_xD_x^{(n)}(t),$$
where $D_x^{(n)}(t)=n\mu_x^{(n)} - \sum_{\theta}\sum_{\ell=0}^{\theta-1}S_{x,\theta,\ell}^{(n)}(t)$ denotes the total number of defaulted institutions with type $x\in \cX$ at time $t$. Note that (for simplicity) we assume a bounded constant type-dependent societal loss $\bar{L}^{\odot}_x$ over each defaulted institution. We discuss in Appendix~\ref{app:Ext}, how one can extend the model to random i.i.d. (type-dependent) losses (with bounded support).
\item {\bf System-wide wealth} : Let $\bar{\Gamma}_n^{\Diamond} $ denotes the total wealth in the financial system if there is no default in the system. 
We define the system-wide aggregation function as   $$\Gamma_n^{\Diamond}(t) := \bar{\Gamma}^{\Diamond}_n-\sum_{x\in \cX} \bar{L}^{\odot}_xD_x^{(n)}(t)- \sum_{x\in \cX} \bar{L}^{\Diamond}_x \sum_{\theta=1}^{d_x^+}\sum_{\ell=1}^{\theta-1} \ell S^{(n)}_{x, \theta, \ell} (t),$$
We consider a fixed (type-dependent) societal cost $\bar{L}^{\odot}_x$ for defaulted institutions and a fixed (host institutions' type-dependent) bounded cost $\bar{L}^{\Diamond}_x$ over each defaulted links. The possible extensions to independent random losses will be discussed in Appendix~\ref{app:Ext}.
\end{itemize}

For the aggregation function $\Gamma_n^\#(t)$, we already stated the limit theorems in Section~\ref{sec:main}. Since the societal aggregation function $\Gamma_n^{\odot}$ can be seen as a particular case of  system-wide aggregation function $\Gamma_n^{\Diamond}$ (by setting $\bar{L}^{\Diamond}_x=0$), we only state limit theorems for $\Gamma_n^{\Diamond}$. 

\medskip

In order to state limit theorems, it is natural to assume that $\bar{\Gamma}_n^{\Diamond}/n\rightarrow\bar{\Gamma}^{\Diamond}$ when the size of network $n\rightarrow\infty$. Let us define
$$f^{(n)}_{\Diamond}(z):=\bar{\Gamma}^{\Diamond}_n/n-\sum_{x\in \cX} \bar{L}^{\odot}_x f^{(n)}_D(z)
 - \sum_{x\in \cX} \bar{L}^{\Diamond}_x \sum_{\theta=1}^{d_x^+}\sum_{\ell=1}^{\theta-1} \ell s^{(n)}_{x, \theta, \ell} (z),$$
$$f_{\Diamond}(z):=\bar{\Gamma}^{\Diamond}-\sum_{x\in \cX} \bar{L}^{\odot}_x f_D(z)
 - \sum_{x\in \cX} \bar{L}^{\Diamond}_x \sum_{\theta=1}^{d_x^+}\sum_{\ell=1}^{\theta-1} \ell s_{x, \theta, \ell} (z).$$
Similarly we also set
\begin{equation*}
  \hf^{(n)}_{\Diamond}(t) =f^{(n)}_{\Diamond}(e^{-t}),\ \ \hf_{\Diamond}(t)=f_{\Diamond}(e^{-t}).
\end{equation*}

\medskip

By applying Theorem~\ref{centrality} and Theorem~\ref{thm:LLN}, under Assumption~\ref{cond-average}, the following holds.

\begin{theorem}\label{thm-agg-LLN}
Let Assumption~\ref{cond-average} holds and $\tau_n\leq \tau^\star_n$ be a stopping time such that $\tau_n\top t_0$ for some $t_0>0$. Then, as $n\to \infty$,
\begin{align}\label{Gamma_centra}
\sup\limits_{t\leq \tau_n}\bigl|\frac{\Gamma_n^{\Diamond}(t) }{n}-f_{\Diamond}(e^{-t})\bigr|\top 0.
\end{align}
Further, the final (system-wide) aggregation functions satisfy:
\begin{itemize}
\item[(i)] If $\hz=0$ then asymptotically almost all institutions default during the cascade and
\begin{align*}
\frac{\Gamma_n^{\Diamond}(\tst)}{n} \top\bar{\Gamma}^{\Diamond} -\sum_{x\in \cX}\mu_x \bar{L}^{\odot}_x.
\end{align*}
\item[(ii)] If $\hz\in\left(0,1\right]$ and $\hz$ is a stable solution, i.e. $f_W'(\hz)>0$, then
\begin{align*}
\frac{\Gamma_n^{\Diamond}(\tst)}{n}\top f_{\Diamond}(\hz).
\end{align*}
\end{itemize}

\end{theorem}

The proof of the above theorem is provided in Appendix~\ref{sec:thm-agg-LLN}.

\medskip

We next consider the central limit theorems for the societal and system-wide aggregation functions. By applying Theorem~\ref{normality}, under Assumption~\ref{cond:CLT}, the following holds.

\begin{theorem}\label{thm-CLT-Sys}
Let Assumption~\ref{cond:CLT} holds and $\tau_n\leq \tau^\star_n$ be a stopping time such that $\tau_n\top t_0$ for some $t_0>0$. Then jointly in $\cD\left[0,\infty\right)$,
\begin{align}\label{norm_risk}
n^{-1/2}\left(\Gamma_n^{\Diamond}(t\wedge\tau_n) -n\hf^{(n)}_{\Diamond}(t\wedge\tau_n)\right)\tod Z_{\Diamond}(t\wedge t_0),
\end{align}
where $Z_{\Diamond}$ is a continuous Gaussian process on $[0,t_0]$ with mean 0 and variance
$$\bigsigma_{\Diamond}^2(t)=\bigsigma_{1,1}(e^{-t})+2\bigsigma_{1,2}(e^{-t})+\bigsigma_{2,2}(e^{-t}),$$
where $\bigsigma_{i,j}(y), i,j=1,2$, are given by \eqref{sigma_11}-\eqref{sigma_12}.

Moreover, the final system-wide aggregation function satisfy:
\begin{itemize}
\item[(i)] If $\hz=0$ then asymptotically almost all institutions default during the cascade and
\begin{align*}
\frac{\Gamma_n^{\Diamond}(\tst)}{n}\top \bar{\Gamma}^{\Diamond} -\sum_{x\in \cX}\mu_x \bar{L}^{\odot}_x.
\end{align*}
\item[(ii)] If $\hz\in\left(0,1\right]$ and $\hz$ is a stable solution, i.e. $\alpha:=f_W'(\hz)>0$, then we have
\begin{align*}
n^{-1/2}\bigl(\Gamma_n^{\Diamond}(\tst) -n\hf^{(n)}_{\Diamond}(\hp_n)\bigr)\tod \cZ_{\Diamond},
\end{align*}
where $\cZ_{\Diamond}$ is a centered Gaussian random variable with variance $\Sigma_\Diamond$ satisfying 
$$\Sigma_{\Diamond}=\bigsigma_{\Diamond}^2(t^{\star})+\Delta(\hz)^2\bigsigma_{W,W}(\hz)-2\Delta(\hz)\bigsigma_{\Diamond,W}(\hz),$$
and
$$\Delta(z):=(f'_W(z))^{-1}f'_{\Diamond}(z), \bigsigma_{\Diamond,W}(e^{-t}):=\Cov(Z_{\Diamond}(t),Z_{W}(t)).$$
The form of $\bigsigma_{\Diamond,W}(y)$ is given by \eqref{sigma_boxW}-\eqref{sigma_boxW2}.
\end{itemize}

\end{theorem}

The proof of the above theorem is provided in Appendix~\ref{sec:thm-CLT-Sys}.

\section{Targeting Interventions in Financial Networks}\label{sec:intervene}

In this section we consider a planner (lender of last resort or government) who seeks to minimize the systemic risk at the beginning of  the financial contagion, after an exogenous macroeconomic shock $\beps$, subject to a budget constraint. As discussed in Section~\ref{sec:sys}, we assume that the systemic risk is represented by $\rho(\Gamma_n^{\Diamond})$, for some convex function $\rho$ applied to system-wide wealth. Note that since we study the interventions for a given shock $\beps$, the uncertainty (in stress scenario) for the risk measure $\rho$ is only on the network structure (which is assumed to be uniformly at random).  The planner only has information regarding the type of each institution and, consequently, the institutions' threshold distributions.  Hence, the planner' decision is only based on the type of each institution. 

The timeline is as follows. 
At time $t=0$, the financial network is subject to an economic shock $\beps$. At time $t=1$, the planer (observing the external shock $\beps$) calculates the threshold distribution $q_x(.)$ for each $x\in \cX$. Then she makes decisions, under some budget constraint, on the number (fraction) of interventions over all defaulted links leading to any institutions with any type $x\in \cX$. When the planner intervenes on a defaulting bank, its threshold (distance to default) increases by 1. 
These interventions will be type-dependent and at random over all defaulted links leading to the same type institutions. 

For $x\in \cX$, let us denote by $\alpha_x^{(n)}$ the planner intervention decision on the fraction of the  saved links leading to any institution of type $x\in \cX$. We assume that $\alpha_x^{(n)} \to \alpha_x$ for  all $x\in \cX$, and some constants $\alpha_x$ independent of $n$. Let $\balpha_n=\left\{\alpha^{(n)}_x\right\}_{x\in \cX}$, and, $\Gamma_n^{\Diamond}(\balpha_n)$ denotes the system-wide wealth under the intervention decision $\balpha_n$. Further, $S_{x,\theta,\ell}^{(n)}(\balpha_n)$ denotes the number of solvent banks with type $x$, threshold $\theta$ and $\ell$ defaulted neighbors under the intervention decision $\balpha_n$. Similarly, $D_n(\balpha_n)$ denotes the total number of defaults under intervention $\balpha_n$. 


Let $C_x\in \RR^+$ denotes the cost associated to saving any defaulted link leading to an institution of type $x\in \cX$. We assume that $C_x$ is a bounded function. We denote by $\Phi_n(\balpha_n)$ the total cost associated to the planner for the intervention strategy $\balpha_n$.

We next state a limit theorem regarding the number of solvent institutions, defaulted institutions, the total aggregate wealth of the financial system and the total cost of intervention for the planner, under the intervention decision $\balpha_n$. 

Let us define
$$f_{W}^{(\balpha)}(z):= \lambda z - \sum_{x\in \cX} \mu_x d_x^-\sum_{\theta=1}^{d_x^+}q_x(\theta) \beta\bigl(d_x^+,\alpha_x +(1-\alpha_x)z, d_x^+-\theta+1\bigr),$$
and, 
 \begin{align}\label{hz:alpha}
 \hz_{\balpha}:=\sup\bigl\{z\in[0,1]: f^{(\balpha)}_W(z)=0\bigr\}.
 \end{align}

\begin{theorem}\label{planner-LLN}
Let Assumption~\ref{cond-average} holds and $\balpha_n \to \balpha$ as $n\to \infty$. If $z^{\star}_{\balpha}$ is a stable solution, then as $n \to \infty$:

\begin{itemize}
\item[(i)] For all $x\in\cX, \theta=1, \dots, d_x^+$ and $\ell=0, \dots, \theta-1$, the final fraction of solvent institutions with type $x$, threshold $\theta$ and $\ell$ defaulted neighbors under intervention $\balpha_n$ converges to
\begin{align*}
\frac{S_{x,\theta,\ell}^{(n)}(\balpha_n)}{n} \top s^{(\balpha)}_{x, \theta, \ell}(\hz_{\balpha}):=\mu_x q_x(\theta)b\left(d_x^+,(1-\alpha_x)(1-\hz_{\balpha}),\ell\right).
\end{align*}
\item[(ii)] The total number of defaulted institutions under intervention $\balpha_n$ converges to:
\begin{align*}
\frac{D_n(\balpha_n)}{n} \top  f_D^{(\balpha)}(\hz_{\balpha}):= 1- \sum_{x\in \cX} \mu_x \sum_{\theta=1}^{d_x^+}q_x(\theta) \beta\bigl(d_x^+, \alpha_x +(1-\alpha_x)\hz_{\balpha}, d_x^+-\theta+1\bigr).
\end{align*}
\item[(iii)] The system-wide wealth under the intervention decision $\balpha_n$ converges to
\begin{align*}
\frac{\Gamma_n^{\Diamond}(\balpha_n)}{n} \top f^{(\balpha)}_{\Diamond}(\hz_{\balpha}):=\bar{\Gamma}^{\Diamond}-\sum_{x\in \cX} \bar{L}^{\odot}_x f_D^{(\balpha)}(\hz_{\balpha})
 - \sum_{x\in \cX} \bar{L}^{\Diamond}_x \sum_{\theta=1}^{d_x^+}\sum_{\ell=1}^{\theta-1} \ell s^{(\balpha)}_{x, \theta, \ell} (\hz_{\balpha}).
 \end{align*}
 \item[(iv)] The total cost of interventions $\balpha_n$ for the planner converges to
 \begin{align*}
\frac{\Phi_n(\balpha_n)}{n} \top \phi(\hz_{\balpha}):= \sum_{x\in \cX} \mu_x  \alpha_x C_x \sum_{\ell=1}^{d_x^+} \ell b\left(d_x^+,1-\hz_{\balpha},\ell\right).
\end{align*}
\end{itemize}
\end{theorem}
The proof of theorem is provided in Appendix~\ref{sec:planner-LLN}. 

\medskip

We conclude that, as $n\to \infty$, the planner optimal decision problem simplifies to 
\begin{align*}
\max_{\balpha} f^{(\balpha)}_{\Diamond}(\hz_{\balpha}):=&\bar{\Gamma}^{\Diamond}-\sum_{x\in \cX} \bar{L}^{\odot}_x f_D^{(\balpha)}(\hz_{\balpha})
 - \sum_{x\in \cX} \bar{L}^{\Diamond}_x \sum_{\theta=1}^{d_x^+}\sum_{\ell=1}^{\theta-1} \ell s^{(\balpha)}_{x, \theta, \ell} (\hz_{\balpha}),\\
 \text{subject to} \quad  & \phi(\hz_{\balpha}):= \sum_{x\in \cX} \mu_x  \alpha_x C_x \sum_{\ell=1}^{d_x^+} \ell b\left(d_x^+,1-\hz_{\balpha},\ell\right) \leq C,
\end{align*}
for some budget constraint $C>0$ and $\hz_{\balpha}$ given by \eqref{hz:alpha}.

\medskip

It would be interesting to extend the model to a continuous-time Markov decision process by the planner, where the links (starting from fundamentally defaulted institutions) are revealed one by one, and consequently, the planner should decide at any time to intervene or not. This would lead to a Markov decision problem and one could solve it (under some regularity assumptions) by using a dynamic programming approach. We refer to \cite{amini2015control} for a similar model in a simpler setup with a core-periphery network structure. We leave this and some related issues to a future work.


\section{Concluding Remarks}
In this paper, we propose a general tractable framework to study default cascades and systemic risk in a heterogeneous financial network, subject to an exogenous macroeconomic shock. We state various limit theorems for the final state of default contagion and systemic risk depending on the network structure and institutions' (observable) characteristics.  

We show that, under some regularity assumptions, the default cascade could be transferred to a death process problem represented by balls-and-bins model. Under suitable assumptions on the degree and threshold distributions, we show that the final size of default cascade and system-wide wealth aggregation functions have asymptotically Gaussian fluctuations. We also show how these results could be used by a social planner  to optimally target interventions during a financial crisis, under partial information of the financial network and under some budget constraint.

The closed form interpretable limit theorems that we provide could also serve as a mandate for regulators to collect data on those specific network characteristics  and assess systemic risk via more intensive computational methods. 

Our results can also be used in a regulatory risk management framework.  The regulator will impose capital requirements on each bank. 
In practice, the capital ratio constraint is the same for all banks. However, using our heterogeneous setup, we could allow the regulator to choose optimally this capital ratio according to the type of the banks. Fixing the time horizon $T$, the regulator's problem is how to choose optimally capital ratio for the institutions of type $x\in \cX$ at time 0 to minimize $\rho\bigl(\Gamma_{\Diamond}(T)\bigr)$. Then the regulator can update the type of each bank and capital ratios at time $T$ according to the new data.

\bibliographystyle{plain}
\bibliography{biblioDEF.bib}

\providecommand{\noopsort}[1]{}
\begin{thebibliography}{10}

\bibitem{acemoglu2015systemic}
Daron Acemoglu, Asuman Ozdaglar, and Alireza Tahbaz-Salehi.
\newblock Systemic risk and stability in financial networks.
\newblock {\em American Economic Review}, 105(2):564--608, 2015.

\bibitem{acharya2017measuring}
Viral~V Acharya, Lasse~H Pedersen, Thomas Philippon, and Matthew Richardson.
\newblock Measuring systemic risk.
\newblock {\em The Review of Financial Studies}, 30(1):2--47, 2017.

\bibitem{akbarpour2018diffusion}
Mohammad Akbarpour, Suraj Malladi, and Amin Saberi.
\newblock Diffusion, seeding, and the value of network information.
\newblock In {\em Proceedings of the 2018 ACM Conference on Economics and
  Computation}, pages 641--641, 2018.

\bibitem{allen2000financial}
Franklin Allen and Douglas Gale.
\newblock Financial contagion.
\newblock {\em Journal of political economy}, 108(1):1--33, 2000.

\bibitem{ar:AmFount2012}
H.~Amini and N.~Fountoulakis.
\newblock Bootstrap percolation in power-law random graphs.
\newblock {\em Journal of Statistical Physics}, 155:72--92, 2014.

\bibitem{amini10}
Hamed Amini.
\newblock Bootstrap percolation and diffusion in random graphs with given
  vertex degrees.
\newblock {\em Electronic Journal of Combinatorics, 17: R25}, 2010.

\bibitem{amini2020contagion}
Hamed Amini.
\newblock Contagion risks and security investment in directed networks.
\newblock {\em Available at SSRN 3654657}, 2020.

\bibitem{AmBaCh21}
Hamed Amini, Erhan Bayraktar, and Suman Chakraborty.
\newblock A central limit theorem for diffusion in sparse random graphs, 2021.

\bibitem{amcomi-res}
Hamed Amini, Rama Cont, and Andreea Minca.
\newblock Resilience to contagion in financial networks.
\newblock {\em Mathematical Finance}, 26(2):329--365, 2016.

\bibitem{amini2020optimal}
Hamed Amini and Zachary Feinstein.
\newblock Optimal network compression.
\newblock {\em arXiv preprint arXiv:2008.08733}, 2020.

\bibitem{amini2016inhomogeneous}
Hamed Amini and Andreea Minca.
\newblock Inhomogeneous financial networks and contagious links.
\newblock {\em Operations Research}, 64(5):1109--1120, 2016.

\bibitem{amini2015control}
Hamed Amini, Andreea Minca, and Agnes Sulem.
\newblock Control of interbank contagion under partial information.
\newblock {\em SIAM Journal on Financial Mathematics}, 6(1):1195--1219, 2015.

\bibitem{amini2019dynamic}
Hamed Amini, Andreea Minca, and Agn{\`e}s Sulem.
\newblock A dynamic contagion risk model with recovery features.
\newblock {\em Available at SSRN 3435257}, 2019.

\bibitem{anand2015filling}
Kartik Anand, Ben Craig, and Goetz Von~Peter.
\newblock Filling in the blanks: Network structure and interbank contagion.
\newblock {\em Quantitative Finance}, 15(4):625--636, 2015.

\bibitem{armenter2014balls}
Roc Armenter and Mikl{\'o}s Koren.
\newblock A balls-and-bins model of trade.
\newblock {\em American Economic Review}, 104(7):2127--51, 2014.

\bibitem{billingsley1968convergence}
P.~Billingsley, Karreman Mathematics~Research Collection, W.A. Shewhart, Wiley
  series~in probability, mathematical statistics, and S.S. Wilks.
\newblock {\em Convergence of Probability Measures}.
\newblock Wiley Series in Probability and Mathematical Statistics. Wiley, 1968.

\bibitem{boss04}
Michael Boss, Helmut Elsinger, Martin Summer, and Stefan Thurner.
\newblock Network topology of the interbank market.
\newblock {\em Quantitative Finance}, 4(6):677--684, 2004.

\bibitem{cabrales2017risk}
Antonio Cabrales, Piero Gottardi, and Fernando Vega-Redondo.
\newblock Risk sharing and contagion in networks.
\newblock {\em The Review of Financial Studies}, 30(9):3086--3127, 2017.

\bibitem{chen2013axiomatic}
Chen Chen, Garud Iyengar, and Ciamac~C. Moallemi.
\newblock An axiomatic approach to systemic risk.
\newblock {\em Management Science}, 59(6):1373--1388, 2013.

\bibitem{chinazzi2015systemic}
Matteo Chinazzi and Giorgio Fagiolo.
\newblock Systemic risk, contagion, and financial networks: A survey.
\newblock {\em Institute of Economics, Scuola Superiore Sant’Anna, Laboratory
  of Economics and Management (LEM) Working Paper Series}, (2013/08), 2015.

\bibitem{cont2010network}
Rama Cont, Amal Moussa, et~al.
\newblock Network structure and systemic risk in banking systems.
\newblock {\em Edson Bastos e, Network Structure and Systemic Risk in Banking
  Systems (December 1, 2010)}, 2010.

\bibitem{coupechoux2014clustering}
Emilie Coupechoux and Marc Lelarge.
\newblock How clustering affects epidemics in random networks.
\newblock {\em Advances in Applied Probability}, 46(4):985--1008, 2014.

\bibitem{craig2014interbank}
Ben Craig and Goetz Von~Peter.
\newblock Interbank tiering and money center banks.
\newblock {\em Journal of Financial Intermediation}, 23(3):322--347, 2014.

\bibitem{d2019compressing}
Marco D'Errico and Tarik Roukny.
\newblock Compressing over-the-counter markets.
\newblock {\em Available at SSRN 2962575}, 2019.

\bibitem{detering2019managing}
Nils Detering, Thilo Meyer-Brandis, Konstantinos Panagiotou, and Daniel Ritter.
\newblock Managing default contagion in inhomogeneous financial networks.
\newblock {\em SIAM Journal on Financial Mathematics}, 10(2):578--614, 2019.

\bibitem{eisenberg2001systemic}
Larry Eisenberg and Thomas~H Noe.
\newblock Systemic risk in financial systems.
\newblock {\em Management Science}, 47(2):236--249, 2001.

\bibitem{elliott2014financial}
Matthew Elliott, Benjamin Golub, and Matthew~O Jackson.
\newblock Financial networks and contagion.
\newblock {\em American Economic Review}, 104(10):3115--53, 2014.

\bibitem{erol2018network}
Selman Erol and Rakesh Vohra.
\newblock Network formation and systemic risk.
\newblock {\em Available at SSRN 2546310}, 2018.

\bibitem{feinstein2017measures}
Zachary Feinstein, Birgit Rudloff, and Stefan Weber.
\newblock Measures of systemic risk.
\newblock {\em SIAM Journal on Financial Mathematics}, 8(1):672--708, 2017.

\bibitem{fricke2015core}
Daniel Fricke and Thomas Lux.
\newblock Core--periphery structure in the overnight money market: evidence
  from the e-mid trading platform.
\newblock {\em Computational Economics}, 45(3):359--395, 2015.

\bibitem{gai2010contagion}
Prasanna Gai and Sujit Kapadia.
\newblock Contagion in financial networks.
\newblock {\em Proceedings of the Royal Society A: Mathematical, Physical and
  Engineering Sciences}, 466(2120):2401--2423, 2010.

\bibitem{galeotti2020targeting}
Andrea Galeotti, Benjamin Golub, and Sanjeev Goyal.
\newblock Targeting interventions in networks.
\newblock {\em Econometrica}, 88(6):2445--2471, 2020.

\bibitem{gandy2017bayesian}
Axel Gandy and Luitgard~AM Veraart.
\newblock A bayesian methodology for systemic risk assessment in financial
  networks.
\newblock {\em Management Science}, 63(12):4428--4446, 2017.

\bibitem{glasserman2015likely}
Paul Glasserman and H~Peyton Young.
\newblock How likely is contagion in financial networks?
\newblock {\em Journal of Banking \& Finance}, 50:383--399, 2015.

\bibitem{iori2015networked}
Giulia Iori, Rosario~N Mantegna, Luca Marotta, Salvatore Micciche, James
  Porter, and Michele Tumminello.
\newblock Networked relationships in the e-mid interbank market: A trading
  model with memory.
\newblock {\em Journal of Economic Dynamics and Control}, 50:98--116, 2015.

\bibitem{jackson2020credit}
Matthew~O Jackson and Agathe Pernoud.
\newblock Credit freezes, equilibrium multiplicity, and optimal bailouts in
  financial networks.
\newblock {\em arXiv preprint arXiv:2012.12861}, 2020.

\bibitem{jackson2020systemic}
Matthew~O Jackson and Agathe Pernoud.
\newblock Systemic risk in financial networks: A survey.
\newblock {\em Available at SSRN}, 2020.

\bibitem{jacod2013limit}
Jean Jacod and Albert Shiryaev.
\newblock {\em Limit theorems for stochastic processes}, volume 288.
\newblock Springer Science \& Business Media, 2013.

\bibitem{ar:JLTV10}
S.~Janson, T.~{\L}uczak, T.~Turova, and T.~Vallier.
\newblock Bootstrap percolation on the random graph ${G}_{n,p}$.
\newblock {\em The Annals of Applied Probability}, 22(5):1989--2047, 2012.

\bibitem{janson2009percolation}
Svante Janson et~al.
\newblock On percolation in random graphs with given vertex degrees.
\newblock {\em Electronic Journal of Probability}, 14:86--118, 2009.

\bibitem{janson2008asymptotic}
Svante Janson, Malwina~J Luczak, et~al.
\newblock Asymptotic normality of the k-core in random graphs.
\newblock {\em The annals of applied probability}, 18(3):1085--1137, 2008.

\bibitem{kallenberg1997foundations}
Olav Kallenberg.
\newblock {\em Foundations of modern probability}, volume~2.
\newblock Springer, 1997.

\bibitem{kleinberg2007cascading}
Jon Kleinberg.
\newblock Cascading behavior in networks: Algorithmic and economic issues.
\newblock {\em Algorithmic game theory}, 24:613--632, 2007.

\bibitem{kromer2013systemic}
Eduard Kromer, Ludger Overbeck, and Katrin Zilch.
\newblock Systemic risk measures on general probability spaces.
\newblock {\em Mathematical Methods of Operations Research}, 84(2):323--357,
  2016.

\bibitem{leitner2005financial}
Yaron Leitner.
\newblock Financial networks: Contagion, commitment, and private sector
  bailouts.
\newblock {\em The Journal of Finance}, 60(6):2925--2953, 2005.

\bibitem{lelarge12b}
Marc Lelarge.
\newblock Diffusion and cascading behavior in random networks.
\newblock {\em Games and Economic Behavior}, 75(2):752--775, 2012.

\bibitem{li2019dealer}
Dan Li and Norman Sch{\"u}rhoff.
\newblock Dealer networks.
\newblock {\em The Journal of Finance}, 74(1):91--144, 2019.

\bibitem{morris2000contagion}
Stephen Morris.
\newblock Contagion.
\newblock {\em The Review of Economic Studies}, 67(1):57--78, 2000.

\bibitem{morris2009illiquidity}
Stephen Morris and Hyun~Song Shin.
\newblock Illiquidity component of credit risk.
\newblock Technical report, working paper, Princeton University, 2009.

\bibitem{muller2006interbank}
Jeannette M{\"u}ller.
\newblock Interbank credit lines as a channel of contagion.
\newblock {\em Journal of Financial Services Research}, 29(1):37--60, 2006.

\bibitem{nier2007network}
Erlend Nier, Jing Yang, Tanju Yorulmazer, and Amadeo Alentorn.
\newblock Network models and financial stability.
\newblock {\em Journal of Economic Dynamics and Control}, 31(6):2033--2060,
  2007.

\bibitem{rogers2013failure}
Leonard~CG Rogers and Luitgard~AM Veraart.
\newblock Failure and rescue in an interbank network.
\newblock {\em Management Science}, 59(4):882--898, 2013.

\bibitem{roukny2018interconnectedness}
Tarik Roukny, Stefano Battiston, and Joseph~E Stiglitz.
\newblock Interconnectedness as a source of uncertainty in systemic risk.
\newblock {\em Journal of Financial Stability}, 35:93--106, 2018.

\bibitem{tobias2016covar}
Adrian Tobias and Markus~K Brunnermeier.
\newblock Covar.
\newblock {\em The American Economic Review}, 106(7):1705, 2016.

\bibitem{upper2004estimating}
Christian Upper and Andreas Worms.
\newblock Estimating bilateral exposures in the german interbank market: Is
  there a danger of contagion?
\newblock {\em European economic review}, 48(4):827--849, 2004.

\bibitem{hofstad16}
Remco {van der}~Hofstad.
\newblock {\em Random Graphs and Complex Networks}.
\newblock Cambridge University Press, 2016.

\bibitem{van2015hierarchical}
Remco van~der Hofstad, Johan~SH van Leeuwaarden, and Clara Stegehuis.
\newblock Hierarchical configuration model.
\newblock {\em arXiv preprint arXiv:1512.08397}, 2015.

\bibitem{veraart2020does}
Luitgard Anna~Maria Veraart.
\newblock When does portfolio compression reduce systemic risk?
\newblock {\em Available at SSRN 3688495}, 2020.

\bibitem{watts2002simple}
Duncan~J Watts.
\newblock A simple model of global cascades on random networks.
\newblock {\em Proceedings of the National Academy of Sciences},
  99(9):5766--5771, 2002.

\bibitem{young2009innovation}
H~Peyton Young.
\newblock Innovation diffusion in heterogeneous populations: Contagion, social
  influence, and social learning.
\newblock {\em American economic review}, 99(5):1899--1924, 2009.

\end{thebibliography}

\appendix

\section{Proofs}\label{sec:proofs}
This section contains all the proofs in the main text. We first provide some preliminary results on death processes and martingale theory that will be used in our proofs.
\subsection{Some Death Processes}
The studies of death processes in the balls-and-bins problem rely on the following classical result, see e.g.~\cite[Proposition 4.24]{kallenberg1997foundations}.
\begin{lemma}[The Glivenko-Cantelli Theorem]\label{Gli_cantelli}
Let $T_1,\ldots,T_n$ be i.i.d random variables with distribution function $F(t):=\mathbb{P}(T_i\leq t)$. Let $X_n(t)$ be their empirical distribution function $X_n(t):=\#\{i\leq n:T_i\leq t\}/n$. Then $\sup_{t}|X_n(t)-F(t)|\top 0$ as $n\rightarrow\infty$.
\end{lemma}

Since in the contagion process described in Section~\ref{sec:deathFin}, every in- ball dies independently after an exponential time with parameter one, we consider next a pure death process starting with some number of  balls whose lifetimes are i.i.d $\exp(1)$.

\begin{lemma}[Death Process Lemma] Let $N^{(n)}(t)$ be the number of balls alive at time $t$ in a pure death process with rate 1, and starting with initial number $N^{(n)}(0)=n$. Then
$$\sup\limits_{t\geq 0}|N^{(n)}(t)/n-e^{-t}|\top 0 \quad \quad \quad \text{as} \quad n\to \infty.$$
\end{lemma}

\begin{pr}
Since $1-N^{(n)}(t)/n$ is the empirical distribution function of the $n$ i.i.d. random variables of distribution $\exp(1)$, whose distribution function is $1-e^{-t}$, thus the result follows by using the Glivenko-Cantelli theorem.\\
\end{pr}


\subsection{Some Martingale Theory}
Our proof of the asymptotic normality for the default contagion is based on a martingale limit theorem by Jacod and Shiryaev~\cite{jacod2013limit}. 
Let $X$ be a martingale defined on $\left[0,\infty\right)$. We denote its quadratic variation by
$[X,X]_t$. We also denote the (bilinear) covariation of two martingales $X$ and $Y$ by $[X,Y]_t$.  In particular, if $X$ and $Y$ are two martingales with path-wise finite variation, then
$$[X,Y]_t:=\sum_{0<s\leq t}\Delta X(s)\Delta Y(s),$$
where $\Delta X(s):= X(s)- X(s-)$ is the jump of $X$ at $s$ and similarly $\Delta Y(s):= Y(s)- Y(s-)$. Note that in this paper, the considered martingales  are always RCLL (right continuous and with left limit) and have only finite
number of jumps. Hence, the quadratic variation will be finite. We also set $[X,Y]_0=0$. For two vector-valued martingales $\bX=(X_i)_{i=1}^{n}$ and $\bY=(Y_j)_{j=1}^{m}$, we define $[\bX,\bY]$ to be the $n\times m$ real matrix with every entry being $([\bX,\bY])_{i,j}=[X_i,Y_j]$.

\medskip

We will use the following martingale limit theorem from \cite{jacod2013limit}.

\begin{theorem}[Martingale Limit Theorem]\label{martingale_clt}
Assume that for each $n$, $\bM^{(n)}(t)=(\bM^{(n)}_i(t))_{i=1}^{q}$ is a $q$-dimensional martingale on $\interval$ with $\bM^{(n)}(0)={\mathbf 0}$, and that $\Sigma(t)$ is a (nonrandom) continuous matrix-valued function satisfying, for every fixed $t\geq 0$,
\begin{itemize}
\item[(i)] $[\bM^{(n)},\bM^{(n)}]_t\top\Sigma(t)$ as  $n\rightarrow\infty$,
\item[(ii)] $\sup_{n}\EE[M^{(n)}_i,M^{(n)}_i]_t<\infty$, for all $i=1,\ldots,q$.
\end{itemize}
Then $\bM^{(n)}\tod \bZ$ as $n\rightarrow\infty$, in $\cD\interval$, where $\bZ$ is a continuous $q$-dimensional Gaussian martingale with $\EE [\bZ(t)]={\mathbf 0}$ and covariances
$$\EE \left[\bZ(t)\bZ'(s)\right]=\Sigma(t), \quad \quad 0\leq t\leq s<\infty.$$
\end{theorem}

Hence, the above theorem yields joint convergence of the processes $\bigl\{M^{(n)}_i\bigr\}_{i=1}^q$ and could be extended immediately to infinitely many processes (i.e., for the case $q=\infty$). Indeed, by definition, an infinite family of stochastic processes converge jointly if every finite subfamily does. We will use the above theorem for stopped martingales.

%
%
%
%
%

\subsection{Proof of Theorem~\ref{centrality}} \label{sec:thmCentrality}
We denote by $U^{(n)}_{x, \theta,\ell}(t)$ the number of bins (institutions) with type $x\in \cX$, threshold $\theta$ and $\ell$ alive (in-) balls at time $t$. Let $N_{x,\theta}^{(n)}$ denotes the (random) number of bins with type $x$ and threshold $\theta$. Let also $N_x^{(n)}=\sum_\theta N_{x,\theta}^{(n)}$ denotes the number of bins with type $x$.

Note that $U^{(n)}_{x, \theta,0}(0)=N_{x,\theta}^{(n)}$ and $U^{(n)}_{x, \theta,\ell}(0)=0$ for $\ell\geq 1$.
Further, from Assumption~\ref{cond-limit}, $N^{(n)}_{x}/n\rightarrow\mu_x$ and $q^{(n)}_x(\theta)\rightarrow q_x(\theta)$ as $n\rightarrow\infty$, for all $x\in\cX$ and $\theta=0,1,\ldots,d^+_x$. Moreover, by the strong law of large numbers, $N^{(n)}_{x,\theta}/N^{(n)}_x\rightarrow q_x^{(n)}(\theta)$ a.s. as $n\to \infty$.

Consider now the death process as described in Section~\ref{sec:deathFin}.  Let us fix $x\in \cX$ and integers $\theta,r$ with $0\leq r \leq \theta \leq d_x^+$. Consider the $N^{(n)}_{x,\theta}$ bins which starts with $d^+_x$ alive in balls.  For $k= 1,\ldots,N^{(n)}_{x,\theta}$, let $T_k$ be the time that the $(d^+_x-r)$-th ball is removed (killed) in the $k$-th such bin. Then we have $\#\{k:T_k\leq t\}=\sum_{s=0}^{r}U^{(n)}_{x,\theta,s}(t)$. Since the number of remaining balls in any of such bins at time $t$ are i.i.d. random variables with distribution $\Bin(d^+_x,e^{-t})$, then we have
$\PP(T_k\leq t)=\sum_{s=0}^{r}b(d^+_x,e^{-t},s).$ Hence, by using Glivenko-Cantelli theorem,
$$\sup\limits_{t\leq \tau_n}\bigl|\frac{1}{N^{(n)}_{x,\theta}}\sum_{s=r+1}^{d^+_x}U^{(n)}_{x,\theta,s}(t)-\sum_{s=r+1}^{d^+_x}b(d^+_x,e^{-t},s)\bigr|
\top 0, \quad \quad \text{as} \quad N^{(n)}_{x,\theta}\rightarrow\infty.$$
Multiply the above equation by $N^{(n)}_{x,\theta}/N^{(n)}_x$,  and by the law of large numbers, we have
\begin{equation}\label{medi}
  \sup\limits_{t\leq \tau_n}\bigl|\frac{1}{N^{(n)}_x}\sum_{s=r+1}^{d^+_x}U^{(n)}_{x,\theta,s}(t)-q^{(n)}_x(\theta)\sum_{s=r+1}^{d^+_x}b(d^+_x,e^{-t},s)\bigr|
\top 0 , \quad \quad \text{as} \quad N^{(n)}_{x}\rightarrow\infty.
\end{equation}
Moreover, by using Assumption~\ref{cond-limit},
$$\sup\limits_{t\geq 0}\bigl|\frac{N^{(n)}_x}{n}q^{(n)}_x(\theta)\sum_{s=r+1}^{d^+_x}b(d^+_x,e^{-t},s)
-\mu_xq_x(\theta)\sum_{s=r+1}^{d^+_x}b(d^+_x,e^{-t},s)\bigr|
\longrightarrow 0, \quad  \text{as} \quad n\to \infty.$$
By combining the two formulas above and multiplying~\eqref{medi} by $N^{(n)}_x/n$, we obtain
\begin{equation}\label{U_limit}
  \sup\limits_{t\leq \tau_n}\bigl|\frac{1}{n}\sum_{s=r+1}^{d^+_x}U^{(n)}_{x,\theta,s}(t)-\mu_xq_x(\theta)\sum_{s=r+1}^{d^+_x}b(d^+_x,e^{-t},s)\bigr|
\top 0, \quad  \text{as} \quad n\to \infty.
\end{equation}

Consider now $S_{x,\theta,\ell}^{(n)}$, the number of solvent institutions with type $x$, threshold $\theta$ and $\ell=0, \dots, \theta-1$ defaulted neighbors at time $t$. By definition, $S^{(n)}_{x,\theta,\ell}(t)=U^{(n)}_{x,\theta,d^+_x-\ell}(t)$. Hence, by writing~\eqref{U_limit} for $r_1=d^+_x-\ell$ and $r_2=d^+_x-\ell-1$, then taking the difference, we obtain
\begin{align*}
\sup\limits_{t\leq \tau_n}\bigl|\frac{S_{x,\theta,\ell}^{(n)}(t)}{n}-\mu_x q_x(\theta)b\left(d_x^+,1-e^{-t},\ell\right)\bigr|\top 0, \quad  \text{as} \quad n\to \infty.
\end{align*}

Note that the above equation holds for all $x\in\cX$ and $\theta=1, \dots, d^+_x$. Hence, the same convergence also holds for any partial sum over $x$ and $\theta$. In particular,
\begin{equation}\label{bin_limit}
  \sup\limits_{t\leq \tau_n}\bigl|\frac{1}{n}\sum_{\theta=1}^{d^+_x}\sum_{s=d^+_x-\theta+1}^{d^+_x}U^{(n)}_{x,\theta,s}(t)-
  \mu_x\sum_{\theta=1}^{d^+_x}q_x(\theta)\sum_{s=d^+_x-\theta+1}^{d^+_x}b(d^+_x,e^{-t},s)\bigr|
\top 0.
\end{equation}

Let $\cX_K$ be the set of all characteristic $x\in \cX$ such that $d^+_x+d^-_x\leq K$. Since (by Assumption~\ref{cond-average}) $\lambda\in(0,\infty)$ then  for arbitrary small $\varepsilon>0$, there exists $K_{\varepsilon}$ such that $\sum_{x\in\cX\setminus\cX_{K_{\varepsilon}}}\mu_x(d^+_x+d^-_x)<\varepsilon$. Further, by Assumption~\ref{cond-average} and dominated convergence, $$\sum_{x\in\cX\setminus\cX_{K_{\varepsilon}}}(d^+_x+d^-_x)N^{(n)}_x/n\rightarrow
\sum_{x\in\cX\setminus\cX_{K_{\varepsilon}}}(d^+_x+d^-_x)\mu_x<\varepsilon.$$
Hence for $n$ large enough, we have $\sum_{x\in\cX\setminus\cX_{K_{\varepsilon}}}(d^+_x+d^-_x)N^{(n)}_x/n<2\varepsilon$. By \eqref{bin_limit}, we obtain
\begin{align}\label{result}
&\sup\limits_{t\leq \tau_n}\Bigl|\sum_{x\in\cX}(d^+_x+d^-_x)\sum_{\theta=1}^{d^+_x}\sum_{s=d^+_x-\theta+1}^{d^+_x}\bigl(U^{(n)}_{x,\theta,s}(t)/n
  -\mu_xq_x(\theta)b(d^+_x,e^{-t},s)\bigr)\Bigr| \nonumber\\
 \leq& \sup\limits_{t\leq \tau_n}\sum_{x\in\cX_{K_{\varepsilon}}}(d^+_x+d^-_x) \Bigl|\frac{1}{n}\sum_{\theta=1}^{d^+_x}\sum_{s=d^+_x-\theta+1}^{d^+_x}\bigl(U^{(n)}_{x,\theta,s}(t)/n
  -\mu_xq_x(\theta)b(d^+_x,e^{-t},s)\bigr)\Bigr|\nonumber\\
  &+ \sup\limits_{t\leq \tau_n}\sum_{x\in\cX\setminus\cX_{K_{\varepsilon}}} (d^+_x+d^-_x) \Bigl|\frac{1}{n}\sum_{\theta=1}^{d^+_x}\sum_{s=d^+_x-\theta+1}^{d^+_x}\bigl(U^{(n)}_{x,\theta,s}(t)/n
  -\mu_xq_x(\theta)b(d^+_x,e^{-t},s)\bigr)\Bigr|\nonumber\\
  \leq& o_p(1)+\sum_{x\in\cX\setminus\cX_{K_{\varepsilon}}}(d^+_x+d^-_x)(N^{(n)}_x/n+\mu_x)\leq o_p(1)+3\varepsilon.
\end{align}

Since the total number of solvent and defaulted institutions at time $t$ satisfies
$$S_n(t)=\sum_{x\in\cX}\sum_{\theta=1}^{d^+_x}\sum_{s=d^+_x-\theta+1}^{d^+_x}U^{(n)}_{x,\theta,s}(t),$$
which is dominated by $\sum_{x\in\cX}(d^+_x+d^-_x)\sum_{\theta=1}^{d^+_x}\sum_{s=d^+_x-\theta+1}^{d^+_x}U^{(n)}_{x,\theta,s}(t)$. Then, by \eqref{U_limit}, \eqref{result} and letting $\varepsilon\to 0$, we obtain
$$\sup\limits_{t\leq \tau_n}\bigl|\frac{S_n(t)}{n}-f_S(e^{-t})\bigr|\top 0.$$
Further, from $D_n(t)=n-S_n(t)$, the number of defaulted institutions at time $t$ satisfies
$$ \sup\limits_{t\leq \tau_n}\bigl|\frac{D_n(t)}{n}-f_D(e^{-t})\bigr|\top 0.$$

The total number of healthy in links at time $t$ is given by
 given by
$$H^+_n(t)=\sum_{x\in\cX}\sum_{\theta=1}^{d^+_x}\sum_{s=d^+_x-\theta+1}^{d^+_x}sU^{(n)}_{x,\theta,s}(t),$$
which is again dominated by $\sum_{x\in\cX}(d^+_x+d^-_x)\sum_{\theta=1}^{d^+_x}\sum_{s=d^+_x-\theta+1}^{d^+_x}U^{(n)}_{x,\theta,s}(t)$. Then, by \eqref{U_limit}, \eqref{result} and letting $\varepsilon\to 0$, we obtain
$$\sup\limits_{t\leq \tau_n}\bigl|\frac{H^+_n(t)}{n}-f_{H^+}(e^{-t})\bigr|\top 0.$$
Moreover, the total number of infected in links at time $t$ satisfies $I^+_n(t)=L_n(t)-H^+_n(t)$, where $L_n(t)$ denotes the total number of links in the network (also the total number of alive in balls) at time $t$. Since at each interaction we remove 1 link, and from Assumption~\ref{cond-average}, $\sup_{t\geq 0}\bigl|L_n(t)/n - \lambda e^{-t}\bigr| \top 0$. We thus obtain (by definition $f_{I^+}(z) = \lambda z - f_{H^+}(z)$)
$$\sup\limits_{t\leq \tau_n}\bigl|\frac{I^+_n(t)}{n}-f_{I^+}(e^{-t})\bigr|\top 0.$$

Finally, the total number of white (out) balls at time $t$ satisfies $W_n(t)=L_n(t)-H^-_n(t)$, where $H_n^-(t)$ denotes the total number of healthy (out) balls at time $t$, given by
$$H^-_n(t)=\sum_{x\in\cX}\sum_{\theta=1}^{d^+_x}\sum_{s=d^+_x-\theta+1}^{d^+_x}d^-_xU^{(n)}_{x,\theta,s}(t).$$
This is again dominated by $\sum_{x\in\cX}(d^+_x+d^-_x)\sum_{\theta=1}^{d^+_x}\sum_{s=d^+_x-\theta+1}^{d^+_x}U^{(n)}_{x,\theta,s}(t)$, then \eqref{U_limit}, \eqref{result} and Assumption~\ref{cond-average} implies that the number of white balls satisfies
\begin{align*}
\sup\limits_{t\leq \tau_n}\bigl|\frac{W_n(t)}{n}-f_W(e^{-t})\bigr|\top 0.
\end{align*}
This completes the proof.

\subsection{Proof of Lemma~\ref{tau}}\label{sec:lemStop}
Recall that $$f_{W}(z):= \lambda z - \sum_{x\in \cX} \mu_x d_x^-\sum_{\theta=1}^{d_x^+}q_x(\theta) \beta\bigl(d_x^+,z, d_x^+-\theta+1\bigr),$$
and, $\hz:=\sup\bigl\{z\in[0,1]: f_W(z)=0\bigr\}$.  By the initial condition $q_x(\theta)>0$ for some $x\in\cX$. Hence, we have $f_W(1)>0$ and $\zst<1$.
Let us take a constant $t_1>0$ such that $t_1<-\ln\zst$. By continuity of $f_W(z)$ on $[0,1]$, it follows that $f_W(z)>0$ on $\left(\zst,1\right]$. Hence, there exists some constant $C_1>0$ such that $f_W(z)>C_1$ for $t<t_1$.

Since $W_n(\tau_n^{\star})=-1$, if $\tau_n^{\star}\leq t_1$ then $W_n(\tau_n^{\star})/n-f_W(\tau_n^{\star})<-C_1$ for $n$ large enough. But on the other hand,  we have
$$\sup\limits_{t\leq \tau_n^{\star}}|\frac{W_n(t)}{n}-f_{W}(e^{-t}))|\top 0.$$
This is a contradiction. Therefore, we must have
$\mathbb{P}(\tau_n^{\star}\leq t_1)\longrightarrow 0$, as $n\rightarrow\infty$.
In the case $\zst=0$, we can take any finite $t_1$, which implies that $\tau_n^{\star}\top\infty$.

We now consider the case $\hz\in\left(0,1\right]$. Let $\varepsilon>0$ small enough and fix the constant $t_2\in\left(-\ln\zst,-\ln(\zst-\varepsilon)\right)$. By using a similar argument and given the assumption $f_W'(\zst)>0$, we can show there exists some constant $C_2>0$, such that $W_n(\tau_n^{\star})/n-f_W(\tau_n^{\star})\geq C_2$ if $\tau_n^{\star}\geq t_2$. Thus $\mathbb{P}(\tau_n \geq t_2)\longrightarrow 0$ as $n\rightarrow\infty$. Since $t_1$ and $t_2$ are arbitrary, tending both $t_1$ and $t_2$ to $-\ln\zst$, implies that $\tau^\star_n\top-\ln\hz$. This completes the proof of lemma.

\subsection{Proof of Theorem~\ref{thm:LLN}}\label{sec:thmLLN}
By Theorem~\ref{centrality}, it follows that
$$\frac{S_n(\tst)}{n}=f_S(e^{-\tst})+o_p(1).$$

If $\zst=0$ then, by Lemma~\ref{tau},  $\tau_n^{\star}\top\infty$. So $e^{-\tst}\top 0$ and, since $f_S(0)=0$, it follows by the continuity of $f_S$ that $f_S(e^{-\tst})\top 0$. We therefore have
$S_n(\tst)=o_p(n)$.
This implies that $|\cD_{n-1}|=n-S_n(\tst)=n-o_p(n)$, as desired, asymptotically almost all institutions default.

If $\hz\in\left(0,1\right]$ and $f_W'(\hz)>0$, then by Lemma~\ref{tau}, we have $e^{-\tst}\top\zst$. Moreover, the continuity of $f_D$ implies that $f_D(e^{-\tst})\top f_D(\zst)$. Hence, we have by Theorem~\ref{centrality} that
$$\frac{|\cD_{n-1}|}{n}=\frac{D_n(\tst)}{n}\top f_D(\zst).$$
Now using the first statement of Theorem~\ref{centrality} and the continuity of $\mu_xq_x(\theta)b(d^+_x,1-z,\ell)$ on $z$, we obtain
for all $x\in\cX, \theta=1, \dots, d_x^+$ and $\ell=0, \dots, \theta-1$, the final fraction of solvent institutions with type $x$, threshold $\theta$ and $\ell$ defaulted neighbors satisfy
\begin{align*}
\frac{S_{x,\theta,\ell}^{(n)}}{n} \top \mu_x q_x(\theta)b\left(d_x^+,1-\hz,\ell\right).
\end{align*}
This completes the proof of Theorem~\ref{thm:LLN}.

\subsection{Proof of Theorem~\ref{thm-CLT-S}}\label{sec-CLT-S}
Recall that $U^{(n)}_{x, \theta,s}(t)$ denotes the number of bins (institutions) with type $x\in \cX$, threshold $\theta$ and $s$ alive (in-) balls at time $t$. Further, we let $V_{x,\theta,s}^{(n)}(t)$ denotes the number of bins (institutions) with type $x\in \cX$, threshold $\theta$ and at least $s$ alive balls at time $t$, so that  $V_{x,\theta,\ell}^{(n)}(t)=\sum_{s \geq \ell} U_{x,\theta,s}^{(n)}(t)$.

We first study the stochastic process $V_{x,\theta,s}^{(n)}$ for a given $x\in \cX$ and integers $\theta, s$.  Since $V_{x,\theta,s}^{(n)}$ changes by -1 when one of the alive (in) balls in $U_{x,\theta,s}^{(n)}$ bins dies, and there are $sU_{x,\theta,s}^{(n)}(t)$ such balls at time $t$, we obtain
$$dV_{x,\theta,s}^{(n)}(t)=-sU_{x,\theta,s}^{(n)}(t)dt+d\m'_t,$$
where $\m'$ is a martingale.

We define further $\widehat{V}_{x,\theta,s}^{(n)}(t):=e^{st}V_{x,\theta,s}^{(n)}(t)$. Then by Ito's lemma,
\begin{align*}
  d\widehat{V}_{x,\theta,s}^{(n)}(t) & =se^{st}V_{x,\theta,s}^{(n)}(t)dt+e^{st}dV_{x,\theta,s}^{(n)}(t) \\
   & = se^{st}V_{x,\theta,s}^{(n)}(t)dt-se^{st}U_{x,\theta,s}^{(n)}(t)dt+e^{st}d\m'_t\\
   & = se^{-t}\widehat{V}_{x,\theta,(s+1)}^{(n)}(t)dt+d\m_t,
\end{align*}
where $d\m_t=e^{st}d\m'_t$ is also a martingale differential. Thus
\begin{equation}\label{defM}
  M_{x,\theta,s}^{(n)}(t):=\widehat{V}_{x,\theta,s}^{(n)}(t)-s\int_{0}^{t}e^{-r}\widehat{V}_{x,\theta,(s+1)}^{(n)}(r)dr
\end{equation}
is also a martingale for every $0\leq s\leq d^+_x$. We can calculate its quadratic variation by
\begin{align}\label{crochet_Mx}
  \bigl[M_{x,\theta,s}^{(n)},M_{x,\theta,s}^{(n)}\bigr]_t =\sum_{0<r\leq t}|\Delta M_{x,\theta,s}^{(n)}(r)|^2=\sum_{0<r\leq t}|\Delta \widehat {V}_{x,\theta,s}^{(n)}(r)|^2
   =\int_{0}^{t}e^{2sr}d(-V_{x,\theta,s}^{(n)}(r)).
\end{align}

Then,
$$\widetilde{M}_{x,\theta,s}^{(n)}(t):=n^{-1/2}\bigl(M_{x,\theta,s}^{(n)}(t)-M_{x,\theta,s}^{(n)}(0)\bigr),$$
is a martingale with initial value at 0. Further, from Theorem~\ref{centrality}, as $n\rightarrow\infty$,
$$\sup\limits_{t\leq\tst}|V_{x,\theta,s}^{(n)}(t)/n-\mu_xq_x(\theta)\beta(d^+_x,e^{-t},s)|\top 0.$$

For notational convenience, let us denote by
$$\varphi_{x,\theta,s}(y):=\mu_xq_x(\theta)\beta(d^+_x,y,s).$$
Using integration by parts, we obtain for all $t\leq\tst$,
\begin{align*}
  \bigl[\widetilde{M}_{x,\theta,s}^{(n)},\widetilde{M}_{x,\theta,s}^{(n)}\bigr]_t & =n^{-1}\bigl[M_{x,\theta,s}^{(n)},M_{x,\theta,s}^{(n)}\bigr]_t \\
   & = n^{-1}\bigl(V_{x,\theta,s}^{(n)}(0)-e^{2st}V_{x,\theta,s}^{(n)}(t)+\int_{0}^{t}2sV_{x,\theta,s}^{(n)}(r)e^{2sr}dr\bigr) \\
   & = \varphi_{x,\theta,s}(1)-e^{2st}\varphi_{x,\theta,s}(e^{-t})+\int_{0}^{t}2s\varphi_{x,\theta,s}(e^{-r})e^{2sr}dr+o_p(1) \\
   & =\int_{0}^{t}e^{2sr}d(-\varphi_{x,\theta,s}(e^{-r}))+o_p(1).
\end{align*}

Moreover, by \eqref{crochet_Mx}, there exists a constant $C$ such that
\begin{equation}\label{crochet_Mtilde}
  \bigl[\widetilde{M}_{x,\theta,s}^{(n)},\widetilde{M}_{x,\theta,s}^{(n)}\bigr]_t=n^{-1}\bigl[M_{x,\theta,s}^{(n)},M_{x,\theta,s}^{(n)}\bigr]_t\leq n^{-1}e^{2st}V_{x,\theta,s}^{(n)}(0)\leq Ce^{2st},
\end{equation}
which holds uniformly for $n$. Consequently, by Theorem~\ref{martingale_clt}, for the stopped process
$$\widetilde{M}_{x,\theta,s}^{(n)}(t\wedge\tau_n)\tod{Y}_{x,\theta,s}(t\wedge t_0)\quad \quad \text{in} \quad \cD\left[0,\infty\right),$$
where ${Y}_{x,\theta,s}$ is a continuous Gaussian process with $\mathbb{E}[{Y}_{x,\theta,s}(t)]=0$ and covariance
$$\mathbb{E}\left[{Y}_{x,\theta,s}(t){Y}_{x,\theta,s}(u)\right]=\int_{0}^{t}e^{2sr}d(-\varphi_{x,\theta,s}(e^{-r})), \quad\quad 0\leq t\leq u<\infty.$$

Furthermore, for $s\neq r$, we can show that $V_{x,\theta,r}^{(n)}$ and $V_{x,\theta,s}^{(n)}$ never change simultaneously, almost surely. Thus, $[\widetilde{M}_{x,\theta,r}^{(n)},\widetilde{M}_{x,\theta,s}^{(n)}]_t=0$.

Hence, by Theorem~\ref{martingale_clt} applied to the vector-valued martingale $\bigl(\widetilde{M}_{x,\theta,s}^{(n)}\bigr)_{s=0, \dots, d^+_x}$, the convergence holds jointly with a diagonal covariance matrix for $\bigl(Y_{x,\theta,s}\bigr)_{s=0, \dots, d^+_x}$, which implies that the processes $Y_{x,\theta,0}, \dots, Y_{x,\theta, d^+_x}$ are all independent.

As for two different types-thresholds $(x,\theta)$ and $(x',\theta')$, the independence follows since for any $s=0, \dots, d^+_x$ and $s'=0, \dots, d^+_{x'}$, $V_{x,\theta,s}^{(n)}$ and $V_{x'\theta',s'}^{(n)}$ also a.s. never change simultaneously. This coukld also be viewed from the nature of our balls and bins representation: the balls die independently and a.s. never die at the same moment. Hence, the death processes in bins with different types are independent.

We now compute $\widehat{V}_{x,\theta,s}^{(n)}$, using the definition Equation~\eqref{defM} for $M_{x,\theta,s}^{(n)}(t)$, repeatedly. We find that for $s=d^+_x$,
$$\widehat{V}_{x,\theta,d^+_x}^{(n)}(t)=M_{x,\theta,d^+_x}^{(n)}(t),$$
and for $s=d^+_x-1$,
$$\widehat{V}_{x,\theta,s}^{(n)}(t)=M_{x,\theta,s}^{(n)}(t)+s\int_{0}^{t}e^{-r}M_{x,\theta,(s+1)}^{(n)}(r)dr.$$
Then for $s=d^+_x-2$, we find out
\begin{align*}\begin{aligned}
  \widehat{V}_{x,\theta,s}^{(n)}(t) =&M_{x,\theta,s}^{(n)}(t)+s\int_{0}^{t}e^{-r}M_{x,\theta,(s+1)}^{(n)}(r)dr \\
   & +\int_{r_2<r_1<t}s(s+1)e^{-(r_1+r_2)}M_{x,\theta,(s+2)}^{(n)}(r_2)dr_2dr_1 \\
   =& M_{x,\theta,s}^{(n)}(t)+s\int_{0}^{t}e^{-r}M_{x,\theta,(s+1)}^{(n)}(r)dr\\
   & + \int_{0}^{t}s(s+1)(e^{-r}-e^{-t})e^{-r}M_{x,\theta,(s+2)}^{(n)}(r)dr\\
  =& M_{x,\theta,s}^{(n)}(t)+\sum_{j=s+1}^{d^+_x}s\binom{j-1}{s}\int_{0}^{t}(e^{-r}-e^{-t})^{j-s-1}
  e^{-r}M_{x,\theta,j}^{(n)}(r)dr.
\end{aligned}\end{align*}

Assume that the above formula holds for $0<s\leq d^+_x-1$. Then for $s-1$, we deduce
\begin{align*}\begin{aligned}
  \widehat{V}_{x,\theta,s-1}^{(n)}(t) =&M_{x,\theta,s-1}^{(n)}(t)+s\int_{0}^{t}e^{-r}\widehat{V}_{x,\theta,s}^{(n)}(r)dr \\
   =& M_{x,\theta,s-1}^{(n)}(t)+\sum_{j=s+1}^{d^+_i}s\binom{j-1}{s}\int_{0}^{t}\int_{0}^{t_1}(s-1)
   e^{-t_1}(e^{-r}-e^{-t_1})^{j-s-1}e^{-r}M_{x,\theta,j}^{(n)}(r)drdr_1\\
   =& M_{x,\theta,s-1}^{(n)}(t)+(s-1)\int_{0}^{t}e^{-r}M_{x,\theta,s}^{(n)}(r)dr \\
   &+\sum_{j=s+1}^{d^+_x}\frac{s(s-1)}{j-s}\binom{j-1}{s}\int_{0}^{t}(e^{-r}-e^{-t})^{j-s}
  e^{-r}M_{x,\theta,j}^{(n)}(r)dr\\
  =& M_{x,\theta,s-1}^{(n)}(t)+\sum_{j=s}^{d^+_x}(s-1)\binom{j-1}{s-1}\int_{0}^{t}(e^{-r}-e^{-t})^{j-s}
  e^{-r}M_{x,\theta,j}^{(n)}(r)dr.
\end{aligned}\end{align*}
By induction, we obtain that
$$\widehat{V}_{x,\theta,s}^{(n)}(t)= M_{x,\theta,s}^{(n)}(t)+\sum_{j=s+1}^{d^+_x}s\binom{j-1}{s}\int_{0}^{t}(e^{-r}-e^{-t})^{j-s-1}
  e^{-r}M_{x,\theta,j}^{(n)}(r)dr.$$

We next define $ \widehat{v}_{x,\theta,s}^{(n)}(t)$  for all $t\geq 0$, as the conditional expectation of $\widehat{V}_{x,\theta,s}^{(n)}(t)$ given its initial value $V_{x,\theta,s}^{(n)}(0)$. Namely, we have
\begin{align*}
  \widehat{v}_{x,\theta,s}^{(n)}(t):&=\mathbb{E}\bigl[\widehat{V}_{x,\theta,s}^{(n)}(t)|V_{x, \theta,s}^{(n)}(0)\bigr] \\
  & =M_{x,\theta,s}^{(n)}(0)+\sum_{j=s+1}^{d^+_x}s\binom{j-1}{s}\int_{0}^{t}(e^{-r}-e^{-t})^{j-s-1}
  e^{-r} M_{x,\theta,j}^{(n)}(0)dr.
\end{align*}

Note that by definition, $\mathbb{E}\bigl[\widehat{V}_{x,\theta,s}^{(n)}(t)\bigr]=e^{st}\mathbb{E}[V_{x,\theta,s}^{(n)}(t)]$. Further, $V_{x,\theta,s}^{(n)}(t)$ is the number of bins of type $x$, threshold $\theta$ and with at least $s$ balls at time $t$ in the death process where balls die independently with rate 1 (without stopping). Then at time $t$, each of such bins has the binomial probability $\beta(d^+_x,e^{-t},s)$ to have at least $s$ alive balls remaining. The initial number is $$V_{x,\theta,s}^{(n)}(0)=N^{(n)}_{x,\theta},$$
for all $s=0,\dots, d^+_x$. Consequently,
\begin{equation*}
  \widehat{v}_{x,\theta,s}^{(n)}=e^{st}N^{(n)}_{x,\theta}\beta(d^+_x,e^{-t},s).
\end{equation*}

We further define for $t\leq\tst$,
$$\widetilde{V}_{x,\theta,s}^{(n)}(t):=n^{-1/2}\bigl(\widehat{V}_{x,\theta,s}^{(n)}(t)-\widehat{v}_{x,\theta,s}^{(n)}\bigr).$$
It is clear that,
$$\widetilde{V}_{x,\theta,s}^{(n)}(t)= \widetilde{M}_{x,\theta,s}^{(n)}(t)+\sum_{j=s+1}^{d^+_x}s\binom{j-1}{s}\int_{0}^{t}(e^{-r}-e^{-t})^{j-s-1}
  e^{-r}\widetilde{M}_{x,\theta,j}^{(n)}(r)dr.$$
We can then apply Theorem~\ref{martingale_clt} to the above finite sum and take the limit (in distribution) under the summation sign. It follows that
\begin{align*}
\widetilde{V}_{x,\theta,s}^{(n)}(t\wedge\tau_n)\tod \widetilde{\cZ}_{x,\theta,s}(t\wedge t_0),
\end{align*}
in $\cD\left[0,\infty\right)$, where
$$\widetilde{\cZ}_{x,\theta,s}(t):=Y_{x,\theta,s}(t)+\sum_{j=s+1}^{d^+_x}s\binom{j-1}{s}\int_{0}^{t}(e^{-r}-e^{-t})^{j-s-1}
  e^{-r}Y_{x,\theta,j}(r)dr.$$

Note that, although the initial $V_{x,\theta,s}^{(n)}(0)$ is random,
by the standard central limit theorem applied to $N^{(n)}_{x}=n\mu^{(n)}_x$ i.i.d. random variables $\ind{\{\Theta^{(n)}_i=\theta\}}$, we have as $n\rightarrow\infty$,
\begin{align*}
  n^{-1/2}\Bigl(N^{(n)}_{x,\theta}-n\mu^{(n)}_xq^{(n)}_x(\theta)\Bigr)\tod\cZ_{x,\theta},
\end{align*}
where $\cZ_{x,\theta}\sim N(0,\mu_xq_x(\theta)(1-q_x(\theta)))$ and $\Cov(\cZ_{x,\theta_1},\cZ_{x,\theta_2})=-\mu_xq_x(\theta_1)q_x(\theta_2)$ for $\theta_1\neq\theta_2$. We denote by (for all $\theta_1, \theta_2$)
$$\psi_{x,\theta_1,\theta_2}:=\Cov(\cZ_{x,\theta_1},\cZ_{x,\theta_2}).$$

Notice that for all $x_1\neq x_2$, $\cZ_{x_1,\theta_1}, \cZ_{x_2,\theta_2}$ are independent and $\Cov(\cZ_{x_1,\theta_1},\cZ_{x_2,\theta_2})=0$. Thus we have jointly for all triple $(x,\theta,s)$,
\begin{equation}\label{clt_initial}
  n^{-1/2}\Bigl(\widehat{v}_{x,\theta,s}^{(n)}-ne^{st}\mu^{(n)}_xq_x^{(n)}(\theta)\beta(d^+_x,e^{-t},s)\Bigr)\tod \widehat{\cZ}_{x,\theta,s}(t),
\end{equation}
in $\mathcal{C}\interval$ as $n\rightarrow\infty$, where $\widehat{\cZ}_{x,\theta,s}(t)$ is also a Gaussian process with mean 0 and covariance
$$\EE\bigl[\widehat{\cZ}_{x,\theta,s}(t)\widehat{\cZ}_{x,\theta,s}(u)\bigr]=
e^{(t+u)s}\beta(d^+_x,e^{-t},s)\beta(d^+_x,e^{-u},s)\psi_{x,\theta,\theta}.$$
Moreover, $\widehat{\cZ}_{x,\theta,s}$ and $\widetilde{\cZ}_{x,\theta,s}$ are independent.

The above argument shows that $\widetilde{V}_{x,\theta,s}^{(n)}$ converges to a Gaussian process. We next define
$$\widetilde{\widehat{V}}_{x,\theta,s}^{(n)}(t):=n^{-1/2}\Bigl(\widehat{V}_{x,\theta,s}^{(n)}(t)
-ne^{st}\mu^{(n)}_xq_x^{(n)}(\theta)\beta(d^+_x,e^{-t},s)\Bigr).$$
By \eqref{clt_initial}, we obtain that
\begin{equation*}
  \widetilde{\widehat{V}}_{x,\theta,s}^{(n)}(t\wedge\tau_n)\tod\widetilde{\widehat{\cZ}}_{x,\theta,s}(t\wedge t_0),
\end{equation*}
where $\widetilde{\widehat{\cZ}}_{x,\theta,s}$ is a Gaussian process with mean 0 and covariance
$$\mathbb{E}\bigl[\widetilde{\widehat{\cZ}}_{x,\theta,s}(t)\widetilde{\widehat{\cZ}}_{x,\theta,s}(u)\bigr]=\EE\bigl[\widehat{\cZ}_{x,\theta,s}(t)
\widehat{\cZ}_{x,\theta,s}(u)\bigr]
+\mathbb{E}\bigl[\widetilde{\cZ}_{x,\theta,s}(t)\widetilde{\cZ}_{x,\theta,s}(u)\bigr],$$
for all $0\leq t\leq u<\infty$.

Next, we define for all triple $(x,\theta,s)$,
$$V_{x,\theta,s}^{*(n)}(t):=e^{-st}\widetilde{\widehat{V}}_{x,\theta,s}^{(n)}(t),\quad \cZ'_{x,\theta,s}(t):=e^{-st}\widetilde{\cZ}_{x,\theta,s}(t), $$
and
$$\cZ^{\ast}_{x,\theta,s}(t):=e^{-st}\widehat{\cZ}_{x,\theta,s}(t),\quad \cZ_{x,\theta,s}(t):=e^{-st}\widetilde{\widehat{\cZ}}_{x,\theta,s}(t).$$
Then we have
\begin{equation}\label{conver_X}
  V_{x,\theta,s}^{*(n)}(t\wedge\tau_n)\tod \cZ_{x,\theta,s}(t\wedge t_0).
\end{equation}

We let further
$$\bigsigma^{\ast}_{x,\theta_1,\theta_2, r, s}(y):=\Cov\bigl((\cZ^{\ast}_{x,\theta_1,r}(-\ln y),\cZ^{\ast}_{x,\theta_2,s}(-\ln y))\bigr),$$
and,
$$\bigsigma_{x,\theta,r,s}(y):=\mbox{Cov}\bigl(\cZ'_{x,\theta,r}(-\ln y),\cZ'_{x,\theta,s}(-\ln y)\bigr).$$

By using all the independence and covariance formulas above, it follows that
\begin{equation}\label{relation_Cov}
\begin{split}
  \Cov\bigl(\cZ_{x,\theta,s_1}(t),\cZ_{x,\theta,s_2}(t)\bigr)=& \bigsigma^{\ast}_{x,\theta,\theta,s_1,s_2}(e^{-t}) +\bigsigma_{x,\theta,s_1,s_2}(e^{-t}), \\
   \Cov\bigl(\cZ_{x,\theta_1,s_1}(t),\cZ_{x,\theta_2,s_2}(t)\bigr)= & \bigsigma^{\ast}_{x,\theta_1,\theta_2,s_1,s_2}(e^{-t}),\quad \text{for all} \quad \theta_1\neq\theta_2,\\
  \Cov\bigl(X_{x_1,\theta_1,s_1}(t),X_{x_2,\theta_2,s_2}(t)\bigr)= & 0,\quad \text{ for all} \quad x_1\neq x_2,
\end{split}
\end{equation}
where
\begin{align}\label{sigma_ast}
  & \bigsigma^{\ast}_{x,\theta_1,\theta_2,s_1,s_2}(e^{-t})= \beta(d^+_x,e^{-t},s_1)\beta(d^+_x,e^{-t},s_2)\psi_{x,\theta_1,\theta_2} , \nonumber\\
  & \psi_{x,\theta,\theta}= \mu_xq_x(\theta)(1-q_x(\theta)), \quad \psi_{x,\theta_1,\theta_2}= -\mu_xq_x(\theta_1)q_x(\theta_2)\quad \text{for all} \quad  \theta_1\neq\theta_2.
\end{align}

We now compute $\bigsigma_{x,\theta,r,s}(y)$. Recall that
$$\mbox{Cov}(Y_{x,\theta,s}(-\ln y),Y_{x,\theta,r}(-\ln y))=\ind_{r= s}\int_{y}^{1}u^{-2s}d\varphi_{x,\theta,s}(u).$$
Then we obtain
\begin{align*}
 \mbox{Var}\bigl(\widetilde{\cZ}_{x,\theta,s}(-\ln y)\bigr)
 =&  \int_{y}^{1}v^{-2s}d\varphi_{x,\theta,s}(v)\nonumber\\
 +\sum_{j=s+1}^{d^+_x}s^2\binom{j-1}{s}^2& \int\int_{y<u<z<1}(u-y)^{j-s-1}(z-y)^{j-s-1}
 \Bigl(\int_{z}^{1}v^{-2j}d\varphi_{x,\theta,j}(v)\Bigr)dudz \nonumber\\
=&\frac{1}{2}\sum_{j=s}^{d^+_x}\binom{j-1}{s-1}^2\int_{y}^{1}(v-y)^{2j-2s}v^{-2j}d\varphi_{x,\theta,j}(v).
\end{align*}
For $r\geq s$, we can write $r=s+k$ for some $k\geq1$, and deduce that
\begin{equation*}
\begin{split}
 \mbox{Cov} & \bigl(\widetilde{\cZ}_{x,\theta,s}(-\ln y),\widetilde{\cZ}_{x,\theta,s+k}(-\ln y)\bigr)\\
 & =\frac{1}{2}\sum_{j=s+k}^{d^+_x}\binom{j-1}{s-1}\binom{j-1}{s+k-1}\int_{y}^{1}
 (v-x)^{2j-2s-k}v^{-2j}d\varphi_{x,\theta,j}(v).
  \end{split}
\end{equation*}
Hence we have
\begin{equation}\label{sigma sr}
\begin{split}
\bigsigma_{x,\theta,s,s+k}(y): &=y^{2s+k}\mbox{Cov}\bigl(\widetilde{\cZ}_{x,\theta,s}(-\ln y),\widetilde{\cZ}_{x,\theta,s+k}(-\ln y)\bigr)\\
& = \frac{1}{2}y^{2s+k}\sum_{j=s+k}^{d^+_x}\binom{j-1}{s-1}\binom{j-1}{s+k-1}\int_{y}^{1}
(v-y)^{2j-2s-k}v^{-2j}d\varphi_{x,\theta,j}(v).
\end{split}
\end{equation}

For the number of solvent institutions with type $x\in \cX$, threshold $\theta=1,\dots, d_x^+-1$ and $\ell=1, \dots, \theta-1$ defaulted neighbors at time $t$, we have
$$S^{(n)}_{x,\theta,\ell}(t)=V_{x,\theta,d^+_{x}-\ell}^{(n)}-V^{(n)}_{x,\theta,d^+_x-\ell+1}.$$
Moreover, for $\ell=0$, $S^{(n)}_{x,\theta,0}(t)=V_{x,\theta,d^+_x}^{(n)}$. Using the joint asymptotic normality of $V^{*(n)}_{x,\theta,d_x^+-\ell}$ and $V^{*(n)}_{x,\theta,d_x^+-\ell+1}$, we obtain that in $\cD\interval$,
\begin{align*}
n^{-1/2}\left(S^{(n)}_{x,\theta,\ell}(t\wedge\tau_n)-n\mu^{(n)}_xq_x^{(n)}(\theta)b\bigl(d_x^+,1-e^{-(t\wedge \tau_n)},\ell\bigr)\right)\tod Z_{x,\theta,\ell}(t\wedge t_0),
\end{align*}
where $Z_{x,\theta,\ell}=\cZ_{x,\theta,d_x^+-\ell}-\cZ_{x,\theta,d_x^+-\ell+1}$ for $\ell\geq 1$ and $Z_{x,\theta,0}=\cZ_{x,\theta,d^+_x}$. For convenience, we set a convention that $\cZ_{x,\theta,s}= 0$ if $s>d^+_x$.

Thus for any two triple $(x_1,\theta_1,\ell_1)$ and $(x_2,\theta_2,\ell_2)$,
\begin{align}\begin{aligned}\label{Cov_sll}
\Cov\bigl(Z_{x_1,\theta_1,\ell_1}(t),Z_{x_2,\theta_2,\ell_2}(t)\bigr) =& \Cov\bigl(\cZ_{x_1,\theta_1,d_{x_1}^+-\ell_1}(t),\cZ_{x_2,\theta_2,d_{x_2}^+-\ell_2}(t)\bigr)\\
&+\Cov\bigl(\cZ_{x_1,\theta_1,d_{x_1}^+-\ell_1+1}(t),\cZ_{x_2,\theta_2,d_{x_2}^+-\ell_2+1}(t)\bigr)\\
& - \Cov\bigl(\cZ_{x_1,\theta_1,d_{x_1}^+-\ell_1}(t),\cZ_{x_2,\theta_2,d_{x_2}^+-\ell_2+1}(t))\\
& -\Cov\bigl(\cZ_{x_1,\theta_1,d_{x_1}^+-\ell_1+1}(t),\cZ_{x_2,\theta_2,d_{x_2}^+-\ell_2}(t)\bigr), \end{aligned}\end{align}
where the covariances on the right hand side could be calculated by \eqref{relation_Cov}, \eqref{sigma_ast} and \eqref{sigma sr}.

Moreover, the variance of $Z_{x,\theta,\ell}(t)$ is given by
\begin{align}\label{Var_sll}
\bigsigma_{x,\theta,\ell}^2(t)   = & \bigsigma_{x,\theta,d_{x}^+-\ell, d_{x}^+-\ell}(e^{-t})+
  \bigsigma_{x,\theta,d_{x}^+-\ell+1,d_{x}^+-\ell+1}(e^{-t})\nonumber \\
  &-2\bigsigma_{x,\theta,d_{x}^+-\ell, d_{x}^+-\ell+1}(e^{-t})+b^2(d^+_x,e^{-t},d^+_x-\ell)\psi_{x,\theta,\theta},
\end{align}
where $\psi_{x,\theta,\theta}=\mu_xq_x(\theta)(1-q_x(\theta))$. This completes the proof of Theorem~\ref{thm-CLT-S}.

\subsection{Proof of Theorem~\ref{normality}}\label{sec:normality}
We first show that the moment regularity condition (i.e., Assumption~\ref{cond:CLT}) guarantees that, as $n\rightarrow\infty$, we have not only $f^{(n)}_{\clubsuit}(z)\rightarrow f_{\clubsuit}(z)$ for all $\clubsuit\in\{S, D, H^+, I^+, W\}$, but also we have these convergences together with all their derivatives, uniformly on $[0,1]$.

\begin{lemma}\label{lem-remak}
The Assumption~\ref{cond:CLT} guarantees that, as $n\rightarrow\infty$, $f^{(n)}_{\clubsuit}(z)\rightarrow f_{\clubsuit}(z)$, for all $\clubsuit\in\{S, D, H^+, I^+, W\}$, together with all their derivatives, uniformly on $[0,1]$.
\end{lemma}
\begin{proof}
We only provide the proof for $f_{H^+}$. The proof for the case $\clubsuit\in\{S, D, I^+, W\}$ follows in the same way. Let us consider the following function
$$h((v,y);z):=\sum_{\ell=v-y+1}^{v}\ell b(v,z,\ell).$$
We define a sequence of bi-dimensional nonnegative integer valued random variables $\{X_n\}$ and $X$ with distribution respectively,
$$\PP(X_n=(x,y))=\sum_{x\in\cA_{v}}\mu^{(n)}_{x}q^{(n)}_x(y),$$
and
$$\PP(X=(x,y))=\sum_{x\in\cA_{v}}\mu_{x}q_x(y),$$
where $\cA_v:=\{x\in\cX: d^+_x=v\}$. Then it follows that
$$f^{(n)}_{H^+}(z)=\EE h(X_n;z),$$
$$f_{H^+}(z)=\EE h(X;z).$$
By Assumption~\ref{cond-average}, we have $X_n\rightarrow X$ in distribution as $n\rightarrow\infty$. Moreover, for any $z\in[0,1]$, $0\leq h((v,y);z)\leq v$. Thus, $h(X_n)\leq X^{(1)}_n$, where $X^{(1)}_n$ is first dimensional component of $X_n$. Note also that by Assumption~\ref{cond:CLT}, $X^{(1)}_n$ is uniformly integrable. Hence we have (as $n\to \infty$) for all $z\in[0,1]$,
$$f^{(n)}_{H^+}(z)=\EE h(X_n;z)\rightarrow \EE h(X;z)=f_{H^+}(z).$$
Further, an elementary calculation gives that
\begin{equation*}
  \frac{d}{dz}b(v,z,\ell)=vb(v-1,z,\ell-1)-vb(v-1,z,\ell).
\end{equation*}
Combining now with the fact that $b(v,z,\ell)\in[0,1]$, we have $|\frac{d}{dz}b(v,z,\ell)|\leq 2v$. In addition, using a simple induction gives that for every $j\geq 0$, $|\frac{d^j}{dz^j}b(v,z,\ell)|\leq (2v)^j$. We therefore obtain that
\begin{equation}\label{bound_derivt}
 \bigl|\frac{d^j}{dz^j}h(X_n;z)\bigr|\leq (2X^{(1)}_n)^j\sum_{\ell=1}^{X^{(1)}_n}\ell\leq (2X^{(1)}_n)^{j+2}.
\end{equation}
This is again, by Assumption~\ref{cond:CLT}, uniformly integrable. Hence, we also have (as $n\to \infty$)
$$\frac{d^j}{dz^j}f^{(n)}_{H^+}(z)=\EE \frac{d^j}{dz^j} h(X_n;z)\rightarrow \EE \frac{d^j}{dz^j}h(X;z)=\frac{d^j}{dz^j}f_{H^+}(z)$$
for all $z\in[0,1]$. Moreover, \eqref{bound_derivt} also implies all these derivatives are uniformly bounded. Thus by the Arzela-Ascoli theorem, as $n\rightarrow\infty$, $f^{(n)}_{H^+}(z)\rightarrow f_{H^+}(z)$ together with all its derivatives uniformly on $[0,1]$. This completes the proof for $\clubsuit=H^+$. 
\end{proof}

We can now start the proof of Theorem~\ref{normality}. Recall that $L_n(t)$ denotes the total number of alive in balls at time $t$. At the initial time, $L_n(0)=n\lambda^{(n)}$ and $L_n(t)$ decreases by 1 each time a (in) ball dies. Since the death happens after an exponential time with rate 1 independently, therefore on $[0,\tst]$, writing in differential form, we have
$$dL_n(t)=-L_n(t)dt+d\m_t,$$
where $\m$ is a martingale. Then by Ito's lemma we have
$$d(e^tL_n(t))=e^tdL_n(t)+e^tL_n(t)dt=e^td\m_t,$$
which implies $\widehat{L}_n(t):=e^tL_n(t)$ is another martingale.
The quadratic variation of $\widehat{L}_n(t)$ is
\begin{align*}
  [\widehat{L}_n,\widehat{L}_n]_t & =\sum_{0<s\leq t}|\Delta \widehat{L}_n(s)|^2=\sum_{0<s\leq t}e^{2s}|\Delta L_n(s)|^2 \\
   &=\sum_{0<s\leq t}e^{2s}(L_n(s)-L_n(s-))=\int_{0}^{t}e^sd(-L_n(s)) \\
   & =-e^{2t}L_n(t)+L_n(0)+\int_{0}^{t}2e^{2s}L_n(s)ds.
\end{align*}
In particular,
$$[\widehat{L}_n,\widehat{L}_n]_t\leq e^{2t}\int_{0}^{t}d(-L_n(s))\leq n\lambda^{(n)}e^{2t}.$$

In addition, as shwon in the proof of Lemma~\ref{centrality}, uniformly on $[0,\tau_n]$ we have
$$\sup\limits_{t\leq \tau_n}\bigl|\frac{L_n(t)}{n}-\lambda e^{-t}\bigr|\top 0.$$
Going back to the quadratic variation, for every $t\in[0,T\wedge\tau_n]$ with $T$ fixed,
\begin{align*}
  [\widehat{L}_n,\widehat{L}_n]_t & =-e^{2t}L_n(0)e^{-t}+L_n(0)+\int_{0}^{t}2e^{2s}L_n(0)e^{-s}ds+o_p(n) \\
   & =L_n(0)(-e^t+1+2e^t-2)+o_p(n) \\
   & =L_n(0)(e^t-1)+o_p(n)\\
   & =\frac{\lambda n}{2}(e^t-1)+o_p(n).
\end{align*}

Next we define
$$\widetilde{L}_n(t):=n^{-1/2}\bigl(\widehat{L}_n(t)-\widehat{L}_n(0)\bigr)=n^{-1/2}\bigl(e^tL_n(t)-L_n(0)\bigr),$$
and the quadratic variation is given by
$$[\widetilde{L}_n,\widetilde{L}_n]_t=n^{-1}[\widehat{L}_n,\widehat{L}_n]_t=\frac{\lambda}{2}(e^t-1)+o_p(1).$$

Let us stop the process at $\tau_n\leq\tst$. By assumption $\tau_n\top t_0$, the quadratic variation of the stopped process converges in probability to $\lambda(e^{t\wedge t_0}-1)/2$. Then for any fixed $t\geq 0$, there exist some constant $C$ such that
$$[\widetilde{L}_n,\widetilde{L}_n]_{t\wedge\tau_n}=n^{-1}[\widehat{L}_n,\widehat{L}_n]_{t\wedge\tau_n}\leq \lambda^{(n)}e^{2t} \leq Ce^{2t},$$
which holds uniformly for $n$. Thus, from Theorem~\ref{martingale_clt} for the stopped process, we have
$$\widetilde{L}_n(t\wedge\tau_n)\tod\widetilde{Z}(t\wedge t_0)\quad \quad \text{in} \quad \cD\left[0,\infty\right),$$
where $\widetilde{Z}$ is a continuous Gaussian process with $\mathbb{E}\widetilde{Z}(t)=0$ and the covariances
$$\mathbb{E}\bigl[\widetilde{Z}(t)\widetilde{Z}(u)\bigr]=\lambda(e^{t}-1)/2, \quad\quad 0\leq t\leq u<\infty.$$

We further define
$$\widetilde{\widehat{L}}_n(t):=e^{-t}\widetilde{L}_n(t), \quad \widehat{Z}(t):=e^{-t}\widetilde{Z}(t),$$
so that
$$\widetilde{\widehat{L}}_n(t\wedge\tau_n)\tod \widehat{Z}(t\wedge t_0)\quad \quad \text{in} \quad \cD\left[0,\infty\right).$$

Let $U_{x,\theta,s}^{(n)}(t)$ and $V_{x,\theta,s}^{(n)}(t)$ be as defined in the proof of Theorem~\ref{sec-CLT-S}. Further, we  denote by $V_{x,s}^{(n)}(t)$ the number of bins that are of type $x$ and with at least $s$ (in) balls at time $t$.

Note that the followings hold:
\begin{equation}\label{Sn}
  S_n(t)=\sum_{x\in\cX}\sum_{\theta=1}^{d^+_x}V_{x,\theta,d^+_x-\theta+1}^{(n)}(t), \quad D_n(t)=n-S_n(t),
\end{equation}
 and,
\begin{align}\begin{aligned}\label{HDn}
  H^+_n(t) & =\sum_{x\in\cX}\sum_{\theta=1}^{d^+_x}\sum_{s=d^+_x-\theta+1}^{d^+_x}sU_{x,\theta,s}^{(n)}(t) \\
   & = \sum_{x\in\cX}\sum_{\theta=1}^{d^+_x}\bigl[(d^+_x-\theta+1)V^{(n)}_{x,\theta,d^+_x-\theta+1}(t)
   +\sum_{s=d^+_x-\theta+2}^{d^+_x}V_{x,\theta,s}^{(n)}(t)\bigr].
\end{aligned}\end{align}
Further, $I^+_n(t)=L_n(t)-H^+_n(t)$, and,
\begin{align}\begin{aligned}\label{ICn}
W_n(t) =& L_n(t)-\sum_{x\in\cX}d^-_x\sum_{\theta=1}^{d^+_x}V_{x,\theta,d^+_x-\theta+1}^{(n)}(t).
\end{aligned}\end{align}

For  $\widetilde{V}_{x,s}^{(n)}(t) =\sum_\theta \widetilde{V}_{x,\theta, s}^{(n)}(t)$, we have
\begin{align*}
\widetilde{V}_{x,s}^{(n)}(t) &= \sum_{\theta=1}^{d^+_x}\widetilde{M}_{x,\theta,s}^{(n)}(t)+\sum_{j=s+1}^{d^+_x}s\binom{j-1}{s}\int_{0}^{t}(e^{-r}-e^{-t})^{j-s-1}
  e^{-r}\sum_{\theta=1}^{d^+_x}\widetilde{M}_{x,\theta,j}^{(n)}(r)dr\\
  &=\widetilde{M}_{x,s}^{(n)}(t)+\sum_{j=s+1}^{d^+_x}s\binom{j-1}{s}\int_{0}^{t}(e^{-r}-e^{-t})^{j-s-1}
  e^{-r}\widetilde{M}_{x,j}^{(n)}(r)dr,
\end{align*}
where $\widetilde{M}_{x,s}^{(n)}(t):=\sum_{\theta=1}^{d^+_x}\widetilde{M}_{x,\theta,s}^{(n)}(t)$. This is again a partial sum and  Theorem~\ref{martingale_clt} applies. We therefore have
$$\widetilde{V}_{x,s}^{(n)}(t\wedge\tau_n)\tod\widetilde{\cZ}_{x,s}(t\wedge t_0):= \sum_{\theta=1}^{d^+_x}\widetilde{\cZ}_{x,\theta,s}(t\wedge t_0),$$
in $\cD\left[0,\infty\right)$. More precisely,
$$\widetilde{M}_{x,s}^{(n)}(t\wedge\tau_n)\tod{Y}_{x,s}(t\wedge t_0)\quad \quad \text{in} \quad\cD\left[0,\infty\right),$$
where ${Y}_{x,s}$ is a continuous Gaussian process with $\mathbb{E}{Y}_{x,s}(t)=0$ and covariance
$$\mathbb{E}\bigl[{Y}_{x,s}(t){Y}_{x,s}(u)\bigr]=\sum_{\theta=1}^{d^+_x}\int_{0}^{t}e^{2sr}d(-\varphi_{x,\theta,s}(e^{-r}))
=\int_{0}^{t}e^{2sr}d(-\varphi_{x,s}(e^{-r})), \quad\quad 0\leq t\leq u<\infty,$$
where
$$\varphi_{x,s}(y)=\sum_{\theta=1}^{d^+_x}\varphi_{x,\theta,s}(y)=\mu_x\beta(d^+_x,y,s).$$

We now prove that the convergence also hold for an infinite sum which is used to prove our final result. Denote $\cX^+_s$ and $\cX^-_s$ the set of characteristics which have in-degree $d^+_x\geq s$ and out-degree $d^-_x\geq s$ respectively. Let
$$Q_{x,\theta,s}^{(n)}(t):=e^{-st}n^{-1/2}\bigl(\widehat{v}_{x,\theta,s}^{(n)}-ne^{st}\mu^{(n)}_x
q_x^{(n)}(\theta)\beta(d^+_x,e^{-t},s)\bigr).$$
  Then we have for all $(x,\theta,s)$ and all $t>0$,
$$\Var(Q_{x,\theta,s}^{(n)}(t))\leq \mu^{(n)}_x
q_x^{(n)}(\theta)(1-q_x^{(n)}(\theta)).$$

Let $\Theta_x$ be an arbitrary subset of $[1,\ldots,d^+_x]$. By \eqref{clt_initial}, the convergence holds for the finite sum $\sum_{x\in\cX\setminus\cX^+_s}\sum_{\theta\in\Theta_x}Q_{x,\theta,s}^{(n)}(t)$. We now consider the following infinite sum
$$\sum_{x\in\cX^+_s}(d^+_x+d^-_x)\sum_{\theta\in\Theta_x}Q_{x,\theta,s}^{(n)}(t).$$
Since power function can be controlled by exponential function, there exists a constant $C>1$ such that for all $t>0$,
\begin{equation*}
\begin{split}
\sum_{x\in\cX^+_s}((d^+_x)^2+(d^-_x)^2)\Var(\sum_{\theta\in\Theta_x}Q_{x,\theta,s}^{(n)}(t))
\leq & \sum_{x\in\cX^+_s}((d^+_x)^2+(d^-_x)^2)\sum_{\theta\in\Theta_x}\Var(Q_{x,\theta,s}^{(n)}(t)) \\
\leq & \sum_{x\in\cX^-_s}(d^-_x)^2\mu^{(n)}_x+2\sum_{x\in\cX^+_s}(d^+_x)^2\mu^{(n)}_x \\
\leq & \sum_{x\in\cX^-_s}C^{d^-_x}\mu^{(n)}_x+2\sum_{x\in\cX^+_s}C^{d^+_x}\mu^{(n)}_x .
\end{split}
\end{equation*}

By Assumption~\ref{cond:CLT}, as $s\rightarrow\infty$,
\begin{equation}\label{Q_conver}
\sup_{t>0}|\sum_{x\in\cX^+_s}((d^+_x)^2+(d^-_x)^2)\sum_{\theta=1}^{d^+_x}Q_{x,\theta,s}^{(n)}(t)|\top 0.
\end{equation}
Then using the convergence of the partial sums of $Q_{x,\theta,s}^{(n)}$, we can extend the convergence to an infinite sum of $Q_{x,\theta,s}^{(n)}$, by using e.g.~\cite[Theorem 4.2]{billingsley1968convergence}. Further, the limit is also
continuous. We therefore have in $\mathcal{C}\interval$,
\begin{equation}\label{Qsum_conver}
\sum_{x\in\cX}(d^+_x+d^-_x)\sum_{\theta\in\Theta_x}Q_{x,\theta,s}^{(n)}(t)\tod \sum_{x\in\cX}(d^+_x+d^-_x)\sum_{\theta\in\Theta_x}\cZ^{\ast}_{x,\theta,s}(t).
\end{equation}
On the other hand, for the following infinite sum, we have
\begin{align*}
\bar{V}_{s}^{(n)}(t): &= \sum_{x\in\cX}(d^+_x+d^-_x)\sum_{\theta\in\Theta_x}\widetilde{V}_{x,\theta,s}^{(n)}(t)\\ &=\sum_{x\in\cX}(d^+_x+d^-_x)\sum_{\theta\in\Theta_x}\widetilde{M}_{x,\theta,s}^{(n)}(t)\\
& \quad +\sum_{x\in\cX}(d^+_x+d^-_x)\sum_{\theta\in\Theta_x}\sum_{j=s+1}^{d^+_x}s\binom{j-1}{s}\int_{0}^{t}(e^{-r}-e^{-t})^{j-s-1}
  e^{-r}\widetilde{M}_{x,\theta,j}^{(n)}(r)dr\\
  &=\bar{M}_{s}^{(n)}(t)+\sum_{j=s+1}^{\infty}s\binom{j-1}{s}\int_{0}^{t}(e^{-r}-e^{-t})^{j-s-1}
  e^{-r}\bar{M}_{j}^{(n)}(r)dr,
\end{align*}
where $\bar{M}_{s}^{(n)}(t):=\sum_{x\in\cX}(d^+_x+d^-_x)\sum_{\theta\in\Theta_x}\widetilde{M}_{x,\theta,s}^{(n)}(t)$ is a martingale with initial value 0.

We set the convention that $\widetilde{M}_{x,s}^{(n)}(t)=0$ for $d^+_x<s$. By a similar way as done for \eqref{crochet_Mtilde}, we obtain that the quadratic variation of $\bar{M}_{s}^{(n)}$, which is
$$[\bar{M}_{s}^{(n)},\bar{M}_{s}^{(n)}]_t\leq 2e^{2st}\sum_{x\in\cX^+_s}((d^+_x)^2+(d^-_x)^2)V_{x,s}^{(n)}(0)/n.$$
Using Assumption~\ref{cond:CLT} we have for all $A>1$, there exists constants $C_s$ and $C_{s,A}$ such that
\begin{equation*}
\begin{split}
\sum_{x\in\cX^+_s}((d^+_x)^2+(d^-_x)^2)V_{x,s}^{(n)}(0)\leq & \sum_{x\in\cX^-_s}(d^-_x)^2V_{x,s}^{(n)}(0)+2\sum_{x\in\cX^+_s}(d^+_x)^2V_{x,s}^{(n)}(0) \\
\leq & A^{-s}\sum_{x\in\cX}\mu^{(n)}_x(2(C_sA)^{d^+_x}+(C_sA)^{d^-_x})\leq C_{s,A} A^{-s}n.
\end{split}
\end{equation*}
 Thus for any $t>0$ and a fixed $T$, by choosing $A=e^{2t+4T}$ we get
$$[\bar{M}_{s}^{(n)},\bar{M}_{s}^{(n)}]_t\leq 2\sum_{x\in\cX^+_s}((d^+_x)^2+(d^-_x)^2)e^{2st}V_s^{(n)}(0)/n\leq C_{s,A} e^{-4Ts}.$$
By Doob's $L^2$ inequality, we have (for some constant $C'_{s,T}$)
$$\mathbb{E}[\sup_{t\leq T}(\bar{M}_{s}^{(n)}(t))^2]\leq4\mathbb{E}[\bar{M}_{s}^{(n)},\bar{M}_{s}^{(n)}]_T\leq C'_{s,T}e^{-4Ts}.$$
Then combining the Cauchy-Schwarz inequality, we obtain (for some constant $C''_{s,T}$)
\begin{equation}\label{supcontrol}
  \mathbb{E}[\sup_{t\leq T}|\bar{M}_{s}^{(n)}(t)|]\leq C''_{s,T}e^{-2Ts}.
\end{equation}
 Let us define
$$\xi_{N,n}(t):=\sum_{j=N}^{\infty}s\binom{j-1}{s}\int_{0}^{t}(e^{-r}-e^{-t})^{j-s-1}
  e^{-r}\bar{M}_{j}^{(n)}(r)dr.$$
 Then by \eqref{supcontrol} and some (simple) calculations we find that
\begin{align*}
\mathbb{E}\bigl(\sup_{t\leq T}|\xi_{N,n}(t)|\bigr) &\leq\sum_{j=N}^{\infty}s\binom{j-1}{s}\int_{0}^{t}(e^{-r}-e^{-t})^{j-s-1}
  e^{-r}\mathbb{E}\bigl[\sup_{t\leq T}|\bar{M}_{j}^{(n)}(t)|\bigr]dr\\
  & \leq C''_{s,T}T\sum_{j=N}^{\infty}s\binom{j-1}{s}(1-e^{-T})^{j-s-1}e^{-2Tj}\\
  & \leq C''_{s,T}Tse^{sT}\sum_{j=N}^{\infty}e^{-2Tj},
\end{align*}
which implies for any fixed $s$ and $T$, $\mathbb{E}(\sup_{t\leq T}|\xi_{N,n}(t)|)\rightarrow 0$ as $N\rightarrow\infty$, uniformly in $n$. 

Using again the convergence of the partial sums of $\widetilde{V}_{x,\theta,s}^{*(n)}$, we can extend the convergence to some infinite sums of $\widetilde{V}_{x,\theta,s}^{(n)}$. It follows that in $\cD\interval$,
\begin{equation}\label{Vsum_conver}
\sum_{x\in\cX}(d^+_x+d^-_x)\sum_{\theta\in\Theta_x}e^{-st}\widetilde{V}_{x,\theta,s}^{(n)}(t\wedge\tau_n)\tod
\sum_{x\in\cX}(d^+_x+d^-_x)\sum_{\theta\in\Theta_x}\cZ'_{x,\theta,s}(t\wedge t_0).
\end{equation}
Combining now \eqref{Qsum_conver} and \eqref{Vsum_conver}, it then follows that jointly
\begin{equation*}
 \sum_{x\in\cX}(d^+_x+d^-_x)\sum_{\theta\in\Theta_x}V^{*(n)}_{x,\theta,s}(t\wedge \tau_n)\tod
\sum_{x\in\cX}(d^+_x+d^-_x)\sum_{\theta\in\Theta_x}\cZ_{x,\theta,s}(t\wedge t_0),
\end{equation*}
and also the partial sum for any fixed $r$,
\begin{equation*}
 \sum_{s=1}^{r}\sum_{x\in\cX}(d^+_x+d^-_x)\sum_{\theta\in\Theta_x}V^{*(n)}_{x,\theta,s}(t\wedge \tau_n)\tod
\sum_{s=1}^{r}\sum_{x\in\cX}(d^+_x+d^-_x)\sum_{\theta\in\Theta_x}\cZ_{x,\theta,s}(t\wedge t_0).
\end{equation*}
Then notice that
\begin{align*}
   \sum_{x\in\cX}\sum_{\theta=1}^{d^+_x}\sum_{s=d^+_x-\theta+2}^{d^+_x}V_{x,\theta,s}^{*(n)}(t)
  =  \sum_{s=2}^{\infty}\sum_{x\in\cX}\sum_{\theta\in\Theta_x}V_{x,\theta,s}^{*(n)}(t).
\end{align*}
Similarly, we define the following infinite tail sum
\begin{equation*}
  \bar{\xi}_{N,n}(t):=\sum_{s=N}^{\infty}\sum_{x\in\cX}\sum_{\theta\in\Theta_x}\widetilde{V}_{x,\theta,s}^{(n)}.
\end{equation*}
Note that when $s$ is large, $C''_{s,T}$ can be bounded by another constant $C_T$ only depending on $T$. Then by the same way as above and \eqref{supcontrol}, we obtain
\begin{align*}
  \mathbb{E}\bigl(\sup_{t\leq T}|\bar{\xi}_{N,n}(t)|\bigr) &\leq \sum_{s=N}^{\infty}\sum_{j=s+1}^{\infty}s\binom{j-1}{s}\int_{0}^{t}(e^{-r}-e^{-t})^{j-s-1}
  e^{-r}\mathbb{E}\bigl[\sup_{t\leq T}|\bar{M}_{j}^{(n)}(t)|\bigr]dr
  +\sum_{s=N}^{\infty}\bar{M}_{s}^{(n)}(t) \\
  & \leq C_{T}\sum_{s=N}^{\infty}e^{-2Ts}+ \leq C_{T}T\sum_{j=N+1}^{\infty}\sum_{s=N}^{j-1}s\binom{j-1}{s}(1-e^{-T})^{j-s-1}e^{-2Tj} \\
   & \leq C_{T}\sum_{s=N}^{\infty}e^{-2Ts}+ C_T\sum_{j=N+1}^{\infty}je^{jT}e^{-2Tj} \leq 2C_{T}\sum_{s=N}^{\infty}se^{-sT},
\end{align*}
which implies that for any fixed $T>0$, $\mathbb{E}(\sup_{t\leq T}|\xi_{N,n}(t)|)\rightarrow 0$ as $N\rightarrow\infty$, uniformly in $n$. Hence the same argument allow us to pass the limit under the infinite sum and with the limit being continuous. Using \eqref{Q_conver}, the other part for the sum of $Q^{(n)}_{x,\theta,s}$ also converges. It then follows that, by using \eqref{Vsum_conver}, jointly in $\cD\interval$,
$$\sum_{x\in\cX}\sum_{\theta=1}^{d^+_x}\sum_{s=d^+_x-\theta+2}^{d^+_x}V_{x,\theta,s}^{*(n)}(t\wedge\tau_n)
\tod
\sum_{x\in\cX}\sum_{\theta=1}^{d^+_x}\sum_{s=d^+_x-\theta+2}^{d^+_x}\cZ_{x,\theta,s}(t\wedge t_0 ).$$

Using the same argument but easier than the one above, we can obtain the other three convergence relations
$$\sum_{x\in\cX}\sum_{\theta=1}^{d^+_x}V^{*(n)}_{x,\theta,d^+_x-\theta+1}(t\wedge\tau_n)
\tod \sum_{x\in\cX}\sum_{\theta=1}^{d^+_x}\cZ_{x,\theta,d^+_x-\theta+1}(t\wedge t_0) ,$$
$$\sum_{x\in\cX}\sum_{\theta=1}^{d^+_x}(d^+_x-\theta+1)V^{*(n)}_{x,\theta,d^+_x-\theta+1}(t\wedge\tau_n)
\tod
\sum_{x\in\cX}\sum_{\theta=1}^{d^+_x}(d^+_x-\theta+1)\cZ_{x,\theta,d^+_x-\theta+1}(t\wedge t_0),$$
and,
$$\sum_{x\in\cX}d^-_x\sum_{\theta=1}^{d^+_x}V^{*(n)}_{x,\theta,d^+_x-\theta+1}(t\wedge\tau_n)
\tod
\sum_{x\in\cX}d^-_x\sum_{\theta=1}^{d^+_x}\cZ_{x,\theta,d^+_x-\theta+1}(t\wedge t_0).$$

Then combining together \eqref{conver_X}, \eqref{Sn}, \eqref{HDn} and \eqref{ICn}  yields
\begin{align*}
n^{-1/2}\left(\clubsuit_n(t\wedge\tau_n)-n\hf^{(n)}_{\clubsuit}(t\wedge\tau_n)\right)\tod Z_{\clubsuit}(t\wedge t_0),
\end{align*}
with
\begin{equation}\label{Zs}
  Z_{S}:=\sum_{x\in\cX}\sum_{\theta=1}^{d^+_x}\cZ_{x,\theta,d^+_x-\theta+1},
\end{equation}
\begin{equation}\label{ZH+}
  Z_{H^+}:=\sum_{x\in\cX}\sum_{\theta=1}^{d^+_x}(d^+_x-\theta+1)\cZ_{x,\theta,d^+_x-\theta+1}
  +\sum_{x\in\cX}\sum_{\theta=1}^{d^+_x}\sum_{s=d^+_x-\theta+2}^{d^+_x}\cZ_{x,\theta,s},
\end{equation}
\begin{equation}\label{ZI+}
  Z_{I^+}:=\widehat{Z}-Z_{H^+},
\end{equation}
\begin{equation}\label{ZW}
  Z_{W}:=\widehat{Z}-\sum_{x\in\cX}d^-_x\sum_{\theta=1}^{d^+_x}\cZ_{x,\theta,d^+_x-\theta+1}.
\end{equation}

 In order to complete the proof, we also need to compute the covariances for the continuous Gaussian processes $Z_{S}, Z_{H^+}, Z_{I^+}$ and $Z_{W}$.

For convenience, we make a change of variable $y=e^{-t}$, which decreases from 1 to 0 as $t$ varies from 0 to $\infty$. We use the notations
\begin{align*}
\widehat{\bigsigma}^2(y):= \Var(\widehat{Z}(-\ln y)), \quad \widehat{\bigsigma}_{x,\theta,s}(y):=\Cov\bigl(\widehat{Z}(-\ln y),\cZ_{x,\theta,s}(-\ln y)\bigr),
\end{align*}
and,
$$\bigsigma_{x,\theta,r,s}(y):=\Cov\bigl(\cZ_{x,\theta,r}(-\ln y),\cZ_{x,\theta,s}(-\ln y)\bigr)$$
which has been already computed in Section~\ref{sec-CLT-S}. We now compute $\widehat{\bigsigma}^2(y)$ and $\widehat{\bigsigma}_{x,\theta,s}(y)$.

In order to compute $\widehat{\bigsigma}_{x,\theta,s}$, we apply Theorem~\ref{martingale_clt} to $\widetilde{L}_n$ and $\widetilde{M}_{x,\theta,s}^{(n)}$ for all $s=d^+_x-\theta+1, \dots, d^+_x$. Observe that each time $V_{x,\theta,s}^{(n)}$ decreases by 1, also an in ball dies and thus $L_n$ decreases by 1. Hence, the quadratic covariation is
\begin{align*}
  [\widetilde{M}_{x,\theta,s}^{(n)},\widetilde{L}_n]_t & =n^{-1}\sum_{0<r\leq t}\Delta M_{x,\theta,s}^{(n)}(r)\Delta \widehat{L}_n(r)=n^{-1}\sum_{0<r\leq t}\Delta \widehat{V}_{x,\theta,s}^{(n)}(r)\Delta \widehat{L}_n(r) \\
  & = n^{-1}\sum_{0<r\leq t}e^{(s+1)r}\Delta V_{x,\theta,s}^{(n)}(r)\Delta L_n(r)=n^{-1}\int_{0}^{t}e^{(s+1)r}d(-V_{x,\theta,s}^{(n)}(r)).
\end{align*}
Using integration by parts as before we obtain that
\begin{align*}
  [\widetilde{M}_{x,\theta,s}^{(n)},\widetilde{L}_n]_t & = \int_{0}^{t}e^{(s+1)r}d(-\varphi_{x,\theta,s}(e^{-r}))+o_p(1)\\
   & =\int_{e^{-t}}^{1}u^{-(s+1)}d\varphi_{x,\theta,s}(u)+o_p(1).
\end{align*}

We can now compute all needed covariances. First, for $\widehat{\bigsigma}^2(y)$ we have
\begin{equation}\label{sigmaww}
  \widehat{\bigsigma}^2(y):=\mbox{Var}(\widehat{Z}(-\ln y))=\mbox{Var}(y\widetilde{Z}(-\ln y))=\lambda(y-y^2)/2.
\end{equation}

Similarly, the above analysis  together with Theorem~\ref{martingale_clt}, gives that for all $x,\theta$,
$$\mbox{Cov}\bigl(Y_{x,\theta,s}(-\ln y),\widetilde{Z}(-\ln y)\bigr)=\int_{y}^{1}u^{-(s+1)}d\varphi_{x,\theta,s}(u).$$

On the other hand, for $v\leq t$, $\mbox{Cov}(Y_{x,\theta,s}(v),\widetilde{Z}(t))=\mbox{Cov}(Y_{x,\theta,s}(v),\widetilde{Z}(v))$. Thus we have
\begin{align*}\begin{aligned}
 & \mbox{Cov}(\widetilde{\cZ}_{x,\theta,s}(-\ln y),\widetilde{Z}(-\ln y))  \\
  = & \mbox{Cov}(Y_{x,\theta,s}(t),\widetilde{Z}(t))+\int_{0}^{t}(e^{-r}-e^{-t})^{j-s-1}
  e^{-r} \mbox{Cov}(Y_{x,\theta,s}(r),\widetilde{Z}(r))dr  \\
  = & \int_{y}^{1}u^{-(s+1)}d\varphi_{x,\theta,s}(u)+\sum_{j=s+1}^{d^+_x}s\binom{j-1}{s}f_{sj}(y),
\end{aligned}\end{align*}
where, with a change of variable $u=e^{-r}$, the function $f_{sj}(y)$ is
\begin{align*}\begin{aligned}
f_{sj}(y)& =\int_{y}^{1}(u-y)^{j-s-1} \int_{u}^{1}v^{-(j+1)}d\varphi_{x,\theta,j}(v)du\\
& = \int_{y}^{1}\int_{y}^{v}(u-y)^{j-s-1}v^{-(j+1)}dud\varphi_{x,\theta,j}(v)\\
 &= \frac{1}{j-s}\int_{y}^{1}(v-y)^{j-s}v^{-(j+1)}d\varphi_{x,\theta,j}(v).
\end{aligned}\end{align*}
We thus obtain
$$\mbox{Cov}(\widetilde{\cZ}_{x,\theta,s}(-\ln y),\widetilde{Z}(-\ln y))=\sum_{j=s}^{d_x^+}
\binom{j-1}{s-1}\int_{y}^{1}(v-y)^{j-s}v^{-(j+1)}d\varphi_{x,\theta,j}(v).$$
Also, notice that $\mbox{Cov}(\widehat{\cZ}_{x,\theta,s}(-\ln y),\widetilde{Z}(-\ln y))=0$ since they are independent. Then we conclude that
\begin{equation}\label{sigmaWs}
\begin{split}
 \widehat{\bigsigma}_{x,\theta,s}(y) &: =y^{s+1}\mbox{Cov}(\widetilde{\cZ}_{x,\theta,s}(-\ln y),\widetilde{Z}(-\ln y))\\
  & =y^{s+1}\sum_{j=s}^{d^+_x}
\binom{j-1}{s-1}\int_{y}^{1}(v-y)^{j-s}v^{-(j+1)}d\varphi_{x,\theta,j}(v).
  \end{split}
\end{equation}

We can write now the covariances for the processes $Z_{S}$, $Z_{H^+}$, $Z_{I^+}$ and $Z_{W}$ from $\widehat{\bigsigma}^2(y)$, $\widehat{\bigsigma}_{x,\theta,s}(y)$, $\bigsigma_{x,\theta,r,s}(y)$ given by \eqref{sigma sr} and $\bigsigma^{\ast}_{x,\theta_1,\theta_2,s_1,s_2}$ given by \eqref{sigma_ast}.  We only write the covariances between $Z_{S}$, $Z_{H^+}$ and $Z_{W}$; the covariances of $Z_{I^+}$ could be easily deduced from those of $Z_{H^+}$. For convenience, we set $\pi_x(\theta):=d^+_x-\theta+1$ in the following.

\medskip

For the variances,  by using \eqref{Zs}-\eqref{ZW}, we have:
\begin{align}\begin{aligned}\label{sigma S}
  \bigsigma_{S,S}(y)=&\sum_{x\in\cX}\sum_{\theta=1}^{d^+_x}\bigsigma_{x,\theta,\pi_x(\theta),\pi_x(\theta)}(y)
  +\sum_{x\in\cX}\sum_{\theta_1=1}^{d^+_x} \sum_{\theta_2=1}^{d^+_x} \bigsigma^{\ast}_{x,\theta_1,\theta_2, \pi_x(\theta_1),\pi_x(\theta_2)}
  (y),
\end{aligned}\end{align}
\begin{align}\begin{aligned}\label{sigma HD}
  \bigsigma_{H^+,H^+}(y)=\sum_{x\in\cX}\sum_{\theta=1}^{d^+_x}&\Bigl[\pi^2_x(\theta)\bigsigma_{x,\theta,\pi_x(\theta),\pi_x(\theta)}(y)
  + 2\sum_{s=\pi_x(\theta)+1}^{d^+_x}\pi_x(\theta)\bigsigma_{x,\theta,\pi_x(\theta),s}(y)\\
 &\quad +\sum_{s=\pi_x(\theta)+1}^{d^+_x} \sum_{r=\pi_x(\theta)+1}^{d^+_x}\bigsigma_{x,\theta,r,s}(y)\Bigr]\\
 +\sum_{x\in\cX}\sum_{\theta_1=1}^{d^+_x}\sum_{\theta_2=1}^{d^+_x}
 &\Bigl[\sum_{s_1=\pi_x(\theta_1)+1}^{d_x^+}\sum_{s_2=\pi_x(\theta_2)+1}^{d_x^+}
 \pi_x(\theta_1)\pi_x(\theta_2)\bigsigma^{\ast}_{x,\theta_1,\theta_2,\pi_x(\theta_1),\pi_x(\theta_2)}(y)\\
 &\quad +\bigsigma^{\ast}_{x,\theta_1,\theta_2,s_1,s_2}(y) +2\pi_x(\theta_2)\sum_{s_2=\pi_x(\theta_2)+1}^{d_x^+}
 \bigsigma^{\ast}_{x,\theta_1,\theta_2, s,\pi_x(\theta_2)}(y)\Bigr],
\end{aligned}\end{align}
and,
\begin{equation}\label{sigma_IC}
\begin{split}
  \bigsigma_{W,W}(y)= &\sum_{x\in\cX}\sum_{\theta=1}^{d^+_x}\Bigl[(d^-_x)^2 \bigsigma_{x,\theta,\pi_x(\theta),\pi_x(\theta)}(y)
  -2d^-_x \widehat{\bigsigma}_{x,\theta,\pi_x(\theta)}(y)\Bigr]+\widehat{\bigsigma}^2(y)\\
  &\quad +\sum_{x\in\cX}(d^-_x)^2\sum_{\theta_1=1}^{d^+_x}\sum_{\theta_2=1}^{d^+_x}
 \bigsigma^{\ast}_{x,\theta_1,\theta_2,\pi_x(\theta_1),\pi_x(\theta_2)}(y).
\end{split}
\end{equation}

For the covariances, by using again \eqref{Zs}-\eqref{ZW}, we have:
\begin{equation}\label{sigma_SHD}
\begin{split}
&  \bigsigma_{S,H^+}(y)= \sum_{x\in\cX}\sum_{\theta=1}^{d^+_x}\Bigl[\pi_x(\theta)\bigsigma_{x,\theta,\pi_x(\theta),\pi_x(\theta)}(y)
  +\sum_{s=\pi_x(\theta)+1}^{d^+_x}\bigsigma_{x,\theta,\pi_x(\theta),s}(y)\Bigr]\\
  &+\sum_{x\in\cX}\sum_{\theta_1,\theta_2=1}^{d^+_x}\Bigl[
 \pi_x(\theta_1)\bigsigma^{\ast}_{x,\theta_1,\theta_2,\pi_x(\theta_1),\pi_x(\theta_2)}(y)
 +
 \sum_{s=\pi_x(\theta_1)+1}^{d_x^+}
 \bigsigma^{\ast}_{x,\theta_1,\theta_2,s,\pi_x(\theta_2)}(y)\Bigr],
\end{split}
\end{equation}
\begin{equation}\label{sigma_SIC}
\begin{split}
  \bigsigma_{S,W}(y)= &\sum_{x\in\cX}\sum_{\theta=1}^{d^+_x}\Bigl[\widehat{\bigsigma}_{x,\theta,\pi_x(\theta)}(y)
  -d^-_x\bigsigma_{x,\theta,\pi_x(\theta),\pi_x(\theta)}(y)\Bigr]\\
  &-\sum_{x\in\cX}d^-_x\sum_{\theta_1=1}^{d^+_x}\sum_{\theta_2=1}^{d^+_x}
 \bigsigma^{\ast}_{x,\theta_1,\theta_2,\pi_x(\theta_1),\pi_x(\theta_2)}(y),
\end{split}
\end{equation}
and,
\begin{equation}\label{sigma_HDIC}
\begin{split}
 & \bigsigma_{H^+,W}(y)=\sum_{x\in\cX}\sum_{\theta=1}^{d^+_x}
  \Bigl[-d^-_x\pi_x(\theta)\bigsigma_{x,\theta,\pi_x(\theta),\pi_x(\theta)}(y)
  +\pi_x(\theta)\widehat{\bigsigma}_{x,\theta,\pi_x(\theta)}(y)\\
  & \quad \quad \quad \quad \quad \quad \quad \quad \quad
  +\sum_{s=\pi_x(\theta)+1}^{d^+_x}\bigl(\widehat{\bigsigma}_{x,\theta,s}(y)-d^-_x\bigsigma_{x,\theta,\pi_x(\theta),s}(y)\bigr)\Bigr]\\
& -\sum_{x\in\cX}d^-_x \sum_{\theta_1,\theta_2=1}^{d^+_x} \Bigl[
 \sum_{s=\pi_x(\theta_1)+1}^{d^+_x}
 \bigsigma^{\ast}_{x,\theta_1,\theta_2,s,\pi_x(\theta_2)}(y)
 +
 \pi_x(\theta_1)\bigsigma^{\ast}_{x,\theta_1,\theta_2,\pi_x(\theta_1),\pi_x(\theta_2)}(y)\Bigr].
\end{split}
\end{equation}

As for the covariances of $Z_{I^+}$, we can deduce from the above six formulas by using the relation \eqref{ZI+}. This completes the proof of Theorem~\ref{normality}.

\subsection{Proof of Theorem~\ref{normalityFinal}}\label{sec:normalityFinal}

As discussed in Remark~\ref{rem-finmo}, the Assumption~\ref{cond:CLT} implies Assumption~\ref{cond-average}. Hence, the case $\hz=0$ follows from Theorem~\ref{thm:LLN}. We now consider the case $\hz\in\left(0,1\right]$, where  $\hz$ is a stable solution and $\alpha:=f_W'(\hz)>0$.

Since $\alpha>0$, for a positive constant $\varepsilon$ small enough, we have $f_W(\hz-\varepsilon)<0$ and $f_W(\hz+\varepsilon)>0$.
By Lemma~\ref{lem-remak}, we have  $f^{(n)}_W\rightarrow f_W$ uniformly on $[0,1]$. For $n$ large enough, we also have $f^{(n)}_W(\hz-\varepsilon)<0$ and $f^{(n)}_W(\hz+\varepsilon)>0$. Hence for $n$ large enough, there exists a sequence $\hp_n$ in $(\hz-\varepsilon,\hz+\varepsilon)$ such that $f^{(n)}_W(\hp_n)=0$ and $f^{(n)}_W>0$ on $[\hz+\varepsilon,1]$. Since $\varepsilon$ can be arbitrarily small, we obtain $\hp_n\rightarrow\hz$.

Define $\widehat{t}_n:=-\ln\hp_n$. Consequently we also have $\widehat{t}_n\rightarrow t^{\star}$.

Next, we use the Skorokhod representation theorem which shows that one can change the probability space where all the random variables are well defined and all the convergence results of Theorem~\ref{thm-CLT-S} and $\tst\rightarrow t^{\star}$ (from Lemma~\ref{tau}) hold a.s..

Taking $t=\tst$ and $t_0=t^{\star}$, we get
\begin{align*}
  W_n(\tst) & =n\hf^{(n)}_W(\tst)+n^{1/2}Z_{W}(\tst\wedge t^{\star})+o(n^{1/2}) \\
   & = n\hf^{(n)}_W(\tst)+n^{1/2}Z_{W}(t^{\star})+o(n^{1/2}),
\end{align*}
by the continuity of $Z_{W}$. Since $W_n(\tst)=-1$, then
$$\hf^{(n)}_W(\tst)=-n^{-1/2}Z_{W}\bigl(t^{\star})+o(n^{-1/2}\bigr).$$

Since, as $n\rightarrow\infty$, $\tst\rightarrow t^{\star}$ and $\widehat{t}_n\rightarrow t^{\star}$ hold a.s., there exists some $\xi_n$ in the interval between $\widehat{t}_n$ and $\tst$ such that $\xi_n\rightarrow t^{\star}$, and further, $$(\hf^{(n)}_W)'(\xi_n)\rightarrow\hf'_W(t^{\star})=-\hz\alpha.$$ It follows then by Mean-Value theorem that
$$\hf^{(n)}_W(\tst)=\hf^{(n)}_W(\tst)-\hf^{(n)}_W(\widehat{t}_n)=(\hf^{(n)}_W)'(\xi_n)
(\tst-\widehat{t}_n)=(-\hz\alpha+o(1))(\tst-\widehat{t}_n).$$

Hence we have
$$\tst-\widehat{t}_n=\Bigl(-\frac{1}{\hz\alpha}+o(1)\Bigr)\hf^{(n)}_W(\tst)=
n^{-1/2}\frac{1}{\hz\alpha}(Z_{W}(t^{\star})+o(1)).$$

Then, by a similar argument for $S_n(\tst)$, we have a.s. for some $\xi'_n\rightarrow t^{\star}$,
\begin{align}\begin{aligned}
n^{-1/2}S_n(\tst)& = n^{1/2}\hf^{(n)}_S(\tst)+Z_{S}(t^{\star})+o(1)\\
&= n^{1/2}\hf^{(n)}_S(\tst)+n^{1/2}(\hf^{(n)}_S)'(\xi'_n)(\tst-\widehat{t}_n)+Z_{S}(t^{\star})+o(1)\\
&= n^{1/2}f^{(n)}_S(\hp_n)+\frac{\hf'_S(t^{\star})}{\alpha\hz}Z_{W}(t^{\star})+Z_{S}(t^{\star})+o(1)  \\
& =n^{1/2}f^{(n)}_S(\hp_n)-\frac{f'_S(\hz)}{\alpha}Z_{W}(t^{\star})+Z_{S}(t^{\star})+o(1),
\end{aligned}\end{align}
where the last equality follows from the fact that $(\hf_S)'(t)=-(f_S)(e^{-t})e^{-t}$ and $e^{-t^{\star}}=\hz$.

Using a similar argument, we have the following analogues:
\begin{equation*}
  n^{-1/2}H^+_n(\tst)=n^{1/2}f^{(n)}_{H^+}(\hp_n)-\frac{f'_{H^+}(\hz)}
  {\alpha}Z_{W}(t^{\star})+Z_{H^+}(t^{\star})+o(1),
\end{equation*}
\begin{equation*}
  n^{-1/2}I^+_n(\tst)=n^{1/2}f^{(n)}_{I^+}(\hp_n)-\frac{f'_{I^+}(\hz)}
  {\alpha}Z_{W}(t^{\star})+Z_{I^+}(t^{\star})+o(1),
\end{equation*}
and,
\begin{equation*}
  n^{-1/2}S^{(n)}_{x,\theta,\ell}(\tst)=n^{1/2}{s}^{(n)}_{x,\theta,\ell}(\hp_n)-\frac{{s}'_{x,\theta,\ell}(\hz)}
  {\alpha}Z_{W}(t^{\star})+Z_{x,\theta,l}(t^{\star})+o(1),
\end{equation*}
for all $x\in\cX$ and $0\leq \ell < \theta\leq d^+_x$. This completes the proof of Theorem~\ref{normalityFinal}.

\subsection{Proof of Theorem~\ref{thm-agg-LLN}}\label{sec:thm-agg-LLN}
Using the notations as in the proof of Theorem~\ref{centrality} (see Section~\ref{sec:thmCentrality}), we have
$$S^{(n)}_{x,\theta,\ell}(t)=U^{(n)}_{x,\theta,d^+_x-\ell}(t),$$
and then
$$\sum_{x\in \cX} \bar{L}^{\Diamond}_x \sum_{\theta=1}^{d_x^+}\sum_{\ell=1}^{\theta-1} \ell S^{(n)}_{x, \theta, \ell} (t) = \sum_{x\in\cX}\bar{L}^{\Diamond}_x\sum_{\theta=1}^{d^+_x}\sum_{\ell=1}^{\theta-1}\ell U^{(n)}_{x,\theta,d^+_x-\ell}(t),$$
and,
$$\sum_{x\in\cX}\bar{L}^{\odot}_xD^{(n)}_x(t)=\sum_{x\in\cX}\bar{L}^{\odot}_x(n\mu^{(n)}_x-
\sum_{\theta=1}^{d^+_x}\sum_{s=d^+_x-\theta+1}^{d^+_x}U^{(n)}_{x,\theta,s}(t)).$$

Since $\bar{L}^{\odot}_x$ and $\bar{L}^{\Diamond}_x$ are bounded by a constant $C$, the two above expressions are both dominated by $$nC\sum_{x\in\cX}(d^+_x+d^-_x)\sum_{\theta=1}^{d^+_x}\sum_{s=d^+_x-\theta+1}^{d^+_x}U_{d^+_x,s}(t).$$
Thus by \eqref{U_limit} and \eqref{result}, it follows that
\begin{align*}
\sup\limits_{t\leq \tau_n}\bigl|\frac{\Gamma_n^{\Diamond}(t)}{n}-f_{\Diamond}(e^{-t})\bigr|\top 0.
\end{align*}

For $\hz=0$, by Lemma~\ref{tau}, $\tst\top\infty$. Then $e^{-\tst}\top 0$, and $f_{\Diamond}(0)=\bar{\Gamma}^{\Diamond}-\sum_{x\in\cX}\bar{L}^{\odot}_x\mu_x$. thus it follows by the continuity of $f_{\Diamond}$ that
$$f_{\Diamond}(e^{-\tst})=\bar{\Gamma}^{\Diamond}-\sum_{x\in\cX}\bar{L}^{\odot}_x\mu_x+o_p(1).$$
We therefore have
$$\frac{\Gamma^{\Diamond}_n(\tst)}{n}\top\bar{\Gamma}^{\Diamond}-\sum_{x\in\cX}\bar{L}^{\odot}_x\mu_x.$$

For $\hz\in\left(0,1\right]$ and $f_W'(\hz)>0$, by Lemma~\ref{tau},  $\tau^\star_n\top-\ln\hz$. Then a similar argument as above and the continuity of $f_{\Diamond}$, implies that
\begin{align*}
\frac{\Gamma_n^{\Diamond}(\tst)}{n}\top f_{\Diamond}(\hz).
\end{align*}
This  completes the proof of Theorem~\ref{thm-agg-LLN}.

\subsection{Proof of Theorem~\ref{thm-CLT-Sys}}\label{sec:thm-CLT-Sys}
Recall that we defined $V_{x,\theta,s}^{(n)}(t)$ as the number of bins (institutions) with type $x\in \cX$, threshold $\theta$ and at least $s$ alive balls at time $t$. We notice that
$$\sum_{x\in\cX}\bar{L}^{\odot}_xD^{(n)}_x(t)=
\sum_{x\in\cX}\bar{L}^{\odot}_x\bigl(n\mu^{(n)}_x-
\sum_{\theta=1}^{d^+_x}V_{x,\theta,d^+_x-\theta+1}^{(n)}(t)\bigr),$$
and,
$$\sum_{\ell=1}^{\theta-1}\ell S^{(n)}_{x,\theta,\ell}(t)=(\theta-1)V_{x,\theta,d^+_x-\theta+1}^{(n)}(t)
-\sum_{s=d^+_x-\theta+1+1}^{d^+_x}V_{x,\theta,s}^{(n)}(t).$$

Since $\bar{L}^{\odot}_x$ and $\bar{L}^{\Diamond}_x$ are assumed bounded, there exists a constant $C$ such that $\bar{L}^{\odot}_x+\bar{L}^{\Diamond}_x \leq C$ for all $x\in\cX$. Hence, using the same arguments as in Section~\ref{sec:normality}, used for the convergence of $H^+_n$ and $S_n$, (and the same bound for the tail sum multiplied by the constant $C$) leads to the convergence of the following (infinite) sum
$$\sum_{x\in\cX}\bar{L}^{\odot}_x
\sum_{\theta=1}^{d^+_x}V^{*(n)}_{x,\theta,d^+_x-\theta+1}(t)\tod
\sum_{x\in\cX}\bar{L}^{\odot}_x
\sum_{\theta=1}^{d^+_x}\cZ_{x,\theta,d^+_x-\theta+1}(t),$$
\begin{align*}
\sum_{x\in\cX}\bar{L}^{\Diamond}_x & \sum_{\theta=1}^{d^+_x}[(\theta-1)V^{*(n)}_{x,\theta,d^+_x-\theta+1}(t)
-\sum_{s=d^+_x-\theta+1+1}^{d^+_x}V_{x,\theta,s}^{*(n)}(t)] \\
& \tod\sum_{x\in\cX}\bar{L}^{\Diamond}_x\sum_{\theta=1}^{d^+_x}[(\theta-1)\cZ_{x,\theta,d^+_x-\theta+1}(t)
-\sum_{s=d^+_x-\theta+1+1}^{d^+_x}\cZ_{x,\theta,s}(t)].
\end{align*}

Hence we have jointly in $\cD\interval$,
\begin{align*}
n^{-1/2}\left(\Gamma_n^{\Diamond}(t\wedge\tau_n) -n\hf^{(n)}_{\Diamond}(t\wedge\tau_n)\right)\tod Z_{\Diamond}(t\wedge t_0),
\end{align*}
where
\begin{equation}\label{Z_risk}
Z_{\Diamond}:=-\sum_{x\in\cX}\bar{L}^{\odot}_x
\sum_{\theta=1}^{d^+_x}\cZ_{x,\theta,d^+_x-\theta+1}(t)-\sum_{x\in\cX}\bar{L}^{\Diamond}_x\sum_{\theta=1}^{d^+_x}\Bigl[(\theta-1)\cZ_{x,\theta,d^+_x-\theta+1}(t)
-\sum_{s=d^+_x-\theta+1+1}^{d^+_x}\cZ_{x,\theta,s}(t)\Bigr].
\end{equation}

Let further
$$Z^{(1)}_{\Diamond}:=\sum_{x\in\cX}\bar{L}^{\odot}_x
\sum_{\theta=1}^{d^+_x}\cZ_{x,\theta,d^+_x-\theta+1}(t),$$
$$Z^{(2)}_{\Diamond}:=\sum_{x\in\cX}\bar{L}^{\Diamond}_x\sum_{\theta=1}^{d^+_x}[(\theta-1)\cZ_{x,\theta,d^+_x-\theta+1}(t)
-\sum_{s=d^+_x-\theta+1+1}^{d^+_x}\cZ_{x,\theta,s}(t)].$$

Let $\bigsigma_{i,j}(e^{-t}):=\Cov(Z^{(i)}_{\Diamond}(t),Z^{(j)}_{\Diamond}(t))$, $i,j=1,2$.
The covariance $\bigsigma^2_{\Diamond}(t)$ is therefore
$$\bigsigma^2_{\Diamond}(t)=\bigsigma_{1,1}(e^{-t})+2\bigsigma_{1,2}(e^{-t})+\bigsigma_{2,2}(e^{-t}),$$
where $\bigsigma_{i,j}(y)$ (for $i,j=1,2$) are calculated in the following. By using results of Section~\ref{normality}, we have (recall that $\pi_x(\theta):=d_x^+-\theta+1$):

\begin{equation}\label{sigma_11}
\begin{split}
  \bigsigma_{1,1}(y)= \sum_{x\in\cX}(\bar{L}^{\odot}_x)^2\Bigl(\sum_{\theta=1}^{d^+_x}\bigsigma_{x,\theta,\pi_x(\theta),\pi_x(\theta)}(y)
  +\sum_{\theta_1,\theta_2=1}^{d^+_x}\bigsigma^{\ast}_{x,\theta_1,\theta_2,\pi_x(\theta_1),\pi_x(\theta_2)}(y)\Bigr),
\end{split}
\end{equation}
\begin{equation}\label{sigma_22}
\begin{split}
&  \bigsigma_{2,2}(y)=
  \sum_{x\in\cX}(\bar{L}^{\Diamond}_x)^2\sum_{\theta=1}^{d^+_x}\Bigl[(\theta-1)^2\bigsigma_{x,\theta,\pi_x(\theta),\pi_x(\theta)}(y)
  -2(\theta-1)\sum_{s=\pi_x(\theta)+1}^{d^+_x}\bigsigma_{x,\theta,\pi_x(\theta),s}(y)\\
 & +\sum_{r,s=\pi_x(\theta)+1}^{d^+_x}\bigsigma_{x,\theta,r,s}(y)\Bigr]+\sum_{x\in\cX}(\bar{L}^{\Diamond}_x)^2\sum_{\theta_1,\theta_2=1}^{d^+_x}
 \Bigl[\sum_{s_1=\pi_x(\theta_1)+1}^{d_x^+}\sum_{s_2=\pi_x(\theta_2)+1}^{d_x^+}
 \bigsigma^{\ast}_{x,\theta_1,\theta_2,s_1,s_2}(y)\\
 &+(\theta_1-1)(\theta_2-1)\bigsigma^{\ast}_{x,\theta_1,\theta_2,\pi_x(\theta_1),\pi_x(\theta_2)}(y)
 -2(\theta_2-1)\sum_{s=\pi_x(\theta_1)+1}^{d_x^+}
 \bigsigma^{\ast}_{x,\theta_1,\theta_2,s,\pi_x(\theta_2)}(y)\Bigr],
\end{split}
\end{equation}
and,
\begin{equation}\label{sigma_12}
\begin{split}
 & \bigsigma_{1,2}(y)= \sum_{x\in\cX}\sum_{\theta=1}^{d^+_x}\Bigl[(\theta-1)\bigsigma_{x,\theta,\pi_x(\theta),\pi_x(\theta)}(y)
  -\sum_{s=\pi_x(\theta)+1}^{d^+_x}\bigsigma_{x,\theta,\pi_x(\theta),s}(y)\Bigr]\\
  &+\sum_{x\in\cX}\sum_{\theta_1,\theta_2=1}^{d^+_x}\Bigl[
 (\theta_1-1)\bigsigma^{\ast}_{x,\theta_1,\theta_2,\pi_x(\theta_1),\pi_x(\theta_2)}(y)
 - \sum_{s=\pi_x(\theta_1)+1}^{d_x^+}
 \bigsigma^{\ast}_{x,\theta_1,\theta_2,s,\pi_x(\theta_2)}(y)\Bigr].
\end{split}
\end{equation}

For the final system-wide aggregation functions, if $\hz=0$ then by Theorem~\ref{normalityFinal} asymptotically almost all institutions default during the cascade and
\begin{align*}
\frac{\Gamma_n^{\Diamond}(\tst)}{n}\top \bar{\Gamma}^{\Diamond} -\sum_{x\in \cX}\mu_x \bar{L}^{\odot}_x.
\end{align*}
We consider now $\hz\in\left(0,1\right]$ and $f_W'(\hz)>0$. Note that the convergence of aggregation functions \eqref{norm_risk} hold jointly with the  convergence of other functions such as $H^+_n, S_n, D_n, W_n$, i.e. \eqref{norm_risk} and \eqref{joint_normality} hold jointly. Hence, similarly as in the proof of Theorem~\ref{normalityFinal} in Section~\ref{sec:normalityFinal},  we find that
\begin{equation*}
  n^{-1/2}\Gamma^{\Diamond}_n(\tst)=n^{1/2}f^{(n)}_{\Diamond}(\hp_n)-\frac{f'_{\Diamond}(\hz)}
  {\alpha}Z_{W}(t^{\star})+Z_{\Diamond}(t^{\star})+o(1).
\end{equation*}
This gives $$\cZ_{\Diamond}:=Z_{\Diamond}(t^{\star})-\alpha^{-1}f'_{\Diamond}(\hz)Z_{W}(t^{\star})=Z_{\Diamond}(t^{\star})-\Delta(\hz) Z_{W}(t^{\star}).$$
Hence $\cZ_{\Diamond}$ is a centered Gaussian random variable with variance
$$\Sigma_{\Diamond}:=\bigsigma_{\Diamond}^2(t^{\star})+\Delta(\hz)^2\bigsigma_{W,W}(\hz)-2\Delta(\hz)\bigsigma_{\Diamond,W}(\hz).$$

We already calculated $\bigsigma_{\Diamond}^2$ and $\bigsigma_{W,W}$ (see Theorem~\ref{normality}). The term $\bigsigma_{\Diamond,W}(e^{-t})$ can be  calculated by using
\begin{align}\begin{aligned}\label{sigma_boxW}
  \bigsigma_{\Diamond,W}(e^{-t})= &-\Cov(Z^{(1)}_{\Diamond}(t),Z_W(t))-\Cov(Z^{(2)}_{\Diamond}(t),Z_W(t)) \\
   & =-\bigsigma^{(1)}_{\Diamond,W}(e^{-t})-\bigsigma^{(2)}_{\Diamond,W}(e^{-t}),
\end{aligned}\end{align}
where
\begin{equation}\label{sigma_boxW1}
\begin{split}
   \bigsigma^{(1)}_{\Diamond,W}(y)= &\sum_{x\in\cX}\bar{L}^{\odot}_x\sum_{\theta=1}^{d^+_x}\widehat{\bigsigma}_{x,\theta,\pi_x(\theta)}(y)
  -\sum_{x\in\cX}\bar{L}^{\odot}_xd^-_x\sum_{\theta=1}^{d^+_x}\bigsigma_{x,\theta,\pi_x(\theta),\pi_x(\theta)}(y)\\
  &-\sum_{x\in\cX}\bar{L}^{\odot}_xd^-_x\sum_{\theta_1,\theta_2=1}^{d^+_x}
 \bigsigma^{\ast}_{x,\theta_1,\theta_2,\pi_x(\theta_1),\pi_x(\theta_2)}(y),
\end{split}
\end{equation}
and,
\begin{equation}\label{sigma_boxW2}
\begin{split}
 &  \bigsigma^{(2)}_{\Diamond,W}(y)= \sum_{x\in\cX}\bar{L}^{\Diamond}_x\sum_{\theta=1}^{d^+_x}
  \Bigl[-d^-_x(\theta-1)\bigsigma_{x,\theta,\pi_x(\theta),\pi_x(\theta)}(y)
  +(\theta-1)\widehat{\bigsigma}_{x,\theta,\pi_x(\theta)}(y)\\
  & \quad \quad \quad \quad \quad \quad
  -\sum_{s=\pi_x(\theta)+1}^{d^+_x}\Bigl(\widehat{\bigsigma}_{x,\theta,s}(y)-d^-_x\bigsigma_{x,\theta,\pi_x(\theta),s}(y)\Bigr)\Bigr]\\
& +\sum_{x\in\cX}\bar{L}^{\Diamond}_x d^-_x \sum_{\theta_1,\theta_2=1}^{d^+_x} \Bigl[
 \sum_{s=\pi_x(\theta_1)+1}^{d^+_x}
 \bigsigma^{\ast}_{x,\theta_1,\theta_2,s,\pi_x(\theta_2)}(y)
 - (\theta_1-1)\bigsigma^{\ast}_{x,\theta_1,\theta_2,\pi_x(\theta_1),\pi_x(\theta_2)}(y)\Bigr].
\end{split}
\end{equation}
This completes the proof of Theorem~\ref{thm-CLT-Sys}.

\subsection{Proof of Theorem~\ref{planner-LLN}}\label{sec:planner-LLN}
Our proof of Theorem~\ref{planner-LLN} is based on ideas applied in~\cite{janson2009percolation} which studies the conditions for existence of giant component in the percolated random (non directed) graph with given vertex degrees. We consider type-dependent bond percolation model where we remove each incoming link to any institution of type $x\in\cX$ with probability $\alpha_x$.

Hence, we first remove all potential saved links by the planner from the network. Note that we  also include extra removed links between solvent institutions that will not play any role in the default contagion process.  Then we run the death process as described in Section~\ref{sec:deathFin} and Appendix~\ref{sec:thmCentrality}. We denote by $\widetilde{U}^{(n)}_{x, \theta,\ell}(t)$ the number of bins (institutions) with type $x\in \cX$, threshold $\theta$ and $\ell$ alive (in-) balls at time $t$ in the percolated random graph. It is then not hard to show that \eqref{U_limit} changes to 
\begin{equation*}
  \sup\limits_{t\leq \tau_n}\bigl|\frac{1}{n}\sum_{s=r+1}^{d^+_x}\widetilde{U}^{(n)}_{x,\theta,s}(t)-\mu_xq_x(\theta)\sum_{s=r+1}^{d^+_x}b(d^+_x,\alpha_x+(1-\alpha_x)e^{-t},s)\bigr|
\top 0, \quad  \text{as} \quad n\to \infty.
\end{equation*}

Consider now $\widetilde{S}_{x,\theta,\ell}^{(n)}$, the number of solvent institutions with type $x$, threshold $\theta$ and $\ell=0, \dots, \theta-1$ defaulted neighbors at time $t$. Hence, by writing above equation for $r_1=d^+_x-\ell$ and $r_2=d^+_x-\ell-1$, then taking the difference, we obtain 
\begin{align*}
\sup\limits_{t\leq \tau_n}\bigl|\frac{\widetilde{S}_{x,\theta,\ell}^{(n)}(t)}{n}-\mu_x q_x(\theta)b\left(d_x^+,(1-\alpha_x)(1-e^{-t}),\ell\right)\bigr|\top 0, \quad  \text{as} \quad n\to \infty.
\end{align*}

Let $\widetilde{W}_n(t)$ and $\widetilde{D}_n(t)$ denote respectively the number of white balls and the total number of defaults at time $t$ in the percolated random graph. Then, by following the same steps as in Section~\ref{sec:thmCentrality}, we obtain
\begin{align*}
\sup\limits_{t\leq \tau_n}\bigl|\frac{\widetilde{W}_n(t)}{n}-f^{(\balpha)}_W(e^{-t})\bigr|\top 0, \ \ \sup\limits_{t\leq \tau_n}\bigl|\frac{\widetilde{D}_n(t)}{n}-f_D^{(\balpha)}(e^{-t})\bigr|\top 0.
\end{align*} 

Let $\widetilde{\tau}^{\star}_n$ be the first time when $\widetilde{W}_n(\widetilde{\tau}^{\star}_n)=-1$. Then, similar to the proof of Lemma~\ref{tau} in Section~\ref{sec:lemStop}, we find that $\widetilde{\tau}^{\star}_n\to-\ln  \hz_{\balpha}$, where $\hz_{\balpha}:=\sup\bigl\{z\in[0,1]: f^{(\balpha)}_W(z)=0\bigr\}$ and 
$$f_{W}^{(\balpha)}(z):= \lambda z - \sum_{x\in \cX} \mu_x d_x^-\sum_{\theta=1}^{d_x^+}q_x(\theta) \beta\bigl(d_x^+,\alpha_x +(1-\alpha_x)z, d_x^+-\theta+1\bigr).$$

Next, by following the same steps as proof of Theorem~\ref{thm:LLN} in Section~\ref{sec:thmLLN}, we obtain 
\begin{align*}
\frac{S_{x,\theta,\ell}^{(n)}(\balpha_n)}{n} \top s^{(\balpha)}_{x, \theta, \ell}(\hz_{\balpha}), \quad 
\frac{D_n(\balpha_n)}{n} \top  f_D^{(\balpha)}(\hz_{\balpha}),
\end{align*}
which then implies (by definition) that system-wide wealth converges to
\begin{align*}
\frac{\Gamma_n^{\Diamond}(\balpha_n)}{n} \top f^{(\balpha)}_{\Diamond}(\hz_{\balpha}):=\bar{\Gamma}^{\Diamond}-\sum_{x\in \cX} \bar{L}^{\odot}_x f_D^{(\balpha)}(\hz_{\balpha})
 - \sum_{x\in \cX} \bar{L}^{\Diamond}_x \sum_{\theta=1}^{d_x^+}\sum_{\ell=1}^{\theta-1} \ell s^{(\balpha)}_{x, \theta, \ell} (\hz_{\balpha}),
 \end{align*}
and,
the total cost of interventions $\balpha_n$ for the planner converges to
 \begin{align*}
\frac{\Phi_n(\balpha_n)}{n} \top \phi(\hz_{\balpha}):= \sum_{x\in \cX} \mu_x  \alpha_x C_x \sum_{\ell=1}^{d_x^+} \ell b\left(d_x^+,1-\hz_{\balpha},\ell\right).
\end{align*}
This completes the proof of Theorem~\ref{planner-LLN}.

\section{Further Extensions}\label{app:Ext}

In this appendix we discuss further extensions for our baseline aggregation functions which can be used for measuring systemic risk, as discussed in Section~\ref{sec:sys}. We will suppress the dependence of parameters on the size of the network $n$, if it is clear from the context.  

Recall that $D_x(t)=\sum_{\theta} D_{x,\theta}(t)$ denotes the total number of defaulted institutions with type $x\in \cX$ at time $t$. We assume that the (system-wide) loss associated to an infected link (coming from a defaulted institution) leading to solvent institutions are independent random variables with distribution depending on the type of the host institution. Further, each defaulted institution trigger (another) system-wide loss which are again assumed to be independent random variables depending on the type of this defaulting institution.

We first define
$$I_{x}(t):=\sum_{\theta=1}^{d^+_x}\Bigl[(\theta-1)D_{x,\theta}(t)+\sum_{\ell=1}^{\theta-1}\ell S_{x,\theta,\ell}(t)\Bigr],$$
as the total number of infected links leading to solvent institutions with type $x\in \cX$ up to time $t$. Let $\Bigl\{L_{x,D}^{(i)}\Bigr\}_{i=1}^\infty$ are  i.i.d. positive bounded random variables with common distribution $F_{x,D}$, which has expectation $\bar{L}_{x,D}$ and variance $\Sigma^2_{x,D}$,
for all $x\in \cX$. Similarly, let $\Bigl\{L_{x,I}^{(i)}\Bigr\}_{i=1}^\infty$ are  i.i.d. positive bounded random variables with common distribution $F_{x,I}$, which has expectation $\bar{L}_{x,I}$ and variance $\Sigma^2_{x,I}$,
for all $x\in \cX$. 

We consider the following modifications for the (random) aggregation function $\Gamma^{\Diamond}_n$:
$$\Gamma_n^{\Diamond}(t) := \bar{\Gamma}^{\Diamond}_n-\sum_{x\in\cX}Y_{D_x}(t)-\sum_{x\in\cX}Y_{I_x}(t),$$
where 
$$Y_{D_x}(t):=\sum_{i=1}^{D_x(t)}L_{x,D}^{(i)} \quad \text{and} \quad Y_{I_x}(t):=
\sum_{i=1}^{I_{x}(t)}L_{x,I}^{(i)}.$$


We first state our main theorems of this section, regarding the central limit theorems for the (extended) aggregation function $\Gamma^{\Diamond}_n$.

Let us define the functions
$$f^{(n)}_{x,I}(z):=\sum_{\theta=1}^{d^+_x}\bigl[(\theta-1)\mu^{(n)}_xq^{(n)}_x(\theta)-\beta(d^+_x,z,d^+_x-\theta+2)\bigr],$$
$$f^{(n)}_{x,D}(z):=\mu^{(n)}_x\bigl(1-\sum_{\theta=1}^{d^+_x}q^{(n)}_x(\theta)\beta(d^+_x,z,d^+_x-\theta+1)\bigr),$$
$$f^{(n)}_{\Diamond}(z):=\bar{\Gamma}^{\Diamond}_n-\sum_{x\in\cX}\bar{L}_{x,D}f^{(n)}_{x,D}(z)-\sum_{x\in\cX}\bar{L}_{x,I}f^{(n)}_{x,I}(z)$$
and correspondingly
$$\hf^{(n)}_{x,I}(t):=f^{(n)}_{x,I}(e^{-t}), \quad \hf^{(n)}_{x,D}(t):=f^{(n)}_{x,D}(e^{-t}), \quad \hf^{(n)}_{\Diamond}(t):=f^{(n)}_{\Diamond}(e^{-t}).$$
When we write all the above functions without the up-script $(n)$, it means that we substitute all the $\mu^{(n)}_x, q^{(n)}_x(\theta)$ by their limit $\mu_x, q_x(\theta)$, e.g., 
$$f_{x,D}(z):=\mu_x\bigl(1-\sum_{\theta=1}^{d^+_x}q_x(\theta)\beta(d^+_x,z,d^+_x-\theta+1)\bigr).$$

\begin{theorem}\label{norm_risk_random}
 Let Assumption~\ref{cond:CLT} holds and $\tau_n\leq \tau^\star_n$ be a stopping time such that $\tau_n\top t^\star$ for some $t^\star>0$. Then jointly in $\cD\left[0,\infty\right)$, as $n\rightarrow\infty$,
 \begin{equation}\label{norm_Gamma_random}
   n^{-1/2}(\Gamma^{\Diamond}_n(t\wedge\tst)-n\hf^{(n)}_{\Diamond}(t\wedge\tst))\tod \cZ_{\Diamond}({t\wedge t^{\star}}),
 \end{equation}
 where $\cZ_{\Diamond}(t)$ is a Gaussian random variable with mean 0 and variance $\Psi(t)$ given by
 $$\Psi(t)=\sum_{x\in\cX}\bigl[\Psi_{x,D}(t)+\Psi_{x,I}(t)+2\Psi_{x,D,I}(t)\bigr],$$
 where
 \begin{align*}
 \Psi_{x,D,I}(t)=&\bar{L}_{x,D}\bar{L}_{x,I}\bigsigma_{x,D,I}(e^{-t}),\\
 \Psi_{x,D}(t)=&\hf_{x,D}(t)\Sigma^2_{x,D}+\bar{L}_{x,D}\bigsigma_{x,D,D}(e^{-t}),\\
 \Psi_{x,I}(t)=&\hf_{x,I}(t)\Sigma^2_{x,D}+\bar{L}_{x,I}\bigsigma_{x,I,I}(e^{-t}),
 \end{align*}
and  the forms of $\bigsigma_{x,D,D}(y)$, $\bigsigma_{x,I,I}(y)$ and $\bigsigma_{x,D,I}(y)$ are given by \eqref{sigma DDx}, \eqref{sigma VVx} and \eqref{sigma DVx}.
\end{theorem}

The above theorem shows the asymptotic normality of system wide wealth $\Gamma^{\Diamond}_n(t)$ at any fixed time $t>0$. Although we don't have the asymptotic normality for the whole process, but we could still conclude the following theorem at the final moment of contagion.

\begin{theorem}\label{norm_risk_random_final}
Let $t^\star=-\ln \hz$ and $\hp_n$ be the largest $z\in[0,1]$ such that $f_W^{(n)}(z)=0$. Under the above assumptions, as $n\rightarrow\infty$, the final (extended) aggregation functions satisfy:
\begin{itemize}
\item[(i)] If $\hz=0$ then asymptotically almost all institutions default during the cascade and
\begin{align*}
\frac{\Gamma_n^{\Diamond}(\tst)}{n} \top\bar{\Gamma}^{\Diamond} -\sum_{x\in \cX}[\mu_x \bar{L}_{x,D}+\sum_{\theta=1}^{d^+_x}(\theta-1)\mu_x q_x(\theta) \bar{L}_{x,I}].
\end{align*}
\item[(ii)] If $\hz\in\left(0,1\right]$ and $\hz$ is a stable solution, i.e. $\alpha:=f_W'(\hz)>0$, then we have the following asymptotic normality for the final aggregate function:
    \begin{equation*}
            n^{-1/2}(\Gamma^{\Diamond}_n(\tst)-nf^{(n)}_{\Diamond}(\hp_n))
           \tod    \cZ_{\Diamond}({t^{\star}})-\alpha^{-1}f'_{\Diamond}(\hz)Z_{W}(t^{\star}).
       \end{equation*}
\end{itemize}
\end{theorem}

\subsection{Auxiliary Lemmas}

We first provide (under some regularity conditions) a central limit theorem  for functions which could be written as $Y_n(t):=\sum_{i=1}^{\lfloor X_n(t)\rfloor}G_i$, where $X_n(t)$ is a non-decreasing stochastic process satisfying $X_n(t)=O(n)$ for all $t>0$ and $\bigl\{G_i\bigr\}_{i\geq 1}$ are i.i.d. positive bounded random variables with mean $g$ and variance $\sigma^2$. 

\begin{lemma}\label{limitcondistri}
  Using the notations above and for fixed $t>0$, if $X_n(t):=f_n(t)n+\cV n^{1/2}$ with $(f_n(t))_{n=1}^{\infty}$ a real valued sequence converging to $f(t)$, and $\cV$ a bounded real-valued random variable, then as $n\rightarrow\infty$, conditioned on $\left\{\cV=x\right\}$ for some $x$ on $\mbox{supp}(\cV)$, we have
  $$\Bigl(\frac{Y_n(t)-gX_n(t)}{\sqrt{nf(t)}\sigma}\mid \cV=x \Bigr)\tod\cN(0,1).$$
\end{lemma}

\begin{proof}
Conditioned on the event $\left\{\cV=x\right\}$, $X_n(t)=f_n(t)n+xn^{1/2}$ which is non-random. Hence, by standard central limit theorem (CLT), we have
$$\Bigl(\frac{Y_n(t)-g\lfloor X_n(t)\rfloor}{\sqrt{\lfloor nf_n(t)+xn^{1/2}\rfloor}\sigma}|\cV=x\Bigr) \tod\cN(0,1).$$
Further, we have the decomposition
\begin{align*}
 \frac{Y_n(t)-gX_n(t)}{\sqrt{nf(t)}\sigma} & =\frac{\sqrt{\lfloor nf_n(t)+xn^{1/2}\rfloor}}{\sqrt{nf(t)}}\centerdot\frac{Y_n(t)-g\lfloor X_n(t)\rfloor}{\sqrt{\lfloor nf_n(t)+xn^{1/2}\rfloor}\sigma}+\frac{g\lfloor X_n(t)\rfloor-gX_n(t)}{\sqrt{nf(t)}\sigma}\\
   & = \sqrt{1+O(n^{-1/2})}\frac{Y_n(t)-g\lfloor X_n(t)\rfloor}{\sqrt{\lfloor nf(t)+xn^{1/2}\rfloor}\sigma}+O(n^{-1/2}).
\end{align*}
It follows thus by Slutsky's theorem that as $n\rightarrow\infty$,
$$\Bigl(\frac{Y_n(t)-gX_n(t)}{\sqrt{nf(t)}\sigma}|\cV=x\Bigr)\tod\cN(0,1).$$
\end{proof}

Using the above lemma, we prove the following proposition.

\begin{proposition}\label{finitedimen}
 Using the notations above and for fixed $t>0$, let $X_n(t):=f_n(t)n+\cV_n n^{1/2}$ with $\bigl\{f_n(t)\bigr\}_{n=1}^{\infty}$ a real-valued sequence converging to $f(t)$ and $\cV_n$ a sequence of random variables which converges to a Gaussian random variable $\cV\thicksim\cN(0,\upsilon^2)$ in distribution. Then we have, as $n\rightarrow\infty$,
$$\frac{Y_n(t)-ngf_n(t)}{\sqrt{n(f(t)\sigma^2+\upsilon^2g^2)}}\tod\cN(0,1).$$
\end{proposition}

\begin{proof}
We first consider the following integration
\begin{equation}\label{density Z}
 A(z):=\int_{-\infty}^{\infty}\frac{1}{2\pi\sqrt{f(t)}\sigma r}\exp\left\{-\frac{u^2}{2\upsilon^2}
-\frac{(z-gu)^2}{2f(t)\sigma^2}\right\}du.
\end{equation}
Let us denote by $a:=\upsilon^2g^2+f(t)\sigma^2$. Then by a change of variable $y=\frac{\sqrt{a}}{r\sigma\sqrt{f(t)}}u-\frac{rgz}{\sigma\sqrt{af(t)}}$, we obtain
\begin{align*}
  A(z) & = \frac{1}{2\pi\sqrt{f(t)}\sigma r}\int_{-\infty}^{\infty}\exp \left\{-\frac{1}{2f(t)\sigma^2\upsilon^2}(f(t)\sigma^2u^2+z^2\upsilon^2-2gzu\upsilon^2+\upsilon^2g^2u^2)\right\}du\\
   & = \frac{1}{2\pi\sqrt{f(t)}\sigma r}\int_{-\infty}^{\infty}\exp \left\{-\frac{1}{2f(t)\sigma^2\upsilon^2}((\sqrt{a}u-\frac{\upsilon^2gz}{\sqrt{a}})^2+\frac{f(t)\sigma^2\upsilon^2z^2}{a})\right\}du\\
   & = \frac{1}{2\pi\sqrt{f(t)}\sigma r}e^{-\frac{z^2}{2a}}
   \int_{-\infty}^{\infty}\frac{r\sigma\sqrt{f(t)}}{\sqrt{a}}e^{-\frac{y^2}{2}}dy\\
   & =\frac{1}{\sqrt{2\pi a}}e^{-\frac{z^2}{2a}}.
\end{align*}
On the other hand, define a function
$$h_z(x):=\frac{1}{\sqrt{2\pi f(t)}\sigma}\exp\left\{-\frac{(z-gx)^2}{2f(t)\sigma^2}\right\},$$
which is continuous and bounded. Thus by $\cV_n\tod\cV$, we have (as $n\to \infty$)
$$A_n(z):=\mathbb{E}[h_z(\cV_n)]\longrightarrow\mathbb{E}[h_z(\cV)]=A(z).$$
We denote
$$Z_n(t):=\frac{Y_n(t)-ngf_n(t)}{\sqrt{n}},$$
and let $Z(t)$ be a random variable with distribution
$$Z(t)\thicksim \cN(0,\sigma^2f(t)).$$
Let further $\mu_n$ be the probability measure of $\cV_n$ and $\mu$ be that of $\cV$. For convenience, we also denote by
$$\Phi_x(B):=\mathbb{P}(Z(t)-gx\in B),$$
and
$$G_{\cV_n}(B|x):=\mathbb{P}(Z_n(t)\in B|\cV_n=x).$$
Then for any Borel set $B\subset\mathbb{R}$, we have

\begin{equation*}\label{diff}
\begin{split}
\bigl|\mathbb{P}(Z(t)\in B)-\mathbb{P}(Z_n(t)\in B)\bigr| =& \bigl|\int_{\mathbb{R}}\Phi_x(B)d\mu(x)-\int_{\mathbb{R}}G_{\cV_n}(B|x)d\mu_n(x)\bigr|\\
\leq & 2\epsilon+\bigl|\int_{[-K,K]}G_{\cV_n}(B|x)d\mu_n(x)-\int_{[-K,K]}\Phi_x(B)d\mu_n(x)\bigr|\\
&+ \bigl|\int_{\mathbb{R}}\Phi_x(B)d\mu_n(x)-\int_{\mathbb{R}}\Phi_x(B)d\mu(x)\bigr|,
\end{split}
\end{equation*}
where we take $K$ large enough such that $\int_{\mathbb{R}\setminus[-K,K]}1d\mu_n(x)\leq\epsilon,$
uniformly on $n$. 

Next we check the right hand side of the above inequality term by term. For the second term, we have
$$\bigl|\int_{[-K,K]}G_{\cV_n}(B|x)d\mu_n(x)-\int_{[-K,K]}\Phi_x(B)d\mu_n(x)\bigr|\rightarrow 0.$$
Since any Borel set is a continuity set of Gaussian distribution, for every $x\in\mbox{supp}(\cV_n)\cap[-K,K]$, $G_{\cV_n}(B|x)\rightarrow\Phi_x(B)$ by Lemma~\ref{limitcondistri}. Then the result follows by the dominant convergence theorem.

For the third term, we have
\begin{align*}
 & \bigl|\int_{\mathbb{R}}\Phi_x(B)d\mu_n(x)-\int_{\mathbb{R}}\Phi_x(B)d\mu(x)\bigr|  
  \leq  \int_{B}|\mathbb{E}[h_z(\cV_n)]-\mathbb{E}[h_z(\cV)]|dz \\
  \leq& \int_{\mathbb{R}}|A_n(z)-A(z)|dz\rightarrow 0,
\end{align*}
where the first inequality follows by Fubini's theorem and the third line is by Scheff\'e's lemma since $\int_{\mathbb{R}}A_n(z)dz=\int_{\mathbb{R}}A(z)dz=1$ and $A_n(z)\rightarrow A(z)$ for every $z\in\mathbb{R}$.

Since we can choose $\epsilon$ is arbitrarly, we finally get for any borel set $B\in\mathbb{R}$, $$\mathbb{P}(Z_n(t)\in B)\rightarrow\int_{B}A(z)dz.$$ 
On the other hand, since $A(z)$ is the density of $\cN(0,a)$ and all Borel sets are continuity set of $\cN(0,a)$, it follows that $Z_n(t)\tod\cN(0,a)$, which is equivalent to the statement of proposition. The proof is complete.
\end{proof}

As proved above, under some proper assumptions, $Y_n(t)$ converges to a Gaussian random variable. We assume now a vector ${\bf Y}_n:=(Y^{(1)}_n(t),Y^{(2)}_n(t),\ldots,Y^{(m)}_n(t))$,  such that every component $Y^{(i)}_n(t)$ and every $X^{(i)}_n(t)$ satisfy the conditions in Proposition~\ref{finitedimen}. In addition, for any couple $i\neq j$, $G^{(i)}$ and $G^{(j)}$ are independent, and $\cV^{(i)}_n, i=1,\ldots,m$, converge jointly to a centered Gaussian vector $(\cV^{(1)},\ldots,\cV^{(m)})$ with covariance $A_{i,j}:=\Cov(\cV^{(i)},\cV^{(j)})$, $i,j=1,\ldots,m$.
Let
$$Z^{(i)}_n(t):=\frac{Y^{(i)}_n(t)-ng_if_i(t)}{\sqrt{n}},$$
where $Y^{(i)}_n(t):=\sum_{j=1}^{\lfloor X^{(i)}_n(t)\rfloor}G^{(i)}_j$, and let $Z_i(t)$ be the limit of $Z^{(i)}_n(t)$. On the other hand, conditioned on $(\cV^{(1)},\ldots,\cV^{(m)})$, $(Z_1(t),\ldots,Z_m(t))$ is an independent Gaussian vector, and by assumption $(\cV^{(1)},\ldots,\cV^{(m)})$ itself is also Gaussian. We observe that the formula inside the integration \eqref{density Z} is actually the joint density of $(Z(t),\cV)$and they are indeed jointly Gaussian. Using the properties of Gaussian vector, one can verify that $(Z_1(t),\cV^{(1)},\ldots,Z_m(t),\cV^{(m)})$ is jointly Gaussian. Hence $(Z_1(t),\ldots,Z_m(t))$ is also jointly Gaussian.

Moreover, we have the following covariance
\begin{equation*}
  \Cov(Z_i(t),Z_j(t))=A_{i,j}g_ig_j.
\end{equation*}
Indeed, by calculating $\Cov(Z^{(i)}_n(t),Z^{(j)}_n(t))$, we have (as $n\to \infty$)
\begin{align*}
  \Cov(Z^{(i)}_n(t),Z^{(j)}_n(t)) & =\mathbb{E}[\mathbb{E}[Z^{(i)}_n(t)Z^{(j)}_n(t)|\cV^{(i)}_n,\cV^{(j)}_n]] \\
   & =\mathbb{E}[n^{-1}\lfloor n^{-1/2}\cV^{(i)}_n\rfloor \lfloor n^{-1/2}\cV^{(j)}_n\rfloor g_i g_j] \\
   & =g_ig_j\Cov(\cV^{(i)}_n,\cV^{(j)}_n)+O(n^{-1})\longrightarrow A_{i,j}g_ig_j.
\end{align*}

\subsection{Proof of Theorem~\ref{norm_risk_random}}

For a fixed $0<t\leq\tst$, the processes $D_x(t\wedge\tst)$, $I_x(t\wedge\tst)$, $x\in\cX$, satisfy the conditions for $X_n(t)$ in Proposition~\ref{finitedimen}. Indeed, by Lemma~\ref{tau} and using a similar argument as in the proof of Theorem~\ref{normality} for $H^+_n(t)$, $S_n(t),\ldots$, we have for all $x\in\cX$, as $n\rightarrow\infty$, jointly
$$n^{-1/2}(D_x(t\wedge\tst)-n\hf^{(n)}_{x,D}(t\wedge\tst))\tod \cZ_{D_x}(t\wedge t^{\star}),$$
$$n^{-1/2}(I_x(t\wedge\tst)-n\hf^{(n)}_{x,I}(t\wedge\tst))\tod \cZ_{I_x}(t\wedge t^{\star}),$$
where $\cZ_{D_x}(t)$, $\cZ_{I_x}(t)$ are jointly Gaussian with mean 0 and covariance $$\bigsigma_{x,\clubsuit,\spadesuit}(e^{-t}):=\Cov\bigl(\cZ_{\clubsuit_x}(t),\cZ_{\spadesuit_x}(t)\bigr)$$
will be given later for $\clubsuit,\spadesuit\in\{D,I\}$. Note that for two different characteristics $x_1\neq x_2$, the covariance is 0.

In addition, Lemma~\ref{lem-remak} implies that $\hf^{(n)}_{x,D}(t)\rightarrow\hf_{x,D}(t)$ and $\hf^{(n)}_{x,I}(t)\rightarrow\hf_{x,I}(t)$ uniformly on $\interval$, for all $x\in\cX$. Combining with the continuity of $\hf_{x,D}(t)$ and of $\hf_{x,I}(t)$, it follows that $\hf^{(n)}_{x,D}(t\wedge\tst)\top\hf_{x,D}(t\wedge t^{\star})$ and $\hf^{(n)}_{x,I}(t\wedge\tst)\top\hf_{x,I}(t\wedge t^{\star})$, as $n\to \infty$. The Lemma~\ref{limitcondistri} still applies.
Further, for all $x\in \cX$  and $\clubsuit\in \{D,I\}$, $\clubsuit_x(t\wedge\tst)$ satisfies the conditions for $X_n(t)$ in Proposition~\ref{finitedimen}, with $f_n(t\wedge\tst)=\hf^{(n)}_{x,\clubsuit}(t\wedge\tst)$, $f(t)=\hf_{x,\clubsuit}(t\wedge t^{\star})$ and $\cV=\cZ_{\clubsuit_x}(t\wedge t^{\star})$.

For convenience, we let
$$\Delta^{(n)}_{D_x}(t):=n^{-1/2}\bigl(Y_{D_x}(t)-n\bar{L}_{x,D}\hf^{(n)}_{x,D}(t)\bigr) \ \text{and} \ \Delta^{(n)}_{I_x}(t):=n^{-1/2}\bigl(Y_{I_x}(t)-n\bar{L}_{x,I}\hf^{(n)}_{x,I}(t)\bigr).$$
Thus by Proposition~\ref{finitedimen}, we have that for all $x\in\cX$, $\clubsuit\in\{D,I\}$ and a fixed $t>0$, as $n\rightarrow\infty$, the following joint convergence holds
$$\Delta^{(n)}_{\clubsuit_x}(t\wedge\tst)\tod\cZ_{\clubsuit_x}(t\wedge t^{\star}),$$
where $\cZ_{\clubsuit_x}(t)$ is Gaussian with mean 0 and variance
$$\Psi_{x,\clubsuit}(t)=\hf_{x, \clubsuit}(t)\Sigma^2_{x,\clubsuit}+\bar{L}_{x,\clubsuit}\bigsigma_{x,\clubsuit,\clubsuit}(e^{-t}).$$
Moreover, the covariance is given by 
$$\Psi_{x,D,I}(t)=\bar{L}_{x,D}\bar{L}_{x,I}\bigsigma_{x,D,I}(e^{-t}).$$

We next consider the convergence of the following infinite sum 
$$\sum_{x\in\cX}\Delta^{(n)}_{D_x}(t\wedge\tst)+\sum_{x\in\cX}\Delta^{(n)}_{I_x}(t\wedge\tst).$$
Recall that  $L_{x,D}$ and $L_{x,I}$ are assumed to be bounded for all $x\in\cX$. Then there exists a constant $C$ such that $L_{x,D}+L_{x,I}<C$, for all  $x\in\cX$. Thus we have
\begin{align*}
  & \bigl|\sum_{x\in\cX^+_s}\bigl[\Delta^{(n)}_{D_x}(t\wedge\tst)+\Delta^{(n)}_{I_x}(t\wedge\tst)\bigr]\bigr| 
  \leq  C\bigl|\sum_{x\in\cX^+_s}(d^+_x+d^-_x)\sum_{\theta=1}^{d^+_x}V^{*}_{x,\theta,\pi_x(\theta)}(t\wedge\tst)\bigr|\\
 & \quad \quad  +\sum_{x\in\cX^+_x}n^{-1/2}\bigl[(Y_{D_x}(t\wedge\tst)-\bar{L}_{x,D}D_x(t\wedge\tst))+(Y_{I_x}(t\wedge\tst)-\bar{L}_{x,I}I_x(t\wedge\tst))\bigr].
\end{align*}
As shown in Section~\ref{sec:normality}, the first term converge to 0 in probability as $s\rightarrow\infty$ uniformly for $n$. For the second term, we can control its variance by
$C\sum_{x\in\cX^+_x}(d^+_x+d^-_x)\mu^{(n)}_x$. 

Hence, again by Assumption~\ref{cond:CLT}, the above formula converges to 0 as $s\rightarrow\infty$ uniformly for $n$. We therefore have also that the second term converges to 0 in probability as $s\rightarrow\infty$. Thus again we can take the limit under an infinite sum. It follows that
$$\sum_{x\in\cX}[\Delta^{(n)}_{D_x}(t\wedge\tst)+\Delta^{(n)}_{I_x}(t\wedge\tst)]\tod \sum_{x\in\cX}[\cZ_{D_x}(t\wedge t^{\star})+\cZ_{I_x}(t\wedge t^{\star})].$$
The left hand side is exactly 
$$-n^{-1/2}(\Gamma^{\Diamond}_n(t)-n\hf^{(n)}_{\Diamond}(t)),$$
so we obtain
$$\cZ_{\Diamond}(t\wedge t^{\star}):=-\sum_{x\in\cX}\bigl[\cZ_{D_x}(t)+\cZ_{I_x}(t)\bigr],$$
which is a centered Gaussian random variable with variance
$$\Psi(t)=\sum_{x\in\cX}\bigl[\Psi_{x,D}(t)+\Psi_{x,I}(t)+2\Psi_{x,D,I}(t)\bigr].$$

We now calculate the covariance $\bigsigma_{x,\clubsuit,\spadesuit}(t)$, for all $x\in\cX$ and $\clubsuit,\spadesuit\in\{D,I\}$. Recall the definitions of $V^{(n)}_{x,\theta,s}$ and $\cZ_{x,\theta,s}$ in section~\ref{sec:normality}. We observe that
$$I_{x}(t)=\sum_{\theta=1}^{d^+_x}\bigl[(\theta-1)n\mu^{(n)}_xq^{(n)}_x(\theta)-\sum_{s=\pip+1}^{d^+_x}V^{(n)}_{x,\theta,s}(t)\bigr],$$
and,
$$D_x(t)=\sum_{\theta=1}^{d^+_x}[n\mu^{(n)}_xq^{(n)}_x(\theta)-V^{(n)}_{x,\theta,\pi_x(\theta)}(t)].$$

Hence, by the same arguments as in section~\ref{sec:normality}, we obtain that
$$\cZ_{D_x}:=-\sum_{\theta=1}^{d^+_x}\cZ_{x,\theta,\pip} \quad \text{and} \quad \cZ_{I_x}:=-\sum_{\theta=1}^{d^+_x}\sum_{s=\pip+1}^{d^+_x}\cZ_{x,\theta,s}.$$
The covariances are finally given by the followings:
\begin{equation}\label{sigma DDx}
  \bigsigma_{x,D,D}(y)=\sum_{\theta=1}^{d^+_x}\bigsigma_{x,\theta,\pip,\pip}(y)
  +\sum_{\theta_1,\theta_2=1}^{d^+_x}\bigsigma^{\ast}_{x,\theta_1,\theta_2,\pi_x(\theta_1),\pi_x(\theta_2)}
  (y),
\end{equation}
\begin{equation}\label{sigma VVx}
  \bigsigma_{x,I,I}(y)=\sum_{\theta=1}^{d^+_x}\sum_{r,s=\pip+1}^{d^+_x}\bigsigma_{x,\theta,r,s}(y)
  +\sum_{\theta_1,\theta_2=1}^{d^+_x}
 \sum_{s_1=\pi_x(\theta_1)+1}\sum_{s_1=\pi_x(\theta_1)+1}
 \bigsigma^{\ast}_{x,\theta_1,\theta_2,s_1,s_2}(y),
\end{equation}
and, 
\begin{equation}\label{sigma DVx}
  \bigsigma_{x,D,I}(y)=\sum_{\theta_1,\theta_2=1}^{d^+_x}
 \sum_{s=\pi_x(\theta_1)}
 \bigsigma^{\ast}_{x,\theta_1,\theta_2,s,\pi_x(\theta_2)}(y)
  +\sum_{\theta=1}^{d^+_x}\sum_{s=\pip+1}^{d^+_x}\bigsigma_{x,\theta,\pip,s}(y).
\end{equation}
This completes the proof of Theorem~\ref{norm_risk_random}.

\subsection{Proof of Theorem~\ref{norm_risk_random_final}}

First note that, by using a similar argument as in the proof of Theorem~\ref{thm-agg-LLN}, since all the random losses $L_{x,D}, L_{x,I}, x\in\cX$ are bounded, we have 
\begin{align*}
\sup\limits_{t\leq \tau_n}\bigl|\frac{\Gamma_n^{\Diamond}(t)}{n}-f_{\Diamond}(e^{-t})\bigr|\top 0.
\end{align*}
If $\hz=0$, by Lemma~\ref{tau}, $\tst\top\infty$. Then $e^{-\tst}\top 0$, and we have
 $$f_{\Diamond}(0)=\bar{\Gamma}^{\Diamond}-\sum_{x\in \cX}\bigl[\mu_x \bar{L}_{x,D}+\sum_{\theta=1}^{d^+_x}(\theta-1)\mu_x q_x(\theta) \bar{L}_{x,I}\bigr].$$
Thus it follows by the continuity of $f_{\Diamond}$ that
$$f_{\Diamond}(e^{-\tst})=\bar{\Gamma}^{\Diamond}-\sum_{x\in \cX}[\mu_x \bar{L}_{x,D}+\sum_{\theta=1}^{d^+_x}(\theta-1)\mu_x q_x(\theta) \bar{L}_{x,I}]+o_p(1).$$
We therefore have
\begin{align*}
\frac{\Gamma_n^{\Diamond}(\tst)}{n} \top\bar{\Gamma}^{\Diamond} -\sum_{x\in \cX}[\mu_x \bar{L}_{x,D}+\sum_{\theta=1}^{d^+_x}(\theta-1)\mu_x q_x(\theta) \bar{L}_{x,I}].
\end{align*}

We now consider the case $\hz\in\left(0,1\right]$ and  $\alpha:=f_W'(\hz)>0$.  Note that the variance $\Psi(t)$ of $\cZ_{\Diamond}(t)$ is continuous on $t$. Moreover, since $\cZ_{\Diamond}(t)$ is a centered Gaussian random variable, its distribution is determined by $\Psi(t)$. Thus for a sequence $\{t_n\}_n$ which converges to $t$, we can show that, as $n\rightarrow\infty$,
\begin{equation}\label{conver on t}
  \cZ_{\Diamond}(t_n)\tod\cZ_{\Diamond}(t).
\end{equation}

Then we can use again the Skorokhod representation theorem, which shows that one can change
the probability space where all the random variables are well defined and all the convergence results of Theorem~\ref{thm-CLT-S}, Lemma~\ref{tau} ($\tst\rightarrow t^{\star}$)  and \eqref{conver on t} hold almost surely. Taking $t=\tst$ and $t_0=t^{\star}$, we obtain by Theorem~\ref{normality} that 
$$n^{-1/2}\Gamma^{\Diamond}_n(\tst) = n^{1/2}\hf^{(n)}_{\Diamond}(\tst)+\cZ_{\Diamond}(\tst\wedge t^{\star})+o_p(1).$$
Since $\tst\wedge t^{\star}\rightarrow t^{\star}$ a.s., we obtain that a.s. $\cZ_{\Diamond}(t_n)\rightarrow\cZ_{\Diamond}(t)$. It then follows that
$$n^{-1/2}S_n(\tst) = n^{1/2}\hf^{(n)}_S(\tst)+Z_{S}(\tst\wedge t^{\star})+o_p(1).$$
Then, similarly as in section~\ref{sec:normalityFinal}, it follows that
\begin{equation*}
  n^{-1/2}\Gamma^{\Diamond}_n(\tst)=n^{1/2}f^{(n)}_{\Diamond}(\hp_n)-\frac{f'_{\Diamond}(\hz)}
  {\alpha}Z_{W}(t^{\star})+Z_{\Diamond}(t^{\star})+o_p(1).
\end{equation*}
This completes the proof of Theorem~\ref{norm_risk_random_final}.

\end{document}